\documentclass[a4paper, 11pt]{article}
\usepackage{fullpage}
\RequirePackage[round]{natbib}
\RequirePackage[colorlinks,citecolor=blue,urlcolor=blue, linkcolor=blue]{hyperref}
\RequirePackage{graphicx}

\usepackage{comment}
\usepackage{lineno}
\usepackage[enable]{easy-todo}
\usepackage{amsmath, amsthm, amssymb, mathtools, bm, bbm, mathrsfs, bigints}
\usepackage{blkarray}
\usepackage{authblk}
\usepackage{chapterbib}

\usepackage{rotating} 
\usepackage{tikz} 
\usetikzlibrary{arrows, decorations.pathreplacing, positioning, shapes,arrows.meta}

\usepackage[capitalise]{cleveref}
\crefname{assumption}{Assumption}{Assumptions}
\usepackage{booktabs, multirow}
\usepackage{enumerate}
\usepackage{algorithm,algpseudocode}

\usepackage{mathtools}

\definecolor{blue-violet}{rgb}{0.54, 0.17, 0.89}
\definecolor{antiquefuchsia}{rgb}{0.57, 0.36, 0.51}
\definecolor{amethyst}{rgb}{0.6, 0.4, 0.8}
\definecolor{ao}{rgb}{0.0, 0.5, 0.0}
\definecolor{blue(ncs)}{rgb}{0.0, 0.53, 0.74}
\definecolor{dgreen}{rgb}{0.12, 0.3, 0.17}
\definecolor{cadmiumgreen}{rgb}{0.0, 0.42, 0.24}
\definecolor{darkolivegreen}{rgb}{0.33, 0.42, 0.18}
\definecolor{dartmouthgreen}{rgb}{0.05, 0.5, 0.06}

\newcommand{\tild}{\raise.17ex\hbox{ $\scriptstyle\sim$ }}

\newcommand{\floor}[1]{\lfloor #1 \rfloor}
\newcommand{\ceil}[1]{\lceil #1 \rceil}

\newcommand{\T}{\top}

\DeclareMathOperator*{\TV}{\mathrm{d}_\text{TV}}
\newcommand{\iid}{\overset{\textsf{iid}}{\sim}}
\newcommand{\indep}{\mathrel{\text{\scalebox{1.07}{$\perp\mkern-10mu\perp$}}}}

\DeclareMathOperator{\var}{\mathrm{var}}

\DeclareMathOperator{\sd}{\mathrm{sd}}

\DeclareMathOperator{\cov}{\mathrm{cov}}

\DeclareMathOperator{\E}{\mathbb{E}}
\DeclareMathOperator{\pr}{\mathbb{P}}

\DeclareMathOperator*{\argmax}{arg\,max}

\newcommand*\dd{\mathop{}\!\mathrm{d}}

\newcommand{\I}{\mathbb{I}}

\newcommand{\R}{\mathbb{R}}
\DeclareMathOperator{\Id}{Id}

\newcommand{\JJ}{J}

\newcommand{\unif}{\mathrm{unif}}
\newcommand{\expo}{\mathsf{Exp}}
\newcommand{\ber}{\mathsf{Ber}}
\newcommand{\gam}{\mathsf{Ga}}

\newcommand{\N}{\mathcal{N}}
\newcommand{\bin}{\mathrm{binom}}

\newcommand{\pow}{\text{pow}}

\theoremstyle{plain}
\newtheorem{theorem}{Theorem}
\newtheorem{corollary}{Corollary}
\newtheorem{proposition}{Proposition}
\newtheorem{lemma}{Lemma}

\newtheorem{assumption}{Assumption}
\newtheorem{condition}{Condition}
\theoremstyle{definition}
\newtheorem{definition}{Definition}
\newtheorem{example}{Example}
\theoremstyle{remark}
\newtheorem{remark}{Remark}

\usepackage{chngcntr}
\usepackage{apptools}
\AtAppendix{\counterwithin{theorem}{section}}
\AtAppendix{\counterwithin{figure}{section}}
\AtAppendix{\counterwithin{table}{section}}
\AtAppendix{\counterwithin{lemma}{section}}
\AtAppendix{\counterwithin{remark}{section}}
\AtAppendix{\counterwithin{corollary}{section}}
\AtAppendix{\counterwithin{proposition}{section}}
\AtAppendix{\counterwithin{assumption}{section}}
\AtAppendix{\counterwithin{example}{section}}

\makeatletter
\renewenvironment{proof}[1][\relax]{\par
  \pushQED{\qed}\normalfont \topsep6\p@\@plus6\p@\relax
  \trivlist
  \item[\hskip\labelsep\itshape
    \ifx#1\relax \proofname\else\proofname{} #1\fi\@addpunct{.}]\ignorespaces
}{\popQED\endtrivlist\@endpefalse
}
\makeatother

\let \hat \widehat
\let \tilde \widetilde
\let \epsilon \varepsilon

\newcommand{\leqcx}{\leq_{\mathrm{cx}}}

\newcommand{\g}{\mathcal{G}}

\title{Rank-transformed subsampling:\\ Inference for multiple data splitting and exchangeable p-values}

\date{\today}

\author[1]{F. Richard Guo \thanks{\texttt{ricguo@umich.edu}}}
\author[2]{Rajen D. Shah \thanks{\texttt{r.shah@statslab.cam.ac.uk}}}
\affil[1]{Department of Statistics, University of Michigan, Ann Arbor, USA}
\affil[2]{Statistical Laboratory, University of Cambridge, Cambridge, UK}

\begin{document}
\maketitle

\begin{abstract} 
Many testing problems are readily amenable to randomised tests such as those employing data splitting. However despite their usefulness in principle, randomised tests have obvious drawbacks. Firstly, two analyses of the same dataset may lead to different results. Secondly, the test typically loses power because it does not fully utilise the entire sample. As a remedy to these drawbacks, we study how to combine the test statistics or p-values resulting from multiple random realisations such as through random data splits. We develop rank-transformed subsampling as a general method for delivering large sample inference about the combined statistic or p-value under mild assumptions. We apply our methodology to a wide range of problems, including testing unimodality in high-dimensional data, testing goodness-of-fit of parametric quantile regression models, testing no direct effect in a sequentially randomised trial and calibrating cross-fit double machine learning confidence intervals. In contrast to existing p-value aggregation schemes that can be highly conservative, our method enjoys type-I error control that asymptotically approaches the nominal level. Moreover, compared to using the ordinary subsampling, we show that our rank transform can remove the first-order bias in approximating the null under alternatives and greatly improve power.
\end{abstract}


\section{Introduction} \label{sec:intro}
Many modern statistical procedures are randomised in the sense that the output is a random function of the data. A prominent class of such procedures are hypothesis testing methods that involve splitting the dataset into independent parts \citep{moran1973dividing,cox1975note}. These procedures randomly divide the data into two non-overlapping subsets, A and B, and then perform two steps which can be described as ``hunt and test'': first, sample A is used to choose one from among a collection of test statistics; next the chosen statistic is applied to sample B to produce the final test statistic. This approach is attractive because in the first stage an arbitrarily complicated procedure may be employed to \emph{hunt} for an appropriate test. Clearly, were we to ignore the fact that our test statistic was selected from data and simply apply it to sample A, we would fail to control the type-I error,
a phenomenon sometimes referred to as ``double dipping'' or ``data snooping''. 
However, because the data in A and B are independent, we can, in the second stage, effectively forget that the test statistic was chosen in a data-driven way, which permits straightforward calibration even when using a complicated ``hunting'' procedure. This strategy is particularly useful in settings with complex alternatives, as the test may be chosen to target the particular alternative from which the data appear to have arisen. This approach has been used for a variety of problems, such as testing the location of multiple samples \citep{cox1975note}, constructing conformal prediction intervals \citep{lei2018distribution, solari2022multi}, goodness-of-fit testing \citep{jankova2020goodness}, conditional (mean) independence testing \citep{scheidegger2021weighted,lundborg2022projected},
and conducting inference that is agnostic to the asymptotic regime \citep{kim2020dimension}, to list just a few. As we show in this work (see \cref{sec:unimodal}), a hunt and test approach can also be used to test for a clustering structure, i.e., for testing the null of unimodality, in high-dimensional data.
	
Another use of data splitting is related to nonparametric or semiparametric methods where the estimator for the parameter of interest depends on nuisance parameters that must also be estimated. To ensure proper asymptotic behaviour of the final estimator, the bias from nuisance parameter estimation needs to be controlled,
and this may be achieved by employing a form of sample splitting known as cross-fitting. Here, the data are first split into folds (i.e., parts) of roughly equal size, and then estimators are computed on each fold using nuisance parameters estimated from out-of-fold data. The per-fold estimates are then combined to produce the final estimate.
The independence afforded by data splitting permits the use of flexible machine learning methods to estimate these parameters, as adopted by targeted estimation or double/debiased machine learning methods. This second use of sample splitting has recently become popular in practice, although the idea has a long history; see also \citet{newey2018cross,chernozhukov2018dml} and references therein.
	
Randomised procedures also arise naturally in settings where a null hypothesis one wishes to test may more easily be stated in terms of a reweighted distribution. \citet{thams2023statistical} show how a wide range of problems may be cast in this framework,
including testing properties of a new policy in a contextual bandits setting, model selection after covariate shift and testing so-called generalised conditional independencies (also known as dormant independencies, or Verma constraints; see Section~\ref{sec:verma} for further discussion). In this work we focus on the latter, which after appropriate reweighting may be reduced to simpler independence testing problems. \citet{thams2023statistical} further propose to use resampling or rejection sampling to select from the original set of observations, a random subset that behaves like a sample from the reweighted distribution, to which an off-the-shelf testing procedure may then be applied. 

Despite its simplicity and broad applicability, as pointed out by \citet{cox1975note}, randomised procedures have obvious drawbacks. Firstly, the extra randomness hinders replication of the analysis, an issue that is particularly concerning in view of today's ``replication crisis'' in many scientific disciplines \citep{ioannidis2005most,open2015estimating,baker2016lid}. Although one may insist that the random seed used in a randomised algorithm should be part of the replication, when the substantive result of an analysis hinges on a particular seed, considering this as replicable is questionable.
Moreover, although it is sometimes argued that replicability is less of an issue when the sample size is large, as we show in \cref{ex:gauss-loc} below, it remains a problem when the effect size is moderate.

A second major issue in the context of hypothesis testing is that when sampling or data-splitting is used in the construction of a test, we may expect a loss of power due to the sample not being fully utilised; see, e.g., \citet[Theorem 2.6]{kim2020dimension} for a concrete case.
In the context of debiased machine learning, cross-fitting can sometimes be applied to alleviate the two drawbacks, to a certain extent. Ideally, the per-fold estimates are asymptotically independent and jointly Gaussian, thus giving an approximately Gaussian final averaged statistic. However, in finite samples or less ideal settings, these estimates are correlated and can  result in under-coverage of standard confidence intervals; we discuss this further in \cref{sec:dml}.

To illustrate the two main drawbacks mentioned above, consider the following toy example. 

\begin{example}[Gaussian location experiment] \label{ex:gauss-loc}
Let $T^{(1)}, T^{(2)}, \dots $ be test statistics resulting from repeatedly applying a randomised (e.g., data splitting) procedure to a given dataset.
Because $T^{(1)}, T^{(2)}, \dots$ are iid conditional on the data, \emph{unconditionally} they are exchangeable
in the sense that the joint distribution is invariant under any finite permutation \citep[\S1.1]{kallenberg2005probabilistic}. Suppose the statistic is
constructed such that the location of $T^{(1)}$ captures the signal. In particular, suppose $T^{(1)} \sim \N(\mu, 1)$ marginally and we are interested in testing $H_0:\mu = 0$ against $H_1: \mu > 0$.
The corresponding $\alpha$-level test is to reject $H_0$ when $T^{(1)} > z_{\alpha} := \Phi^{-1}(1-\alpha)$, which we refer to as the single-split test. 
We compare this to a test that aggregates the values from $L$ realizations (e.g., random splits of the data). In this case, it is natural to aggregate by taking the average $\bar{T}_{1:L} := (T^{(1)} + \dots + T^{(L)})/L$. The aggregated test rejects $H_0$ when $\bar{T}_{1:L} > \bar{t}_{\alpha}$, where the critical value $\bar{t}_{\alpha}$ is determined by the null distribution of $\bar{T}_{1:L}$.

In order to permit closed form expressions for the distribution of $\bar{T}_{1:L}$ and hence $\bar{t}_{\alpha}$, let us assume $T^{(1)}, T^{(2)}, \dots$ are jointly normal, i.e., follow a Gaussian process. One can show that the likelihood ratio test against any fixed $\mu > 0$ is monotone in $\bar{T}_{1:L}$. Due to the exchangeability, the distribution is parametrised by $(\mu, \rho)$, where $\rho \in [0,1)$ is the pairwise correlation between $T^{(i)}$ and $T^{(j)}$ for every $i \neq j$. In practice, we expect $\rho > 0$. 
	
	 Let $\phi_{T^{(l)}} := \I\{T^{(l)} > z_{\alpha} \}$ be the test corresponding to the $l$-th randomised test statistic. For each $\mu$ and $\rho$, we may compute the probability
	 that the test cannot be replicated by another application of the same procedure on the same dataset: $\pr(\phi_{T^{(1)}} \neq \phi_{T^{(2)}})$. Similarly, the probability that the aggregated test does not replicate is $\pr(\phi_{\bar{T}_{1:L}} \neq \phi_{\bar{T}_{(L+1):2L}})$, where $\phi_{\bar{T}_{1:L}} := \I\{\bar{T}_{1:L} > \bar{t}_{\alpha} \}$. 
	 The two probabilities (see \cref{app:gauss-loc} for their expressions) are compared in \cref{fig:gauss-loc}. Note that the probability of non-replication for $\phi_{T^{(1)}}$ can be quite high when the effect size (relative to the sample size for constructing $\phi_{T^{(1)}}$) is neither too weak nor too strong. The probability is significantly reduced by aggregating $L=200$ realizations and approaches zero as $L$ tends to infinity.
	 
	 Moreover, the aggregation also boosts the power as evidenced by \cref{fig:gauss-loc}. One can show that the power of the aggregated test is given by
	 \[ \E \phi_{\bar{T}_{1:L}} = \Phi\left( \frac{\mu}{\sqrt{1/L + \rho(L-1)/L}} - z_{\alpha} \right),\]
	 where $L=1$ corresponds to the power of the single-split test. For a small $\mu$ and a large $L$, we have
	 \[ \E \phi_{\bar{T}^{(1)}} \approx \alpha + \phi(z_{\alpha}) \mu, \quad \E \phi_{\bar{T}_{1:L}} \approx \alpha + \phi(z_{\alpha}) \mu / \sqrt{\rho}, \]
	 so the slope of local power is improved by a factor of $1 / \sqrt{\rho}$. This improvement can be particularly significant when $\rho$ takes a small positive value, which is not uncommon for the settings considered in the paper such as statistics resulting from two random splits of data. 

\begin{figure}[!htb]
\centering
\includegraphics[width=1.\textwidth]{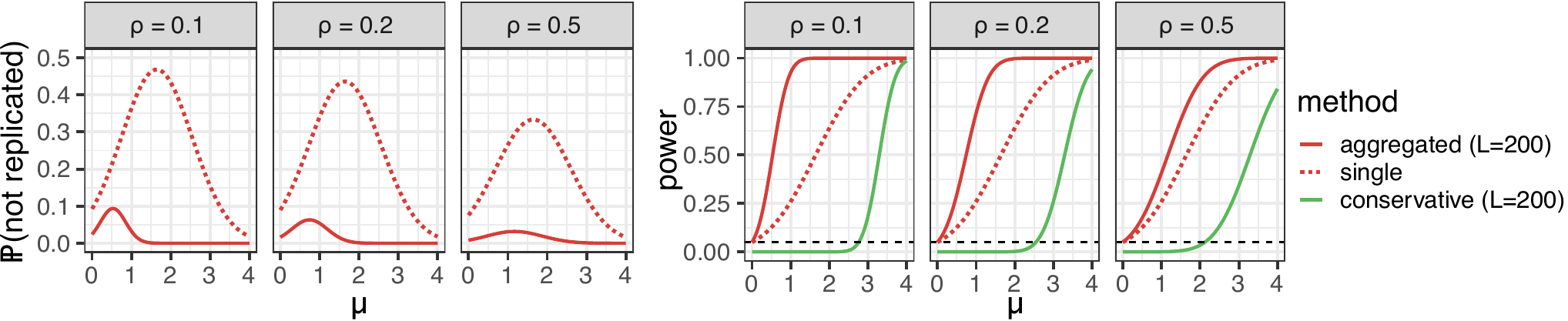}
\caption{Probability of non-replication (left) and power (right) for the Gaussian location experiment in \cref{ex:gauss-loc} ($\alpha = 0.05$, $L = 200$). The conservative aggregation rule rejects $H_0$ when $\bar{T}_{1:L} > 2 z_{\alpha}$ (see \cref{app:aggregation}).} 
\label{fig:gauss-loc}
\end{figure} 
\end{example}
From the example above, one may be tempted to conclude that to perform the aggregated test, we only need to estimate the correlation $\rho$ (e.g., through the empirical variance computed from $T^{(1)}, T^{(2)}, \dots$). 
However, the joint normality assumption on the statistics $T^{(1)}, T^{(2)}, \dots$, which further implies the normality of $\bar{T}_{1:L}$, need not hold in practice. In other words, even if every $T^{(l)}$ is marginally normal or asymptotically normal, the \emph{dependence} among the statistics need not be a normal copula even in large samples. Consider the following example due to \citet{kim2020dimension}.

\begin{example} \label{ex:kim-ramdas}
Let $X_1, \dots, X_n$ be iid random vectors in $\mathbb{R}^p$ with mean $\mu$ and covariance $\Sigma$. We are interested in testing $H_0: \mu = \bm{0}$ against $H_1: \mu \neq \bm{0}$.
Let $I_n$ be a random subset of $\{1,\dots,n\}$ of size $n_1$ and let $n_2 := n - n_1$. Let $\hat{\mu}_1 := n_1^{-1} \sum_{i \in I_n} X_i$ and $\hat{\mu}_2 := n_2^{-1} \sum_{i \in I_n^c} X_i$. We expect
$ \hat{\mu}_1^\top  \hat{\mu}_2  > 0$
 if $\mu$ deviates from zero.
Treating $\hat{\mu}_1$ as a fixed vector we see that $ \hat{\mu}_1^\top  \hat{\mu}_2 = n_2^{-1}\sum_{i \in I_n^c} X_i^\top \hat{\mu}_1$ is simply a sample mean computed from $I_n^c$. 
By studentising this quantity, we may obtain via a central limit theorem (under appropriate conditions), a test statistic
\[
T_n(X_1, \dots, X_n; I_n) := \frac{\sqrt{n_2} \hat{\mu}_1^\top  \hat{\mu}_2}{{\sqrt{\hat{\mu}_1^{\T} \hat{\Sigma}_2 \hat{\mu}_1}}} \rightarrow_{d} \N(0, 1) \quad \text{under $H_0$},
\]
as $n \rightarrow \infty$ and $n_2 / n \rightarrow q \in (0,1)$, where $\hat{\Sigma}_2$ is the empirical covariance computed from $I_n^c$. We reject $H_0$ when $T_n$ is large compared to $\N(0,1)$. One benefit of using such a test statistic is that unlike for example a norm of the empirical mean from the whole sample, the limit distribution does not depend delicately on the asymptotic limit of $p/n$. 

As indicated earlier, however, a disadvantage is that the approach does not fully utilise the information in the sample. One might hope that this can be alleviated by considering the cross-fitted statistic $[T_n(X_1, \dots, X_n; I_n) + T_n(X_1, \dots, X_n; I_n^c)]/2$.
However, \citet[Proposition A.1]{kim2020dimension} showed that this does not have a normal limit, rendering calibration challenging.
Further, let $I_n^{(l)}$ for $l=1,\dots,L$ be independent random subsets of size $\floor{q n}$ and define
$T_n^{(l)} :=  T_n(X_1, \dots, X_n; I_n^{(l)})$ and  the aggregated statistic $S_n:=\sum_{l=1}^L T_n^{(l)}/L$.
We see in \cref{fig:sampling-KR} that its sampling distribution under $X_1, \dots, X_n \iid \N(\mathbf{0}, \Sigma)$ is clearly non-normal even in large sample. 
\begin{figure}[!htb]
\centering
\includegraphics[width=0.8\textwidth]{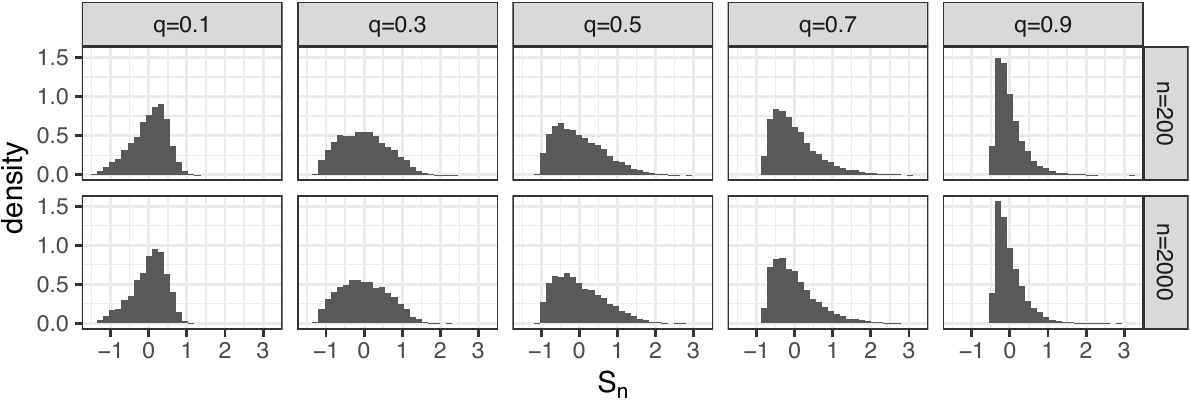}
\caption{Non-normal null distribution of the aggregated statistic $S_n$ in \cref{ex:kim-ramdas} under $X_1, \dots, X_n \iid \N(\mathbf{0}, \Sigma)$ with $L=200$, $p=3$ and $\Sigma_{ij} = 2^{-|i-j|}$. } 
\label{fig:sampling-KR}
\end{figure}
\end{example}

Using an aggregate $S_n$ of the exchangeable test statistics $T_n^{(1)},  \dots, T_n^{(L)}$ is more sensitive to departures from the null as it makes better use of the full data. In addition, because the conditional variability in $S_n$ given the data decreases in $L$ (e.g., like $1/L$ when $S_n$ is the average), by taking a relatively large $L$, $S_n$ is effectively derandomised; we will revisit \cref{ex:kim-ramdas} in \cref{sec:revisit,app:non-rep} to demonstrate the improvement from using $S_n$. However it should be clear from this very simple example that the main challenge in general lies in calibrating $S_n$. In particular, we need to handle the unknown, potentially complicated dependence among $T_n^{(1)}, \dots, T_n^{(L)}$, which typically causes $S_n$ to not follow any textbook distribution. 
To solve this problem, we develop a data-driven calibration scheme based on subsampling to obtain tests with size that under mild assumptions will approach the nominal level. Further, we will demonstrate that through the rank transform we introduce, our method is able to accurately approximate the null distribution even when data is drawn from a local alternative, and this leads to a power almost as good as an oracle procedure.
Our approach is applicable to all of the randomised tests mentioned above, and as we will see, by inverting particular aggregated hypothesis tests, we can also obtain confidence intervals for cross-fitted double machine learning that can reliably deliver coverage when the standard confidence intervals cannot. 
Before introducing our method, we briefly summarise some existing proposals in the literature.

\subsection{Existing proposals and related literature} \label{sec:lit}
There is a long line of work that has considered how to combine multiple test statistics. Most approaches either aim to control the type-I error under an arbitrary dependence of the test statistics, or make some specific assumptions about their dependence, such as a Gaussian copula.

Most approaches of the first type consider aggregating p-values, i.e., when $T^{(1)}$ is marginally (super-)uniformly distributed under the null.
Classical examples include Bonferroni correction, i.e., taking the minimum p-value and multiplying by the total number of p-values, and the average of the p-values multiplied by $2$ \citep{ruschendorf1982random,meng1994posterior}. Related results have been shown for quantiles \citep{meinshausen2009p,diciccio2020exact} and generalised means \citep{vovk2020combining}. Other aggregation rules include taking a weighted sum of Cauchy transformations \citep{liu2020cauchy}, those developed through concentration inequalities \citep{diciccio2020exact} and those involving converting p-values to the so-called e-values \citep{vovk2021values}, to name a few.

While these methods have the attractive guarantee of giving finite-sample valid p-values, as they necessarily must cater for the worst-case dependence, this benefit comes with the downside of conservativeness. Indeed, as pointed out by \citet{diciccio2020exact}, when used with p-values produced through sample-splitting, these methods can sometimes be outperformed by using a single-split test; see \cref{app:comp-conservative} for several such numerical examples from applications considered in this paper. For combining $Z$-statistics, the right panel of \cref{fig:gauss-loc} illustrates a similar phenomenon in the context of \cref{ex:gauss-loc} with the `conservative' rule that rejects when $\bar{T}_{1:L} > 2 z_{\alpha}$, where the extra factor of $2$ guarantees its validity; see \cref{thm:avg-Z} in the supplementary material. Indeed, as we show in \cref{prop:2-alpha}, directly comparing $\bar{T}_{1:L}$ to a standard normal can lead to a size of $2\alpha$, even when enforcing exchangeability of the underlying $Z$-statistics. This may come as a surprise given that by Jensen's inequality, $\bar{T}_{1:L}$ has variance at most $1$ for example. An analogous negative result for averaging p-values, which shows that the doubling rule cannot be improved under exchangeability,
is proved by \citet{choi2022avg}.

In view of this, other work has considered approaches relying on asymptotic joint normality of the underlying test statistics $T_n^{(1)}, T_n^{(2)}, \dots$; see, e.g., \citet[Theorem 4.1]{romano2019multiple}, \citet[Theorem 3.2]{diciccio2018hypothesis}, \citet{tian2021large} and \citet{liu2022multiple}. However, as evidenced by \cref{ex:kim-ramdas}, such a dependence structure among the test statistics is unlikely to hold in practice, especially for data splitting.

Our work also connects more broadly to a body of literature on subsampling; see \citet{politis1999subsampling} for a monograph on the topic. \citet{berg2010subsampling,mcmurry2012subsampling} consider subsampling for hypothesis testing and use the data to centre the estimated null distribution to improve power; this is in similar spirit to our rank transform introduced in \cref{sec:rank-transformed-sub}, which is a more aggressive form of centring that enforces the mean and all other moments of the null.
Subsampling has also been used to reduce the variance of an estimator through what is known as `subagging' \citep{buhlmann2002analyzing}. Stability selection procedures exploit this for variable selection \citep{meinshausen2010stability,shah2013variable} and can provide a form of finite-sample error control for a user-chosen variable selection method. More broadly, in the literature of resampling-based inference, ranks from bootstrap can be used to ``prepivot'' \citep{beran1987prepivoting,beran1988prepivoting} a statistic to improve error control under the null, though this is somewhat different from our use of ranks to improve power; this is discussed further in \cref{app:prepivot}.

A large number of randomised algorithms have been developed for various testing problems to which our method is applicable. Some of these have been mentioned in the introduction under the umbrellas of ``hunt and test'' and reweighting. Other methods that do not explicitly fall within these categories include approximate co-sufficient sampling \citep{barber2022testing} for goodness-of-fit testing and several approaches for assessing variable importance nonparametrically \citep{cai2022model,williamson2020unified,dai2022significance,tansey2022holdout}.

\subsection{Main contributions and organisation of the paper} \label{sec:contrib}
We develop a general framework for hypothesis testing with an aggregated statistic $S_n := S(T_n^{(1)},\dots, T_n^{(L)})$, where $S(\cdot)$ is a user-specified, symmetric aggregation function such as the arithmetic mean.
To fully exploit the signal in $S_n$, it is essential that the aggregated test is not conservative. Therefore, we construct a test from $S_n$ whose type-I error asymptotically approaches the nominal level $\alpha$, and which outperforms existing conservative aggregation rules. 

To achieve this, we employ subsampling \citep{politis1994large,politis1999subsampling}, a generic tool for approximating sampling distributions under minimal assumptions. 
However, using subsampling alone is not enough, because it not only approximates the sampling distribution under the null, but under alternatives  also picks up a visible, finite-sample bias from the sampling distribution (we formalise this in a first-order asymptotic analysis of subsampling in \cref{app:sub-first-order}). Hence, naively comparing $S_n$ to its subsampling critical value leads to a test with suboptimal power. To fix this crucial issue, we exploit the fact that the asymptotic null distribution of $T_n^{(1)}$ is typically known, e.g., $\unif(0,1)$ for a p-value or $\N(0,1)$ for a $Z$-statistic. We introduce a rank transform that we apply to the subsampled statistics to enforce the known null marginal distribution of $(T_n^{(1)}, \dots, T_n^{(L)})$. In other words, to approximate the null distribution of $(T_n^{(1)}, \dots, T_n^{(L)})$ and hence the aggregated $S_n$, we effectively combine the known \emph{marginal} null distribution with the unknown \emph{copula} estimated from subsampling. We demonstrate favourable performance of our method with
 three types of applications.
 \begin{enumerate}[(i)]
 	\item ``Hunt-and-test'' procedures: we use our framework to develop new tests for the goodness of fit of parametric quantile regression models, and for testing unimodality in high-dimensional data. We illustrate the latter on cancer gene expression data to detect the presence of cancer subtypes.
 	\item Testing hypotheses under reweighting or ``distributional shift'' with resampling or rejection sampling: specifically we study testing the sharp null of no direct treatment effect in the context of a sequentially randomised trial.
 	\item Calibrating confidence intervals for cross-fitted, double/debiased machine learning estimators: in our simulations we look in particular at confidence intervals for partially linear models, though the methodology we develop is applicable more broadly.
 \end{enumerate}
The rest of this paper is organised as follows. In \cref{sec:method},
we introduce our rank-transformed subsampling procedure. We present an aggregated, multiple-split test (\cref{alg:agg-test}) and a variant (\cref{alg:agg-test-adaptive}) that allows for several user-specified aggregation functions and adapts to the best one.
In \cref{sec:theory}, we study the theoretical properties of our method. We show that these algorithms
give tests that asymptotically have size equal to a given significance level $\alpha$.
Further, we show that if the test statistic and the aggregated statistic converge uniformly under the null, our procedures inherit such uniformity in terms of type-I error control. In terms of power, we show
under mild conditions that
our test is as powerful as an oracle test that has access to the null distribution of the aggregated statistic. Moreover, we show that the power gap between the oracle test and our test is smaller than the gap between the oracle test and a test based on ordinary subsampling (i.e., without the rank transform) by an asymptotic order, and this leads to a significant power improvement in practice.
We establish this result under general conditions that go beyond the settings where the Edgeworth expansion, the standard tool for higher-order asymptotic analysis in the literature, can be applied. 
In \cref{sec:application}, we demonstrate our method with a variety of applications as mentioned above. 
Finally, we conclude with a discussion in \cref{sec:discuss} outlining possible directions for future research. The supplementary material contains all proofs, additional theoretical and numerical results; all the appendices can be found there.  An R package \texttt{MultiSplit} implementing our method and scripts for reproducing numerical results are available from \url{https://github.com/richardkwo/MultiSplit}.
\section{Method} \label{sec:method}
\subsection{Setup} \label{sec:setup}
Let $X_1, \dots, X_n \in \mathcal{X}$ be data points drawn iid from an underlying distribution $P$. We are interested in testing 
\[ H_0: P \in \mathcal{P}_0 \quad \text{vs} \quad H_1: P \in \mathcal{P} \setminus \mathcal{P}_0, \]
where $\mathcal{P}$ is the set of relevant laws of $X$ that includes both the null and the alternative. 
Let $T_n^{(1)}, \dots, T_n^{(L)}$ be test statistics that can be computed from sample $X := (X_1, \dots, X_n)$ and a piece of external randomness $\Omega$ generated by the analyst. Throughout, we assume that the random vector
\begin{equation} \label{eqs:T-exch}
(T_n^{(1)}, \dots, T_n^{(L)}) \quad \text{is exchangeable}
\end{equation}
under $(X, \Omega) \sim P^n \times P_{\Omega}$ for every $P \in \mathcal{P}$, where $P_\Omega$ is the distribution of $\Omega$. 

Often, such $T_n^{(1)}, \dots, T_n^{(L)}$ are obtained by applying the same randomised procedure $L$ times on $X$. That is, we have
\begin{equation} \label{eqs:T-random}
T_n^{(l)} = T_n(X_1, \dots, X_n; \Omega^{(l)}), \quad l=1,\dots,L.
\end{equation}
where, without loss of generality, we assume $\Omega^{(l)} \iid \unif(0,1)$ independently from $X$. Formally, for $n=1,2,\dots$, the randomised test is a measurable map $T_n: \mathcal{X}^n \times [0,1] \rightarrow \mathbb{R}$.
For example, $\Omega^{(l)}$ can be used by $T_n$ to realise a random data split or a sequence of $\unif(0,1)$ random variables for acceptance-rejection sampling (e.g., by splitting the bits in a binary expansion). In this case, $\Omega = (\Omega^{(1)}, \dots, \Omega^{(L)})$. 

Alternatively, it can be  that every $T_n^{(l)}$ is a deterministic function of $X_1, \dots, X_n$. This can happen, for example, when there are $L$ pre-specified ways of splitting the full sample and every such way looks no different from any other way.  We will study this case in the context of cross-fitting in \cref{sec:dml}.

\medskip Throughout, we require that under the null, $T_n^{(1)}$ converges to a known, continuous distribution $F_0$, such as $\unif(0,1)$ or $\N(0,1)$. 
Without loss of generality, we assume $H_0$ is rejected for large values of $T_n^{(1)}$; other cases can be handled by redefining $T_n^{(1)}$, e.g., replacing $T_n^{(1)}$ with $|T_n^{(1)}|$ for a two-sided test, or with $1 - T_n^{(1)}$ for a p-value. 
To abuse the term slightly, we call
the test that rejects when $T_n^{(1)} > F_0^{-1}(1-\alpha)$
the ``single-split'' test, even though $T_n^{(1)}$ itself may not be constructed with data splitting. Consider the aggregated, ``multiple-split'' statistic
\[ S_n := S(T_n^{(1)}, \dots, T_n^{(L)}) \]
constructed with a symmetric, continuous aggregation function $S: \mathbb{R}^L \rightarrow \mathbb{R}$, such as the arithmetic mean or the maximum. By taking $L$ reasonably large, we can expect that the conditional variance of $S_n$ given $X_1, \dots, X_n$ is small enough such that the aggregated test statistic is effectively derandomised. Note the restriction that $S$ is symmetric is rather reasonable: it follows from the Neyman--Pearson lemma that a most powerful test (for a simple null against a simple alternative) necessarily combines the exchangeable test statistics in some symmetric fashion; see \cref{prop:most-powerful} in the supplementary material.

We will make the mild assumption (see \cref{assump:stable-G} and the following discussion) that $S_n$ converges to \emph{some} distribution $G_P$ under the null; in practice the limit $G_P$ is typically an unknown and often non-Gaussian continuous distribution that depends on $P \in \mathcal{P}_0$.
Our aggregated test rejects $H_0$ for large values of $S_n$, and under the null aims to mimic an oracle procedure that rejects whenever $S_n$ exceeds the unknown upper $\alpha$ quantile of $G_P$. To do this, we use subsampling to compute $\tilde{G}_{n}$, an approximation to $G_P$, and use its quantile to determine the critical values for $S_n$. As mentioned earlier in \cref{sec:contrib}, to have good power, however, $\tilde{G}_{n}$ must continue to closely mimic the null sampling distribution of $S_n$ even when data is generated under alternatives.
For example, when $T_n^{(1)} \sim \N(0,1)$ under $H_0$ and $S_n := (T_n^{(1)} + \dots + T_n^{(L)}) / L$, under alternatives, $\tilde{G}_{n}$, as expected from an oracle procedure, should maintain zero mean even when $S_n$ takes a positive mean. This is achieved by the rank transform introduced below.

\subsection{Rank-transformed subsampling} \label{sec:rank-transformed-sub}
In this section, we describe our procedure when using a single aggregation function $S$. We first introduce some notation relating to distribution functions and then set out our subsampling scheme.

\paragraph{Notation} Given a set of points $\{x_i\}$ on the real line, we use $\mathbb{F}_{\{x_i\}}$ to denote their empirical distribution function. For a real-valued function $F$, let $\|F\|_{\infty} := \sup_{x} |F(x)|$. For a distribution function $F$, its upper $\alpha$ quantile is defined as $F^{-1}(1-\alpha) := \inf\{x: F(x) \geq 1-\alpha\}$. 

\paragraph{Subsampling}
Our method is based on subsampling, which ensures type-I error control under minimal assumptions. In general, subsampling cannot be replaced by the bootstrap without sacrificing the wide applicability of our method; we explain this in \cref{app:bootstrap}. 
Let $m < n$ be a user-chosen subsample size. Throughout the paper, we require $m \rightarrow \infty$ and $m / n \rightarrow 0$; for the description of the algorithms and all the numerical experiments in this paper, we use $m = \floor{n / \log n}$ (see \cref{sec:discuss} for a discussion). We randomly select a total of $B$ sets of indices, each of size $m$, such that there is a sufficiently low degree of overlap among the sets. To do this, we first choose a positive integer (e.g., $\JJ=100$) and let $B := \JJ \floor{n / m}$. Then our collection of sets of indices $\mathcal{B} := \{(i_{1,b}, \dots, i_{m,b}): b=1,\dots,B\}$ is formed using \cref{alg:gen-tuple}.

\begin{algorithm}[!htb] \raggedright  \caption{Generate ordered tuples} \label{alg:gen-tuple}
	\textbf{Input}: Sample size $n$, subsample size $m$, positive integer $\JJ$. \\
	\vskip .3em
	\begin{algorithmic}[1]
		\State $\mathcal{B} \leftarrow \emptyset$. 
		\For {$j=1,\dots,\JJ$}
			\State $\pi \gets$ a random permutation of $\{1,\dots,n\}$.
			\State $\mathcal{B} \gets \mathcal{B} \cup \left\{(\pi_1, \dots, \pi_{m}),\, (\pi_{m+1}, \dots, \pi_{2m}),\, \dots, \, (\pi_{(\floor{n / m} - 1)m + 1}, \dots, \pi_{\floor{n / m} m}) \right\}$.
		\EndFor
		\State \Return $\mathcal{B}$
	\end{algorithmic}
\end{algorithm}
Note that the construction guarantees that $\mathcal{B}$ contains $\JJ$ collections of $\floor{n/m}$ sets that are non-overlapping, and so statistics computed on these subsamples are independent. This will allow us to obtain guarantees for subsampling that do not rely on approximating a scheme (e.g., \citealp[\S2.4]{politis1999subsampling}) where statistics on every possible subsample of size $m$ are evaluated.

Let $\hat{\mathbf{H}} = (\hat{H}_{b,l})$ be a $B \times L$ matrix consisting of rows 
\begin{equation} \label{eqs:H-mat}
\hat{\mathbf{H}}_{b,\cdot} := \left( T_m^{(1)}(X_{i_{1,b}}, \dots, X_{i_{m,b}}), \quad \dots \quad , T_m^{(L)}(X_{i_{1,b}}, \dots, X_{i_{m,b}}) \right), \quad b=1,\dots,B,
\end{equation}
i.e., by computing $(T_m^{(1)}, \dots, T_m^{(L)})$ on each subsample listed in $\mathcal{B}$.
When the statistic is a randomised test in the form of \cref{eqs:T-random}, the external random number is regenerated for every entry of $\hat{\mathbf{H}}$; that is, 
\[ \hat{H}_{b,l} = T_m(X_{i_{1,b}}, \dots, X_{i_{m,b}}; \Omega^{(b,l)}), \quad \Omega^{(b,l)} \iid \unif(0,1), \quad b=1,\dots,B,\: l=1,\dots,L. \]
Note that although we have arranged the subsampled test statistics into a matrix, entries in the same column but different rows do not correspond directly to one another, that is, $\hat{H}_{b,l}$ is no more related to $\hat{H}_{b',l}$ than $\hat{H}_{b',l'}$ for $b' \neq b$ and $l' \neq l$.

Now if we were to apply the aggregation function $S$ to each row of $\mathbf{H}$, we would obtain
\[ \hat{S}_{b} := S(\hat{H}_{b,1}, \dots, \hat{H}_{b,L}), \quad b=1,\dots,B, \]
whose empirical distribution function $\hat{G}_{n}(x) := \mathbb{F}_{\{\hat{S}_{b}\}}(x)$
is the natural subsampling estimate for $G_P(x)$. By the standard consistency result of subsampling (\citealp{politis1999subsampling}; see also \cref{sec:app:subsample}), we have $\|\hat{G}_{n} - G_{P}\|_{\infty} \rightarrow_p 0$ under $P \in \mathcal{P}_0$. However, directly using $\hat{G}_n$ to construct the test is suboptimal because under a sequence of local alternatives, $\hat{G}_n$ contains an upward bias from sampling under alternative. Although such a bias may vanish asymptotically, the rate at which it vanishes can be rather slow and this can severely reduce power; this is illustrated in the bottom left panel of \cref{fig:schematic-transform}. In the supplementary material, we formalise this point in \cref{thm:crit-ordinary} and numerically demonstrate the bias in \cref{app:power-numerical}. Therefore, instead of using $\hat{G}_n$ to calibrate our test statistic $S_n$, we perform the rank transform introduced below.

\paragraph{Rank transform} Using exchangeability, we can pool the entries of $\hat{\mathbf{H}}$ and let $\mathbb{F}_{\hat{\mathbf{H}}}$ be the resulting empirical distribution function.
With this, we form a rank-transformed version of $\hat{\mathbf{H}}$, denoted by $\tilde{\mathbf{H}} = (\tilde{H}_{b,l})$, filled with entries
\begin{equation} \label{eqs:rank}
\begin{split}
\tilde{H}_{b,l} &:= F_0^{-1}\left(\mathbb{F}_{\hat{\mathbf{H}}}(\hat{H}_{b,l}) - 1 / (2B\, L) \right) \\
&= F_0^{-1}\left(\frac{(\text{rank of $\hat{H}_{b,l}$ among entries in $\hat{\mathbf{H}}$}) - 1/2}{B\, L} \right),
\end{split}
\end{equation}
where the subtraction of $1/2$ from the ranks is simply a finite sample correction to prevent infinity being produced when applying $F_0^{-1}$.
We then compute the aggregated statistics
\[ \tilde{S}_b := S(\tilde{H}_{b,1}, \dots, \tilde{H}_{b,L}), \quad b=1,\dots,B \]
and their resulting empirical distribution function $\tilde{G}_{n} := \mathbb{F}_{\{\tilde{S}_b\}}(x)$, which we then use to determine the critical value $\tilde{G}_n^{-1}(1-\alpha)$ for $S_n$. The full procedure is given in~\cref{alg:agg-test}.

\begin{figure}[!htb]
\hspace{3em}
\includegraphics[width=.7\textwidth]{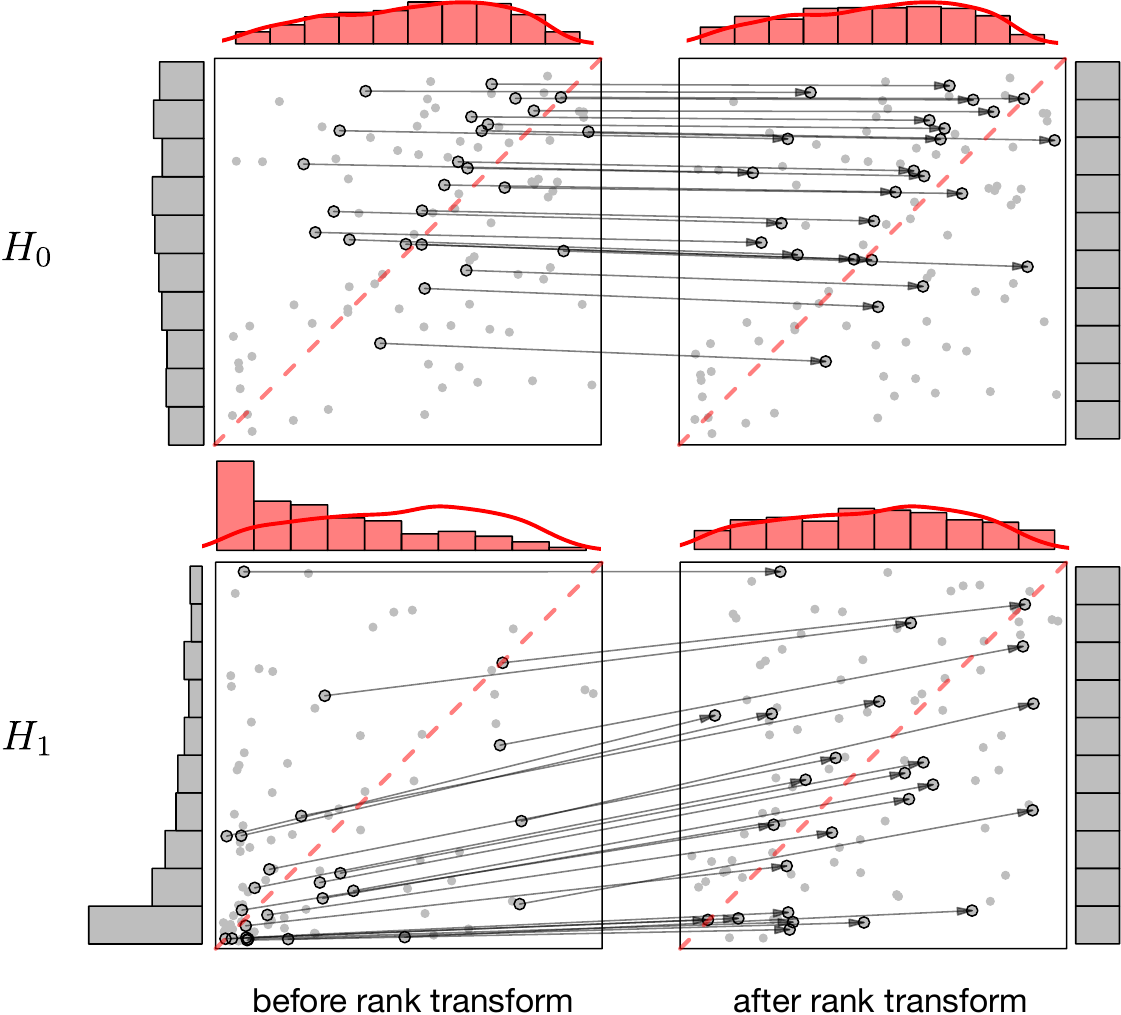}
\caption{Illustration of the rank transform. Here $L=2$ and the rows of $\hat{\mathbf{H}}$ and $\tilde{\mathbf{H}}$ are plotted as points in two dimensions in the left and right panels respectively. Arrows indicate for certain points their image after the rank transform. We consider aggregation function $S$ as the arithmetic mean which may be visualised as projection onto the red dashed line. Top (red) histograms: distributions $\hat{G}_n$ and $\tilde{G}_n$ (curve: null density of $S_n$); side (grey) histograms: marginal distributions $\mathbb{F}_{\hat{\mathbf{H}}}$ and $\mathbb{F}_{\tilde{\mathbf{H}}}$ ($F_0 = \unif(0,1)$).}
\label{fig:schematic-transform}
\end{figure}

We provide some intuition on why the rank transform works. 
Consider first the null case where $P \in \mathcal{P}_0$. As $n \rightarrow \infty$ (and hence $n/m \rightarrow \infty$, $B \rightarrow \infty$), by consistency of subsampling and exchangeability of $T_m^{(1)}, \dots, T_m^{(L)}$ (so they share the same marginal distribution), we expect $\mathbb{F}_{\hat{\mathbf{H}}}(\cdot) - 1 / (2BL) \approx F_0(\cdot)$ in \cref{eqs:rank}. Therefore, $\tilde{\mathbf{H}} \approx \hat{\mathbf{H}}$ under the null. Because $\tilde{G}_n$ is computed from $\tilde{\mathbf{H}}$ in the same way as $\hat{G}_n$ is computed from $\hat{\mathbf{H}}$, we can expect that $\tilde{G}_n \approx \hat{G}_n \approx G_P$ with high probability under the null. This is formalised in~\cref{thm:main-null} and illustrated in the top panel of \cref{fig:schematic-transform}, from which we see that under the null the rank-transform leaves the points almost unchanged, and both $\hat{G}_n$ and $\tilde{G}_n$ well-approximate the sampling distribution of $S_n$.

Under an alternative $P \in \mathcal{P} \setminus \mathcal{P}_0$, we expect that $\mathbb{F}_{\hat{\mathbf{H}}}(\cdot) - 1 / (2BL) \approx F_P(\cdot)$, with $F_P$ the distribution function of the test statistic $T_m^{(1)}$ corresponding to the subsample size $m$. Thus, from \cref{eqs:rank} we have
	\[\tilde{H}_{b,l} = F_0^{-1}\left(\mathbb{F}_{\hat{\mathbf{H}}}(\hat{H}_{b,l}) - 1 / (2 B L) \right) \approx F_0^{-1}\left(F_P(\hat{H}_{b,l})\right) \approx F_0^{-1}\left(U_{b,l}\right)\]
	for some $U_{b,l} \sim \unif(0, 1)$. In this way, the rank transform enforces the marginal distribution of $\tilde{H}_{b,l}$ to be $F_0$, the asymptotic null distribution of $\hat{H}_{b, l}$, as we can observe from the side histograms in \cref{fig:schematic-transform}.
	The dependency among the $\hat{H}_{b,1}, \dots, \hat{H}_{b,L}$ in contrast is left unchanged. However, particularly under local alternatives where $T_n^{(1)}$ contains just enough information to detect deviation from the null, we would certainly expect the \emph{dependency} among test statistics constructed from smaller subsampled data to be indistinguishable from that under the null; we will formalise this notion in \cref{def:copula-conv}. In sum then, $\tilde{H}_{b, 1}, \ldots, \tilde{H}_{b, L}$  should  
	continue approximating the null distribution of $T_n^{(1)}, \dots, T_n^{(L)}$, as we desire.  Further, such an approximation should be more accurate than directly using $\hat{H}_{b, 1}, \ldots, \hat{H}_{b, L}$, because any bias stemming from the difference in the marginal distribution has already been removed. Indeed, this underlies the effectiveness of the rank transform in restoring null-like behaviour under the alternative, as demonstrated in the bottom panel of \cref{fig:schematic-transform} and formalised by our theory in \cref{sec:power}.

\begin{algorithm}[!htb] \raggedright  \caption{Aggregated multiple-split test} \label{alg:agg-test}
	\textbf{Input}: Data $(X_1,\dots,X_n)$, exchangeable single-split test statistics $(T_n^{(1)}, \dots, T_n^{(L)})$, asymptotic null distribution function $F_0$, aggregation function $S$, significance level $\alpha \in (0,1)$, positive integer $\JJ$. \\
	\vskip .3em
	\begin{algorithmic}[1]
		\State $m \gets \floor{n / \log n}$, $B \gets \JJ \floor{n / m}$.
		\State Run \cref{alg:gen-tuple} to obtain $\mathcal{B} = \{(i_{1,b}, \dots, i_{m,b}): b=1,\dots B\}$.
		\State Initialise $B \times L$ matrices $\hat{\mathbf{H}}$, $\tilde{\mathbf{H}}$ and $B$-dimensional vector $\tilde{S}$. 
		\For {$b=1,\dots,B$}
			\State $\hat{\mathbf{H}}_{b,\cdot} \gets (T_m^{(1)}(X_{i_{1,b}}, \dots, X_{i_{m,b}}), \dots, T_m^{(L)}(X_{i_{1,b}}, \dots, X_{i_{m,b}}))$
		\EndFor
		\For {$b=1,\dots,B$}
			\For {$l=1,\dots,L$}
				\State $\tilde{H}_{b,l} \gets F_0^{-1}\left( (\{\text{rank of $\hat{H}_{b,l}$ in $\hat{\mathbf{H}}$}\} - 1/2) / B \, L \right)$
			\EndFor
			\State $\tilde{S}_b \gets S(\tilde{H}_{b,1}, \dots, \tilde{H}_{b,L})$
		\EndFor
		\State $\tilde{G}_{n} \gets \mathbb{F}_{\{\tilde{S}_b\}}$
		\State Compute $S_n \gets S(T_n^{(1)}, \dots, T_n^{(L)})$ from $X_1,\dots,X_n$.
		\State Reject $H_0$ if $S_n > \tilde{G}^{-1}_{n}(1-\alpha)$ and report p-value $1 - \tilde{G}_{n}(S_n)$.
	\end{algorithmic}
\end{algorithm}

\subsection{Adapting to the best aggregation function} \label{sec:adaptive}
We can further improve power by choosing a good aggregation function $S$.
The performance of an aggregation function, however, depends on the joint behaviour of the single-split tests under the alternative of interest, which is usually unknown.
For example, we expect $S = (T^{(1)}_n + \dots + T^{(L)}_n) / L$ to work particularly well if most of the single-split statistics are large under the alternative; in contrast, $S = \max(T^{(1)}_n, \dots, T^{(L)}_n)$ should perform better if only a few of them are large under the alternative. This motivates us to allow the user to specify multiple candidate aggregation functions $S^{1}, \dots, S^{W}$, which are expected to accommodate different cases. In \cref{alg:agg-test-adaptive}, we present a variant of our procedure that aims to adapt to the best aggregation function among $S^1, \dots, S^W$.

The algorithm rejects for large values of
\begin{equation} \label{eqs:R_n}
R_n := \max\left(\tilde{G}_{n}^1(S_n^{1}), \dots, \tilde{G}_{n}^W(S_n^{W}) \right),
\end{equation}
where $\tilde{G}_{n}^w$ and $S_n^w$ respectively are the counterparts of $\tilde{G}_n$ and $S_n$ in \cref{alg:agg-test} but relate to the $w$-th aggregation function. The quantity $R_n$ in \cref{eqs:R_n} is therefore one minus the minimum p-value corresponding to each of the aggregation functions. Thus if any one of the aggregation functions yields good power, we should expect $R_n$ to be large: in this way, the test statistic aims to achieve power close to that of the best $S^w$ under consideration.

We could calibrate $R_n$ using a Bonferroni correction, but this would give a conservative test potentially sacrificing any power we might have gained in using multiple aggregation functions. Instead, we can \emph{reuse} our subsampling aggregate statistics $\tilde{S}_b^w$ to approximate the sampling distribution of $R_n$ under the null; the subsampled versions of $R_n$ used for this are computed in lines 9--11 of \cref{alg:agg-test-adaptive}. The advantage of this approach is that it properly takes account of the dependence among $\tilde{G}_{n}^1(S_n^{1}), \dots, \tilde{G}_{n}^W(S_n^{W})$ involved in the construction of $R_n$. As a consequence, the resulting test has asymptotic size equal to its prescribed level (\cref{thm:main-null-adaptive}), and so in this sense no power is lost.

\begin{algorithm}[!htb] \raggedright  \caption{Aggregated multiple-split test that adapts to the best aggregation function} \label{alg:agg-test-adaptive}
	\textbf{Input}: As in \cref{alg:agg-test} but with the single aggregation function $S$ replaced by a collection $S^1, \dots, S^W$. \\
	\vskip .3em
	\begin{algorithmic}[1]
		\State Run up to line 10 of \cref{alg:agg-test} to obtain $\tilde{H}_{b,l}$ for $b=1,\ldots,B$ and $l=1,\ldots,L$.
		\State Initialise $B$-dimensional vectors $\tilde{R}$, $\tilde{S}^1, \dots, \tilde{S}^W$.
		\For {$w=1,\dots,W$}
		\For {$b=1,\dots,B$}
		\State $\tilde{S}_b^{w} \gets S^w(\tilde{H}_{b,1}, \dots, \tilde{H}_{b,L})$
		\EndFor
		\State $\tilde{G}_{n}^w \gets \mathbb{F}_{\{\tilde{S}^w_b:\, b=1,\dots,B\}}$
		\EndFor
		\For {$b=1,\dots,B$}
		\State $\tilde{R}_b \gets \max\left(\tilde{G}_{n}^1(\tilde{S}^1_b), \dots, \tilde{G}_{n}^W(\tilde{S}^W_b) \right)$
		\EndFor
		\State $\tilde{Q}_{n} \gets \mathbb{F}_{\{\tilde{R}_b\}}$
		\For {$w=1,\dots,W$}
		\State Compute $S_n^w \gets S^w(T_n^{(1)}, \dots, T_n^{(L)})$ from $X_1,\dots,X_n$.
		\EndFor
		\State $R_n \gets \max\left(\tilde{G}_{n}^1(S_n^{1}), \dots, \tilde{G}_{n}^W(S_n^{W}) \right)$
		\State Reject $H_0$ if $R_n > \tilde{Q}_{n}^{-1}(1-\alpha)$ and report p-value $1 - \tilde{Q}_{n}(R_n)$.
	\end{algorithmic}
\end{algorithm}

\section{Theory} \label{sec:theory}
\subsection{Behaviour under the null} \label{sec:theory-null}
In this section, we establish that our algorithms lead to asymptotically level $\alpha$ tests under a set of mild assumptions, which ensure the consistency of subsampling and the validity of rank transform. 
\begin{condition}[Asymptotic pivotal null] \label{cond:pivotal}
For every $P \in \mathcal{P}_0$, under $(X, \Omega) \sim P^n \times P_{\Omega}$, as $n \rightarrow \infty$ it holds that
\[ T_n^{(1)} \rightarrow_{d} F_0,\]
where $F_0$ is a known, continuous distribution function.
\end{condition}
Note that when $T_n^{(l)}$ is a randomised test of the form in \cref{eqs:T-random}, a sufficient condition for the above is that for every fixed $\omega \in (0,1)$ and every $P \in \mathcal{P}_0$, it holds that $T_n(X_1, \dots, X_n; \omega) \rightarrow_d F_0$.

We study the tests constructed from rank-transformed subsampling under the null for two leading cases, when (i) $F_0$ is $\unif(0,1)$ (one minus p-value), (ii) $F_0 = \Phi$ ($Z$-statistic).
For other null distributions, probability integral transform and its inverse can be applied to convert the statistic to one of these cases.

First, we show that \cref{alg:agg-test} is an asymptotic level $\alpha$ test under mild assumptions. Further, when $T_n^{(1)}$ and $S_n$ converge uniformly to their limiting distributions over the null, under a mild finite density condition, we show that the test also controls size below $\alpha$ \emph{uniformly} over the null. It can be argued that uniform asymptotic size control, as opposed to pointwise asymptotic size control, is more relevant to practice because in contrast to the latter, it ensures that the sample size required to control the actual type-I error below, say 0.051, does not depend on the underlying $P \in \mathcal{P}_0$; see
\citet[\S11.1]{lehmann2005testing}. As uniform size control involves consideration of the behaviour of random variables under different $P$ (rather than a single $P$ in pointwise asymptotics), in the below, we will use a subscript in $\pr_P(\cdot)$ to denote that $(X, \Omega) \sim P^n \times P_\Omega$.

Recall that \cref{alg:agg-test} rejects $H_0$ whenever $S_n$ exceeds the upper $\alpha$ quantile of the rank-transformed subsampling distribution $\tilde{G}_{n}$. Control of the size is therefore intimately linked to the behaviour of $\tilde{G}_{n}$, for which we will require the following. Note that in the below, all densities are with respect to the Lebesgue measure.

\begin{condition}[Lipschitz aggregation] \label{cond:lip}
The aggregation function $S$ is Lipschitz continuous in $\|\cdot\|_{\infty}$ with Lipschitz constant 1.
\end{condition}
\noindent Examples include $S_n = (T_n^{(1)} + \dots + T_n^{(L)}) / L$ and $S_n = \max(T_n^{(1)}, \dots, T_n^{(L)})$. Note that given any Lipschitz continuous $S$, by scaling, the Lipschitz constant $1$ above is not a restriction; however such a scaling affects the value of $g_{P,\max}$ in \cref{assump:stable-G} below, which we require to be finite.
\begin{assumption}[Stability of $S_n$] \label{assump:stable-G}
For every $P \in \mathcal{P}_0$, under $(X, \Omega) \sim P^n \times P_{\Omega}$, it holds that 
\[ S_n \rightarrow_d G_P, \]
where $G_P$ is a continuous distribution that can depend on $P$ and can be unknown. Further, $G_P$
has a density function $g_P$ such that $g_{P,\max} := \sup_x g_P(x) < \infty$. 
\end{assumption}

If $(T_n^{(1)}, \dots, T_n^{(L)})$ has a limiting joint distribution under every $P \in \mathcal{P}_0$, then the first part of \cref{assump:stable-G} holds by definition of $S_n$. Given that $(T_n^{(1)}, \dots, T_n^{(L)})$ is exchangeable with a limit marginal law (\cref{cond:pivotal}), we typically expect the joint distribution to be stable as well (see \cref{fig:sampling-KR} for an example). In fact, for a given $P$, the sequence $(S_n)_{n=1}^\infty$ is uniformly tight, and so by Prohorov's theorem \citep[Thm.~2.4]{van2000asymptotic}, there always exists a subsequence that converges in distribution; see \cref{prop:tight} in the supplementary material. Thus the only way the stability assumption can fail is when the copula  $(F_{n,P}(T_n^{(1)}), \dots, F_{n,P}(T_n^{(L)}))$ in some sense ``oscillates'' as $n \to \infty$, where $F_{n,P}$ is the distribution function of $T_n^{(1)}$.

Below we present results on both pointwise and uniform asymptotic size control for the test in 
\cref{alg:agg-test}.

\begin{theorem}[Validity of \cref{alg:agg-test}] \label{thm:main-null}
Let $(X,\Omega) \sim P^n \times P_{\Omega}$ for $P \in \mathcal{P}_0$.
Suppose $T_n^{(1)}, \dots, T_n^{(L)}$ are exchangeable and
\cref{cond:pivotal} holds with $F_0 = \unif(0,1)$ or $F_0 = \N(0,1)$.
Suppose $S$ is chosen such that \cref{cond:lip} holds.
Then for all $\alpha \in (0, 1)$, the following hold:
\begin{enumerate}[(i)]
	\item Under \cref{assump:stable-G}, the test in \cref{alg:agg-test} is pointwise asymptotically level $\alpha$.
	\item Suppose $T_n^{(1)}$ and $S_n$ converge to their respective limit distributions uniformly over the null, i.e., for every $x \in \R$
	\[ \sup_{P \in \mathcal{P}_0} |F_{n,P}(x) - F_0(x) | \rightarrow 0, \quad \sup_{P \in \mathcal{P}_0} |G_{n,P}(x) - G_P(x) | \rightarrow 0, \]
where $F_{n,P}$ and $G_{n,P}$ respectively denote the distribution function of $T_n^{(1)}$ and $S_n$.
Also, suppose $\sup_{P \in \mathcal{P}_0} g_{\max,P} < \infty$. Then, the test in \cref{alg:agg-test} is uniformly asymptotically level $\alpha$, i.e.,
	\[ \sup_{P \in \mathcal{P}_0} \left| \mathbb{P}_{P} \left\{S_n > \tilde{G}_{n}^{-1}(1-\alpha) \right\} - \alpha \right| \rightarrow 0. \]
\end{enumerate}
\end{theorem}

Now we establish similar results for the adaptive test in \cref{alg:agg-test-adaptive}, under the following joint stability assumption on the chosen aggregation functions $(S^{1}, \dots, S^{W})$.

\begin{assumption}[Joint stability of multiple aggregation functions] \label{assump:stable-joint}
For every $P \in \mathcal{P}_0$, under $(X, \Omega) \sim P^n \times P_{\Omega}$, it holds that 
\[ (S_n^{1}, \dots, S_n^{W}) \rightarrow_d (S^1, \dots, S^W), \]
where every $S^w$ has a continuous distribution function $G_P^w$ that can depend on $P$. Further, suppose $G_P^w$ permits a density $g_P^w$ such that $g_{P,\max}^{w} := \sup_x g_P^{w}(x) < \infty$ for $w=1,\dots,W$. \end{assumption}

\begin{theorem}[Validity of \cref{alg:agg-test-adaptive}] \label{thm:main-null-adaptive}
Let $(X,\Omega) \sim P^n \times P_{\Omega}$ for $P \in \mathcal{P}_0$.
Suppose $T_n^{(1)}, \dots, T_n^{(L)}$ are exchangeable and
\cref{cond:pivotal} holds with $F_0 = \unif(0,1)$ or $F_0 = \N(0,1)$.
Also, suppose \cref{cond:lip} holds for every $S^w$ ($w=1,\dots,W)$.
Then for every $\alpha \in (0, 1)$ the following hold:
\begin{enumerate}[(i)]
	\item Under \cref{assump:stable-joint}, the test in \cref{alg:agg-test-adaptive} has pointwise asymptotic level $\alpha$.
	\item Suppose $T_n^{(1)}, S_n^1, \dots, S_n^W$ and $\max\{G_P^1(S_n^1), \dots, G_P^W(S_n^W)\}$ converge to their respective limit distributions uniformly over the null, i.e., for every $x \in\R$ and each $w=1,\ldots,W$,
	\begin{align*}
		\sup_{P \in \mathcal{P}_0} |F_{n,P}(x) - F_0(x) | \rightarrow 0, \quad \sup_{P \in \mathcal{P}_0} |G^{w}_{n,P}(x) - G^{w}_P(x)| \rightarrow 0 , \quad \sup_{P \in \mathcal{P}_0} |Q_{n,P}(x) - Q_P(x) | \rightarrow 0,
	\end{align*}
	where $Q_{n,P}$ is the distribution function of $\max\{G_P^1(S_n^1), \dots, G_P^W(S_n^W)\}$, and $G^{w}_{n,P}$ is the distribution function of $S_n^w$.
	Also, suppose $\max_w \sup_{P \in \mathcal{P}_0} g^w_{\max,P} < \infty$.Then, the test in \cref{alg:agg-test-adaptive} is uniformly asymptotically level $\alpha$, i.e.,
	\[ \sup_{P \in \mathcal{P}_0} \left| \mathbb{P}_P \left\{R_n > \tilde{Q}_{n}^{-1}(1-\alpha) \right\}- \alpha\right| \to 0. \]
\end{enumerate}
\end{theorem}

\subsection{Power} \label{sec:power}
In this section, we study the power of rank-transformed subsampling and establish its advantage over ordinary subsampling, i.e., subsampling without the rank transform.
We will analyse power under a sequence of local alternatives that converge ``in copula'' to a null case --- such a null is typically also the limit that the sequence of local alternatives weakly converges to. For any sequence $P_n \in \mathcal{P}$, let $F_{m,P_n}$ be the distribution function of $T_m^{(1)}(X_1, \dots, X_m)$ under $(X, \Omega) \sim P_n^m \times P_{\Omega}$ and let $U_m$ denote the copula:
\begin{equation} \label{eqs:copula-m}
 U_m := (U_m^{(1)}, \dots, U_m^{(L)}) := \left(F_{m,P_n}(T_m^{(1)}), \dots, F_{m,P_n}(T_m^{(L)}) \right).
\end{equation}
\begin{definition}[Convergence in copula] \label{def:copula-conv}
Let $U_m$ be the copula of $(T_m^{(1)}, \dots, T_m^{(L)})$ under $P_n^m \times P_{\Omega}$ given by \cref{eqs:copula-m}. We say $P_n$ converges in copula to $P_0$, if there exists some $P_0 \in \mathcal{P}_0$ such that under $P_0^n \times P_{\Omega}$, $(T_n^{(1)}, \dots, T_n^{(L)})$ converges to a limit distribution with copula $C=(C_1, \dots, C_L)$ satisfying $U_m \rightarrow_{d} C$. 
\end{definition}
Convergence in copula is a rather weak notion of convergence for two reasons.  Firstly, it involves the lower sample size $m$ (recall $m = o(n)$) rather than $n$.  Consider a sequence of alternatives $P_n$ that are only just distinguishable from the null at sample size $n$. At sample size $m$, the null and $P_n$ should be indistinguishable; that is, the behaviour of the entire vector of test statistics $(T_m^{(1)}, \dots, T_m^{(L)})$  under $P_n^m \times P_{\Omega}$ and $P_0^m \times P_{\Omega}$ should be asymptotically identical, and in particular convergence in copula would hold.
Secondly, \cref{def:copula-conv} is completely insensitive to the marginal distribution of the test statistics, and so in fact we can even expect a stronger version of the convergence above to hold with $m$ replaced by $n$. In particular, when $P_n$ is a sequence of local alternatives that converges to $P_0 \in \mathcal{P}_0$ in a way such that under $P_0^n \times P_{\Omega}$, $(T_n^{(1)}, \dots, T_n^{(L)})$ and the log-likelihood ratio $\log (\dd P_n^n / \dd P_0^n)$ jointly converge to a normal limit, then by Le Cam's third lemma \citep[Example 6.7]{van2000asymptotic}, $(T_n^{(1)}, \dots, T_n^{(L)})$ under $P_n^n \times P_{\Omega}$ must also converge to a normal limit with the same covariance, and so the same copula, as its null limit. 

Our next result shows that under a sequence of local alternatives that converge in copula to a null $P_0 \in \mathcal{P}_0$, the test in \cref{alg:agg-test} asymptotically has the same critical value and hence achieves the same power as an oracle test that has access to the asymptotic null distribution of $S_n$ under $P_0$. In stating \cref{thm:stable-copula,thm:rank-asymp} below, we use $G_{P_0}$ to denote the limit null distribution function of $S_n := S(T_n^{(1)}, \dots, T_n^{(L)})$ under $P_0^n \times P_{\Omega}$.

\begin{theorem}[Zeroth-order behaviour of rank-transformed subsampling] \label{thm:stable-copula}
Let $(T_n^{(1)}, \dots, T_n^{(L)})$ be exchangeable. 
Suppose \cref{cond:pivotal} holds with $F_0 = \unif(0,1)$ or $F_0 = \N(0,1)$ and
$S$ is chosen to satisfy \cref{cond:lip}. 
Consider a sequence $P_n \in \mathcal{P}$ that converges in copula to some $P_0 \in \mathcal{P}_0$. Under \cref{assump:stable-G}, for any $\alpha \in (0,1)$, we have 
\[
\tilde{G}^{-1}_n(1-\alpha) \rightarrow_{p} G_{P_0}^{-1}(1-\alpha),
\]
where $\tilde{G}_n$ is computed using \cref{alg:agg-test}.
\end{theorem}

In fact, when $P_n$ is a sequence of local alternatives that converges to $P_0 \in \mathcal{P}_0$ such that $P_n^n$ is contiguous to $P_0^n$ (i.e., absolutely continuous asymptotically; see \citealp[Ch.~6]{van2000asymptotic}), we also expect $\hat{G}_n^{-1}(1-\alpha) \rightarrow_p G_{P_0}^{-1}(1-\alpha)$ for $\hat{G}_n$ obtained from ordinary subsampling \citep[Theorem 2.6.1]{politis1999subsampling}. Hence, to capture the power improvement from the rank transform, we need a finer analysis. To this end, we characterise the first-order asymptotic behaviour of rank-transformed subsampling in the next theorem, of which the full statement can be found in \cref{app:rank-first-order}. In the below, $\TV(X,Y)$ denotes the total variation distance between $X$'s distribution and $Y$'s distribution. The regularity condition would require $S$ to be non-decreasing in each coordinate, a condition that holds for $S = \text{avg}$, $S = \max$ and other reasonable choices. 

\begin{theorem}[First-order behaviour of rank-transformed subsampling] \label{thm:rank-asymp}
Suppose \cref{cond:pivotal} holds and $(T_n^{(1)}, \dots, T_n^{(L)})$ is exchangeable. 
Consider a sequence $P_n \in \mathcal{P}$ that converges in copula to some $P_0 \in \mathcal{P}_0$ in the sense of \cref{def:copula-conv} such that 
\begin{equation*} 
\TV(U_m, C) = o(\sqrt{m/n}).
\end{equation*}
Suppose the distribution of $C$ is absolutely continuous with respect to the Lebesgue measure. 
Let $\tilde{G}_n$ denote the rank-transformed subsampling distribution function obtained with a variant of \cref{alg:agg-test} that uses two independent copies of the data under $(X,X',\Omega) \sim P_n^n \times P_n^n \times P_{\Omega}$ (see \cref{app:rank-first-order}).

Suppose \cref{assump:stable-G} holds and fix $\alpha \in (0,1)$ such that the density $G'_{P_0}$  is strictly positive and continuous in a neighbourhood of $G_{P_0}^{-1}(1-\alpha)$.
Then, under regularity conditions posed on the copula and $S$ (see \cref{app:rank-first-order}), for any $M>0$, we have
\begin{equation*} 
\E \left[-M \vee \sqrt{n/m} \, \left(\tilde{G}_n^{-1}(1-\alpha) - G_{P_0}^{-1}(1-\alpha) \right) \wedge M \right] \rightarrow 0.
\end{equation*}

Further, let $G_{n,P_n}$ be the distribution function of $S(T_n^{(1)}, \dots, T_n^{(L)})$ under $P_n^n \times P_{\Omega}$. Suppose additionally that $\|G_{n,P_n} - G_{\text{alt}}\|_{\infty} = o(\sqrt{m/n})$ holds for some distribution function $G_{\text{alt}}$ that is differentiable at $G_{P_0}^{-1}(1-\alpha)$. Then, for any $M>0$, we also have 
\begin{equation*} 
\E \left[-M \vee \sqrt{n/m} \, \left(G_{n,P_n}(\tilde{G}_n^{-1}(1-\alpha)) - G_{n,P_n}(G_{P_0}^{-1}(1-\alpha)) \right) \wedge M \right] \rightarrow 0.
\end{equation*}
\end{theorem}

 To interpret this result, let us take $M$ to be a large constant and choose $m = \floor{n / \log n}$. Then, the first statement above says that up to the first order (with scaling factor $\sqrt{\log n}$), the rank-transformed subsampling delivers an approximation to the oracle critical value that is asymptotically unbiased. In contrast, in \cref{thm:crit-ordinary} (see the supplementary material) we show that the ordinary subsampling approximation to the oracle critical value is biased upwards, and typically the bias grows with the effect size of the alternative. This formalises our observation from \cref{fig:schematic-transform}: under $H_1$, before applying the rank transform, subsampling is biased towards the alternative sampling distribution and we can see a clear discrepancy between the subsampling distribution (red histogram) and the desired null distribution (red curve). Along a sequence of contiguous local alternatives, although this discrepancy vanishes asymptotically, this occurs rather slowly ($1 / \sqrt{\log n} \approx 1/4$ when $n=10^7$) and can result in a significant loss of power in practice. 

Recall that $G_{n,P_n}$ denotes the distribution function of $S_n $ under $P_n^n \times P_{\Omega}$. The power of the oracle test is $\pr(S_n > G_{P_0}^{-1}(1-\alpha)) = 1 - G_{n,P_n}(G_{P_0}^{-1}(1-\alpha))$. If we ignore the dependence between the estimated critical value and the test statistic, the power of our test can be written similarly as
\[ \pr(S_n > \tilde{G}_n^{-1}(1-\alpha)) \approx 1 -  \E G_{n,P_n}(\tilde{G}_n^{-1}(1-\alpha)). \]
Consequently, the second statement of \cref{thm:rank-asymp} implies 
\begin{equation} \label{eqs:pow-rank-oracle}
\pow(\text{rank-transformed subsampling}) \approx \pow(\text{oracle}) - o(1 / \sqrt{\log n}). 
\end{equation}
Meanwhile, with $G_{m,P_n}$ denoting the distribution function of $S_m$ under $P_n^m \times P_{\Omega}$, suppose $(n/m)^{\beta} (G_{m,P_n}^{-1}(1-\alpha) - G_{P_0}^{-1}(1-\alpha))$ converges to $\tau > 0$ that measures the effect size. Then, in contrast to the above, \cref{thm:crit-ordinary} implies 
\begin{equation} \label{eqs:pow-ordinary-oracle}
\pow(\text{ordinary subsampling}) \approx \pow(\text{oracle}) - \kappa_{\alpha,\tau}\, \tau \,/ (\log n)^{\beta}, \quad 0 < \beta \leq 1/2
\end{equation}
for some $\kappa_{\alpha,\tau} > 0$. Typically, we expect $\kappa_{\alpha,\tau}\, \tau$ to grow as $\tau$ increases from zero up to a certain value; see \cref{app:power-numerical} for a concrete example. 

For example, for testing a hypothesis of a regular parameter, under a $\sqrt{n}$ local alternative, we may expect $G_{n,P_n}^{-1}(1-\alpha) - G_{P_0}^{-1}(1-\alpha) \rightarrow \tau$ and hence $\sqrt{m/n}(G_{m,P_n}^{-1}(1-\alpha) - G_{P_0}^{-1}(1-\alpha)) \rightarrow \tau$. 
We illustrate such a case using numerical results for \cref{ex:kim-ramdas}, which considers testing the mean of a random vector. 
The details of the simulation study will be described in \cref{sec:revisit}. For now, let us focus on \cref{fig:power-KR}, which shows the power of several aggregated, multiple-split tests against the effect size of local alternatives. 
Indeed, we can see that the rank-transform (\cref{alg:agg-test}) has a clear power advantage over ordinary (`\texttt{no rank}') subsampling, and this advantage enlarges with the effect size, exactly as we expect from comparing \cref{eqs:pow-rank-oracle,eqs:pow-ordinary-oracle} under $\beta=1/2$. 
Meanwhile, in every setting, the power of the rank-transformed subsampling closely tracks that of the oracle test regardless of the effect size, confirming \cref{eqs:pow-rank-oracle}. 

We prove \cref{thm:rank-asymp} in \cref{app:power-1st-order} using the functional delta method, where a major technical challenge is a certain Hadamard differentiability we establish for handling the errors introduced by the rank transform. For technical reasons, \cref{thm:rank-asymp} is proved for a variant of \cref{alg:agg-test} that has access to two independent copies of the data, but we expect a similar result to hold for the original algorithm as well; see \cref{app:power-numerical} for a concrete example with supporting numerical results. Further, in \cref{app:check-assump}, we show that when $C$ follows a Gaussian copula, the regularity conditions in \cref{thm:rank-asymp} are satisfied by choices of $F_0$ and $S$ considered in this paper. 

\section{Applications} \label{sec:application}
We illustrate our method with three types of applications. First, we study data-splitting, hunt-and-test procedures: specifically, we revisit \cref{ex:kim-ramdas} for testing the zero mean of a high-dimensional random vector; we develop a new test for unimodality in high dimensions; and introduce a simple, flexible approach for goodness-of-fit testing of parametric regression models such as parametric quantile regression. Next, we consider
using the data from a distribution $P$ to test a property of a different distribution $Q$, where $Q$ is related to $P$ through reweighting. We study this in the context of causal inference, where $P$ is the observational distribution and $Q$ is an intervened distribution.
Finally, we study the inference of cross-fitted, double/debiased machine learning estimators. We show that the cross-fold dependence in these estimators, though often argued to be asymptotically negligible in standard well-specified, low-dimensional settings, can, in finite sample or under misspecification, lead to under-coverage of confidence intervals (see \citealp{jiang2022new} for a high-dimensional setting not considered in this paper where this issue also arises). We present an alternative construction of confidence intervals using our method that captures such dependence and restores the desired coverage. These confidence intervals can even maintain coverage when the model for a nuisance parameter in doubly robust estimation is misspecified, which we illustrate in \cref{app:dml}.

For each application, we present numerical results to illustrate and benchmark the new methods we develop. Additional numerical results, including the performance of conservative aggregation rules mentioned in \cref{sec:lit}, can be found in \cref{app:num-extra}.

\subsection{Hunt and test}
In this section, we consider testing a hypothesis that can be expressed as a conjunction of simpler hypotheses 
\begin{equation} \label{eqs:intersect-null}
 H_0 = \bigcap_{\delta \in \mathcal{D}} H_0(\delta),
\end{equation}
where we already have an off-the-shelf test for each $H_0(\delta)$.
As explained in the introduction, such null hypotheses are amenable to a hunt-and-test approach that employs data splitting, where one part of the data is used to find an appropriate $\hat{\delta}$ and the remaining data is used to test $H_0(\hat{\delta})$.
Here is another perspective due to \citet{moran1973dividing}. Consider testing 
\begin{equation} \label{eqs:composite-alt}
H_0: \theta \in \Theta_0 \quad \text{vs} \quad H_1: \theta \in \Theta_1,
\end{equation}
where $\Theta_1$ does not contain $\theta_0$. The alternative parameter space $\Theta_1$ might be so large or heterogeneous that a reasonable test for $H_0$ only has power against certain alternatives in $\Theta_1$. Again, we can split our data and perform hunt and test: use the first part to estimate $\hat{\theta}_1 \in \Theta_1$ and then use the second part to test $\theta \in \Theta_0$ versus $\theta = \hat{\theta}_1$. 

\subsubsection{Revisiting \cref{ex:kim-ramdas}} \label{sec:revisit}
In \cref{ex:kim-ramdas}, we considered a hunt-and-test approach for testing $H_0: \mu = \mathbf{0}$ versus $H_1: \mu \neq \mathbf{0}$ with iid random vectors.
Clearly, the problem is an instance of \cref{eqs:composite-alt}; it can also be viewed as an instance of \cref{eqs:intersect-null}, namely $H_0 = \cap_{\delta \in \mathbb{R}^p} \{\mu: \mu^{\top} \delta = 0\}$.
\cref{fig:power-KR} compares the performance of various tests based on aggregated $S_n = (T_n^{(1)} + \dots + T_n^{(L)}) / L$ with $L=200$ and the single-split test based on $T_n^{(1)}$ alone. In our simulation, we draw $X_1, \dots, X_n \iid \N(\mu, \Sigma)$ where $\Sigma \in \R^{3 \times 3}$ has entries given by $\Sigma_{ij} = 2^{-|i-j|}$ and $\mu = \tau n^{-1/2} v_1$, where $v_1$ is the normalised principal eigenvector of $\Sigma$. When $\tau=0$, we see that our method (\cref{alg:agg-test}) controls the type-I error at the nominal level. Further, as $\tau$ grows, its power clearly dominates both the single-split test and the ordinary subsampling test, while closely tracking the power of the oracle test in all regimes. These observations align with our theory presented in \cref{sec:power}. Also, note that the power of our test is insensitive to the split ratio $q$.

Meanwhile, our aggregated test significantly reduces the chance of non-replication. For example, when $q=0.5$, for a random dataset, there is less than 5\% chance that two applications of \cref{alg:agg-test} will give contradicting results, while the probability can be as large as 30\% for the single-split test; see \cref{app:non-rep}.

\begin{figure}[!htb]
\centering
\includegraphics[width=0.95\textwidth]{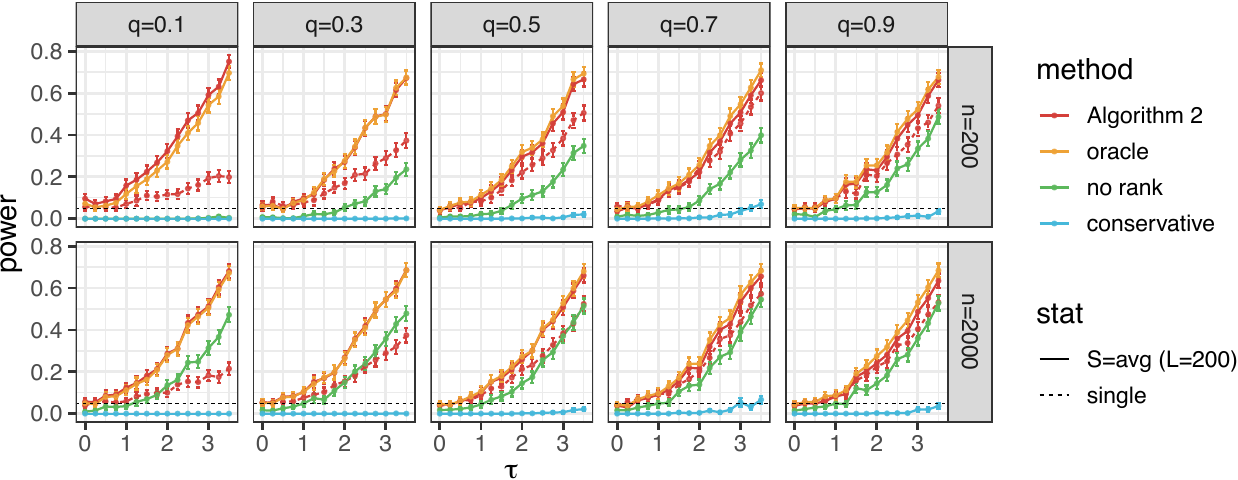}
\caption{Testing $\mu = \mathbf{0}$ in \cref{ex:kim-ramdas}: power (95\% CI) at level $\alpha=0.05$ (dashed horizontal). Location is $\mu = \tau n^{-1/2} v_1$ with $v_1$ the normalised principal eigenvector of $\Sigma$; $q$ is the proportion of the test sample. We compare the single-split test (\texttt{single}) with tests based on $S_n$ being the arithmetic mean (\texttt{S=avg}): Algorithm 2 is our method; `\texttt{oracle}' is an oracle test that compares $S_n$ to its null distribution; `\texttt{no rank}' uses the ordinary subsampling distribution $\hat{G}_n$ to determine the critical value; `\texttt{conservative}' compares $S_n / 2$ to a standard normal, where division by two ensures validity when the statistics are not jointly normal as evidenced by \cref{fig:sampling-KR} (see also \cref{app:aggregation}). Observe that Algorithm 2 has a power advantage over `\texttt{no rank}' and this advantage grows with $\tau$; meanwhile, Algorithm 2 closely tracks  `\texttt{oracle}' regardless of $\tau$. These observations match our first-order theory presented in \cref{sec:power}.} 
\label{fig:power-KR}
\end{figure}

\subsubsection{Testing multivariate unimodality} \label{sec:unimodal}
Testing non-trivial clustering structure of a high-dimensional dataset is a long-standing problem. The problem cannot be directly answered by clustering algorithms, because typically these (e.g., k-means or hierarchical clustering) always return clusters even when the data comes from a homogenous population \citep{huang2015significance}.  This problem is closely connected to selecting the number of clusters as a trivial clustering structure corresponds to the true number being one. 
Here, we work with Euclidean data and we take the perspective that there is only one cluster if the population distribution is unimodal. 

A univariate distribution is called unimodal if there is a point $a$ such that the distribution function is convex on $(-\infty, a)$ and concave on $(a, +\infty)$ \citep{khintchine1938unimodal}. While there are different notions of multivariate unimodality, we take \emph{linear unimodality} as our definition: we say a random vector $(X_1, \dots, X_p)$ is unimodal if $\sum_i a_i X_i$ is unimodal for every non-zero coefficient vector $a = (a_1, \dots, a_p)$. That is, 
\begin{equation} \label{eqs:linear-unim}
H_0: \bigcap_{a \neq \bm{0}} \left\{\sum_{i=1}^p a_i X_i \text{ is unimodal} \right\}.
\end{equation}
Linear unimodality is implied by several other notions related to multivariate unimodality such as log-concavity \citep[Lemma 2.1]{dharmadhikari1988unimodality}; see also \citet[Theorem 2.15]{dharmadhikari1988unimodality}.
The formulation in \cref{eqs:linear-unim} naturally leads to the following hunt-and-test procedure after randomly splitting the data into two parts A and B:
\begin{enumerate}
\item Identify a direction $\hat{a}$ using any suitable clustering algorithm on part A of the data;
\item Test univariate unimodality of $\sum_i \hat{a}_i X_i$ on part B.
\end{enumerate}
We note that the idea of reducing to a univariate test is not new, and is for example used by \citet{ahmed2012investigating} in projecting data onto its principal curve, and in a likelihood ratio test for log-concavity using random projection and data-splitting \citep{dunn2021universal}.

\paragraph{Dip hunting test}
To identify a good direction $\hat{a}$, on part A we run a 2-means algorithm (initialised with k-means++ by \citealp{vassilvitskii2006k}) and choose $\hat{a}$ as the normalised vector connecting the two cluster centres. Then to test for unimodality, we use a test based on the dip statistic due to \citet{hartigan1985dip}, which we describe below. \cref{fig:dip-hunt} shows a schematic of our procedure, which may be described as ``dip hunting'' by analogy with the bump hunting procedure of \citet{good1980density}.

Let $\mathbb{F}_n$ be the empirical distribution function of $Y$ and let $\mathcal{U}$ be the set of unimodal univariate distributions.
The dip statistic is defined as 
\[
\rho_n := \inf_{Q \in \mathcal{U}} \|\mathbb{F}_n - Q\|_{\infty},
\]
and may be computed efficiently using R package \texttt{diptest} \citep{maechler2021diptest}.  \citet{hartigan1985dip} recommend comparing $\rho_n$ to the dip statistic of a sample drawn from $\unif(0,1)$, which serves as the least favourable null distribution. However, this approach typically results in very conservative p-values. To avoid this problem, \citet{cheng1998calibrating} show that when the density $f$ of $Y$ is unimodal, under mild regularity conditions, we have $2 n^{3/5} \rho_n \rightarrow_{d} c \mathcal{Z}$, where $\mathcal{Z}$ is a particular function of a standard Wiener process. 
The constant $c$ depends on the density $f$ and is given by
\[ c = \left\{ f(x_0)^3 / |f''(x_0)| \right\}^{1/5}, \]
where $x_0$ is the unique mode of $f$. The only unknown quantity $c$ can be estimated with $\hat{c}$, which is a plugin from kernel density estimates $\hat{f}$ and $\hat{f''}$ evaluated at $\hat{x}_0 = \argmax \hat{f}$. We use R package \texttt{kedd} \citep{guidoum2015kedd} to estimate $\hat{f}$ and $\hat{f''}$, for which the respective bandwidths are selected with maximum likelihood cross validation (function \texttt{h.mlcv}, \citealp{habbema1974stepwise,duin1976choice}). From our experience, for $\hat{c}$ to behave properly, it is essential to centre and rescale $Y$ so that its value lies between 0 and 1, which does not affect $c$. Because the asymptotic distribution of the dip statistic only depends on $c$, \citet{cheng1998calibrating} suggest the following approach to obtain the corresponding p-value, which we adopt here. Given our observed dip statistic, we compare this to the distribution of $\rho_n$ based on samples drawn from a known distribution whose $c$ equals $\hat{c}$; three families of such distributions covering the range of $c$ are provided by \citet{cheng1998calibrating}.

\begin{figure}[!htb]
\centering
\includegraphics[width=0.6\textwidth]{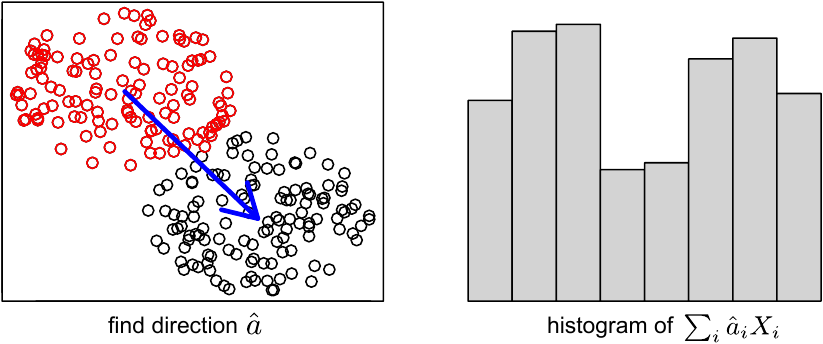}
\caption{Dip hunting test for multivariate unimodality}
\label{fig:dip-hunt}
\end{figure}

\paragraph{Simulations} We consider the following settings. 
\begin{enumerate}
\item \textbf{Mixture of unit balls}. Let $\mathfrak{B}(p)$ be the unit $p$-dimensional ball centred at the origin. Consider the following density with bounded support:
\[ f(x) = 1/2\, \unif_{\mathfrak{B}(p)}(x) + 1/2\, \unif_{\mathfrak{B}(p)}(x - x_0).  \]

Here, $\unif_{\mathfrak{B}(p)}$ denotes the uniform density on the unit ball $\mathfrak{B}(p)$ and $\|x_0\|$ is the Euclidean distance between the centres of the two balls. We set $\|x_0\| = 2 \tau / \sqrt{2+p}$ so that the density in the direction connecting the two ball centres\footnote{The one-dimensional projection of $ \unif_{\mathfrak{B}(p)}(x)$ has density $f(x) = \Gamma(p/2+1) / [\sqrt{\pi} \Gamma(p/2+1/2)] (1-x^2)^{(p-1)/2}$ for $x \in [-1,1]$, which leads to variance $1 / (2+p)$.} becomes bimodal roughly when $\tau \geq 1$.

\item \textbf{Mixture of multivariate t's}. We consider a heavy-tailed setting
\[ f(x) = 1/2\, t_{4}(x; \Sigma) + 1/2 \, t_{4}(x - x_0; \Sigma), \]
where $t_4$ is the density of the $p$-dimensional multivariate t-distribution with 4 degrees of freedom, mean zero and scale matrix $\Sigma \in \R^{p \times p}$ with entries $\Sigma_{ij} = 2^{-|i-j|}$. We set $x_0 = \tau v_2 \sqrt{p}$, where $v_2$ is the second normalised eigenvector of $\Sigma$. 
\end{enumerate}

\cref{fig:unimodal} shows the results based on sample size $n=1000$. 
We compare the single-split dip hunting test (\texttt{single}) with the multiple-split versions that aggregate $L=50$ dip hunting p-values, including $S=\text{avg}$ and $S=\min$ (\cref{alg:agg-test}), as well as the adaptive test (\cref{alg:agg-test-adaptive}) with $(S^1=\text{avg}, S^2=\min)$. 
We compare dip hunting to \texttt{SigClust} \citep{liu2008statistical,huang2015statistical,huang2022R}, a widely-used clustering significance testing method based on a Gaussian mixture model. Other methods include a nonparametric bootstrap approach suitable for ellipsoidal clusters \citep{maitra2012bootstrapping}, and an approach based on simulating from an estimated Gaussian copula model \citep{helgeson2021nonparametric}; see also the review paper \citet{adolfsson2019cluster} and references therein. 
We see that at the null ($\tau = 0$), \texttt{SigClust} incurs a large type-I error for the multivariate t settings, while all the dip hunting tests maintain the correct level in both settings; note that this is to be expected as \texttt{SigClust} assumes a Gaussian distribution under the null. For the unit ball setting, \texttt{SigClust} loses power as $p$ increases. We can see that $S=\text{avg}$ is more powerful than $S=\min$ for the unit ball setting and conversely for the multivariate t setting. Nevertheless, our adaptive test is able to achieve the better performance between the two
and also shows significant power improvement over the single-split test.

\begin{figure}[!htb]
\centering
\includegraphics[width=.99\textwidth]{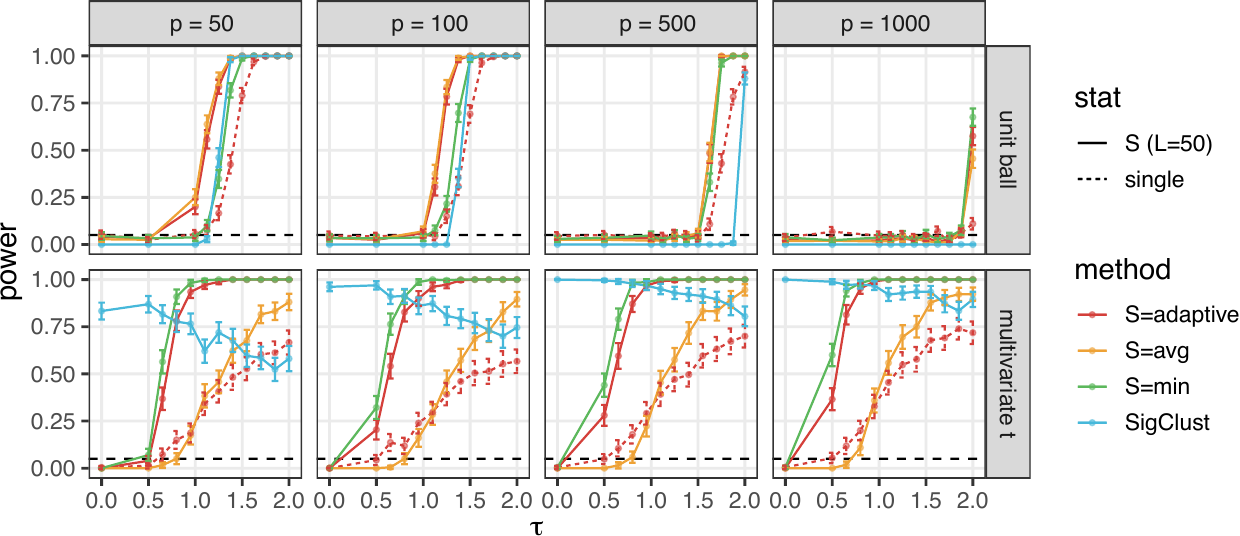}
\caption{Detecting a mixture of two $p$-dimensional unimodal components whose centres are separated $\propto \tau$ away: power (95\% CI) at level $\alpha=0.05$ (dashed horizontal).
The multivariate density is linearly unimodal when $\tau = 0$.
We compare dip hunting with clustering significance testing method \texttt{SigClust}. Test \texttt{single} is the single-split dip hunting test; \texttt{S=avg}, \texttt{S=min} and \texttt{S=adaptive} are multiple-split tests that aggregate $L=50$ dip hunting p-values. The \texttt{S=adaptive} test is \cref{alg:agg-test-adaptive} with $(S^1=\text{avg}, S^2=\min)$, which is able to adapt to the better performance between the two. See also \cref{fig:unimodal-vs-conservative} for the performance of conservative p-value aggregation rules.}
\label{fig:unimodal}
\end{figure}

\paragraph{Gene expression of cancer subtypes} We apply our test to gene expression data on renal cell carcinoma (RCC), which mainly consists of three subtypes: clear cell (ccRCC), papillary  (PRCC) and chromophobe (ChRCC). We use the ICGC/TCGA Pan-Cancer dataset (\citealp{campbell:2020aa}, available from \url{https://bit.ly/3r65AUw}), which contains mRNA expression levels (FPKM-UQ normalised) from 111 kidney samples, including 37 cc, 31 P and 43 Ch RCC cases. We use the expression levels of 1,000 genes that are most relevant to RCC by selecting those with the highest $|\mu_1 - \mu_0| / \sigma_{0}$, where $\mu_1$ and $\mu_0$ are case and control means, $\sigma_0$ is the control standard deviation. We apply both the single-split and the aggregated ($S$ is the arithmetic mean, $L=6000$) dip hunting tests to every subtype, every mixture of two subtypes and the whole sample. \cref{fig:renal} shows the distribution of p-values. The aggregated test produces stable p-values, which indicate clear separation between subtypes and relative homogeneity of each subtype.
In contrast, it is more difficult to disentangle subtypes ccRCC and PRCC from the single-split dip hunting test.

\begin{figure}
\centering
\includegraphics[width=0.95\textwidth]{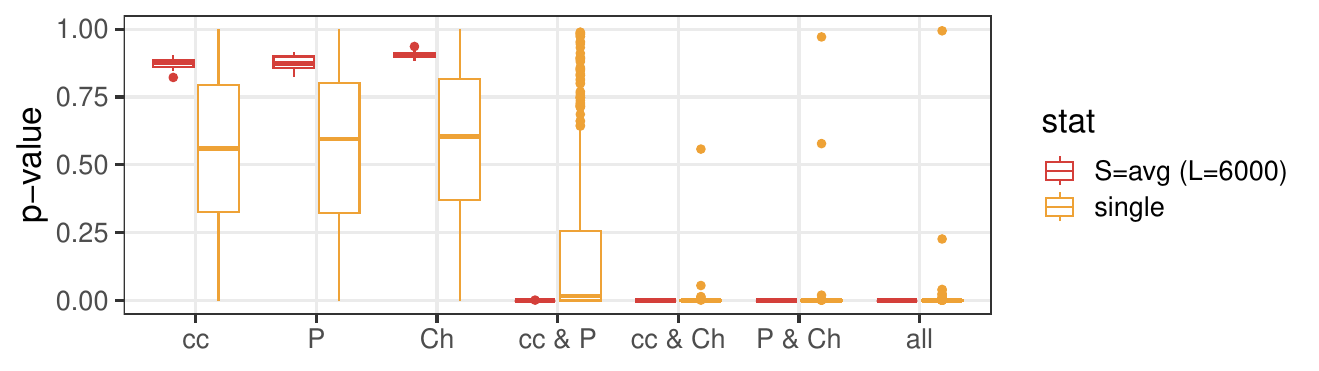}
\caption{Testing homogeneity of gene expression levels of three subtypes of renal cell carcinoma: clear cell (cc), papillary  (P) and chromophobe (Ch). The plot shows p-values resulting from repeatedly applying single-split and aggregated ($S=\text{avg}$, $L=6000$) dip hunting tests to each subtype, every mixture of two subtypes and the whole sample.}
\label{fig:renal}
\end{figure}
\subsubsection{Goodness-of-fit testing for parametric regression models} \label{sec:gof}
As another application, we use hunt and test to construct flexible goodness-of-fit tests for parametric regression models. The approach we take here is closely related to the generalised residual prediction (RP) test (\citealp{jankova2020goodness}; see also \citealp{shah2018goodness}) for assessing the goodness of fit of (potentially high-dimensional) generalised linear models. Generalised RP tests also employ sample splitting and could equally well benefit from our rank transform p-value aggregation scheme. However, our construction here is applicable more broadly to testing model specification of the form 
\begin{equation} \label{eqs:gof_mod}
 	h(X) = \beta^{\T} X,
\end{equation}
where $h(X)$
is
the conditional mean or a conditional quantile of outcome $Y \in \R$ given $X \in \R^p$.

The starting point of our approach is the simple observation that  $h(X) = \beta^{\T} X$ is equivalent to having $\gamma = 0$ in $h(X) = \beta^{\T} X + \gamma g(X)$ where $g: \R^p \to \R$ is any nonlinear (measurable) function. For a given $g$, we may test for whether $\gamma = 0$ by regressing $Y$ on $(X, g(X))$ and utilising existing inference tools for the model at hand to assess the significance of $g(X)$. To obtain good power under an alternative, we would like to pick an appropriate $g$ to expose the lack of fit present in the data. This suggests a hunt-and-test procedure where we randomly divide our data into parts A and B, and use part A to hunt for a suitable $g$, and part B to assess the significance of our artificially constructed additional covariate $g(X)$.

To find an appropriate $g$, we take inspiration from gradient boosting \citep{friedman2001greedy} and proceed as follows. Let $(X_i, Y_i)_{i=1}^{n'}$ be iid covariate--response pairs in part A. Suppose we have an M-estimator $\hat{\beta}$ for estimating $\beta$ in \cref{eqs:gof_mod} that minimises $\sum_{i} \ell(Y_i - \beta^{\T} X_i)$ for some loss function $\ell:\R \to \R$. Defining residual $r_i := Y_i - \hat{\beta}^{\T} X_i$, upon introducing a potential new covariate $g(X)$, the loss can be locally approximated as 
\[ \sum_{i} \ell(Y_i - \hat{\beta}^{\T} X_i - \gamma g(X)) \approx  \sum_{i} \ell(r_i) - \gamma \sum_{i} g(X_i) \, \ell'(r_i),\]
where $\ell'$ is the derivative of the loss function. Hence, when $\gamma >0$ and fixing $\left(\sum_i g(X_i)^2\right)^{1/2}$, to locally decrease the loss by the greatest amount, we should attempt to choose $g$ such that approximately $(g(X_i))_{i=1}^{n'} \propto (\ell'(r_i))_{i=1}^{n'}$. To achieve this, we regress $(\ell'(r_i))_{i=1}^{n'}$ onto the covariates using any flexible regression or machine learning method, and take the fitted regression function to be $g$.
Since we expect that under an alternative the resulting $g(X)$ should have a positive coefficient, we take as our single-split statistic $T_n$ a $Z$-statistic for the significance of $g(X)$ computed on part B. We expect $T_n$ to be large and positive under an alternative.

We demonstrate the effectiveness of this approach for quantile regression. Consider a quantile regression model specified as $q_{\tau}(X) = \beta_0 + \beta^{\T} X$ for a fixed $\tau \in (0,1)$, where $q_{\tau}(X)$ is the $\tau$-th conditional quantile of $Y$ given $p$-dimensional covariates $X$. The construction of goodness-of-fit tests, or more commonly called lack-of-fit tests in the related literature, have largely relied on asymptotic properties of certain statistics or processes concerning the residual; see, e.g., \citet{horowitz2002adaptive,he2003lack,escanciano2010specification,escanciano2014specification}. These tests tend to have difficulty scaling up to more than a handful of covariates \citep{conde2015lack}. Recently, \citet{dong2019lack} recast the goodness-of-fit problem as a two-sample test problem and developed a different, highly competitive method that can handle moderate or large $p$. It is worth mentioning that, unlike our approach that directly repurposes existing parameter inference for quantile regression \citep[Chap.~3]{koenker2005book}, these aforementioned methods rely on asymptotic results that can require substantial development.

For quantile regression, we have $\ell_{\tau}(r) = r(\tau - \I_{r<0})$ and $\ell'_{\tau}(r) = \tau - \I_{r<0}$ for $r \neq 0$. Therefore, we use part A to train a classifier (e.g., random forest, \citealp{breiman2001random}) that predicts the sign of the residual from $X$; we take $g: \R^p \to \{-1,1\}$ to be the resulting prediction function.
To improve numerical stability, we also partial out $X$ from $g(X)$ on part B before adding $g(X)$ to the covariates. Define the single-split test statistic as $T_n := \sqrt{n/2} \hat{\beta}' / \hat{\sigma}$ (the prime indicates fitted from part A), where $\hat{\sigma}$ is estimated from bootstrap. The statistic can be readily computed using R package \texttt{quantreg} \citep{koenker2022R}.

\cref{fig:quant} shows results from a simulation study under sample size $n=1000$. Covariate vector $X$ is drawn from $p$-dimensional Gaussian with covariance $\Sigma_{ij} = 2^{-|i-j|}$. We fix $\tau = 0.5$ and consider two specifications
\begin{equation} \label{eqs:quant-spec}
\begin{split}
\text{(i)} \quad & Y = 1 + \beta_0^{\T} X + (v / \sqrt{n}) \eta(X) + (1 + X_2 + X_3) \epsilon, \\
\text{(ii)} \quad & Y = 1 + \beta_0^{\T} X + \left[1 + X_2 + X_3 + (v / \sqrt{n}) \eta(X) \right] \epsilon,
\end{split}
\end{equation}
where non-linear function $\eta(X) = 4\sqrt{X_1^2 + X_2^2}$ introduces misspecification. We set $\beta_0^{\T} = (-1,2,0,-1,2,0,\dots)$. Error $\epsilon$ is drawn from $\{\expo(1), t_3\}$ for specification (i) and from $\expo(1)$ for specification (ii). Because $\expo(1)$ has a non-zero median, observe that the quantile regression model $q_{\tau}(X) = \beta_0 + \beta^{\T} X$ is well-specified if and only if $v=0$.
We choose $g$ to be a random forest trained with R package \texttt{ranger} \citep{wright2017ranger}.
We run \cref{alg:agg-test} with $S$ being the arithmetic mean.
Our approach is already competitive with the state-of-the-art method of \citet{dong2019lack}; see \cref{app:quant} for results in another setting.

\begin{figure}[!htb]
\centering
\includegraphics[width=.95\textwidth]{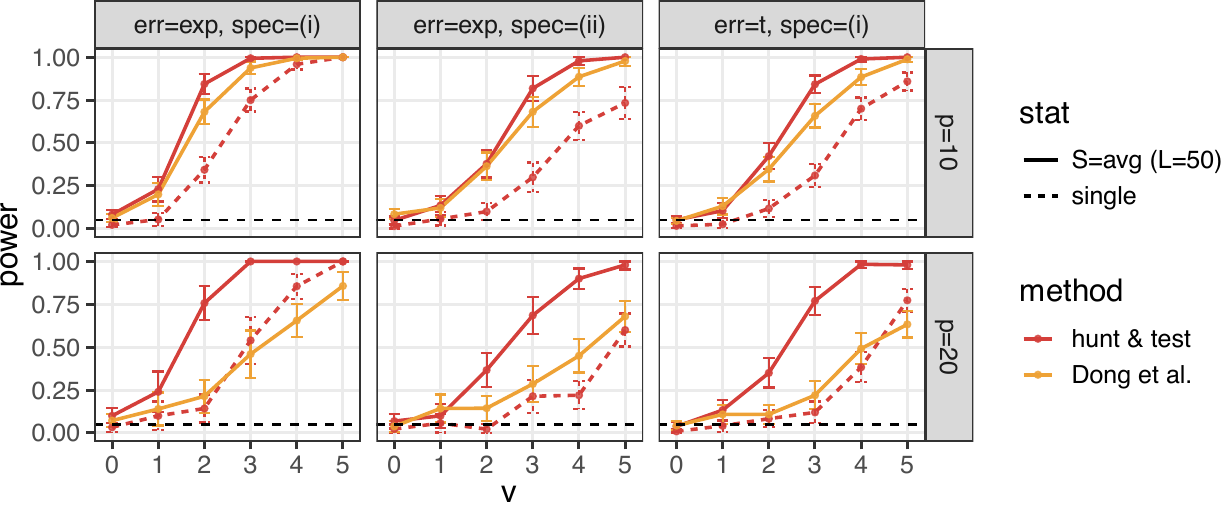}
\caption{Testing goodness-of-fit of a quantile regression model $q_{0.5}(X) = \beta_0 + \beta^{\T} X$: power (95\% CI) at level $\alpha=0.05$ (dashed horizontal) under $n=1000$. The model is well-specified if and only if $v=0$. The non-linear function in \cref{eqs:quant-spec} is $\eta(X) = 4\sqrt{X_1^2 + X_2^2}$. See also \cref{fig:quant-nl-2,fig:quant-vs-conservative} in the supplementary material.}
\label{fig:quant}
\end{figure}

\subsection{Testing generalised conditional independence} \label{sec:verma}
In this section, we consider the use of randomised procedures in causal inference. 
To draw causal conclusions from data, we are often faced with the challenge that the quantity of interest is defined with respect to an ``intervened'' distribution, which is related to but different from the data-generating distribution. This is also known as a ``distributional shift'' in machine learning \citep{shimodaira2000improving}.

For example, consider a two-stage sequentially randomised trial represented by graph $\g$ in \cref{fig:verma}: $A_1$ is the first treatment, $L$ is the first outcome, $A_2$ is the second treatment and $Y$ is the final outcome. The first treatment $A_1$ is completely randomised; the second treatment $A_2$ is randomised according to $A_1$ and $L$. Variable $U$ represents an unobserved confounder, for example, the underlying health status that affects both outcomes and is unobserved. 
Graph $\g$, typically with extra base covariates that we omit here, is often employed to represent observational or follow-up studies where the treatment is time-varying and is affected by previous outcomes. For example, in HIV studies, we have $A_t=1$ ($t=1,2$) if the individual receives antiretroviral therapy at time $t$ and $A_t=0$ otherwise. Outcomes $L$ and $Y$ denote CD4 cell counts that measure the effectiveness of the therapy. See \citet[Chap.~19]{hernan2020causal} for more background.

Suppose we are interested in the first treatment's \emph{direct effect} $\tau$ on the final outcome (i.e., not through the second treatment), represented by the dashed edge in $\g$. Because $A_1$ also affects $Y$ through $A_2$, we cannot learn $\tau$ by regressing $Y$ on $A_1$ alone. Moreover, because of the latent $U$ that affects both $L$ and $Y$ (such variables are called ``phantoms'' by \citealp{bates2022causal}), nor can we learn $\tau$ by additionally adjusting for $L$ or $A_2$ (or both) in the regression. 

\begin{figure}[htb]
\centering
\begin{tikzpicture}
\tikzset{rv/.style={circle,inner sep=1pt,fill=gray!20,draw,font=\sffamily},
lv/.style={circle,inner sep=1pt,fill=gray!50,draw,font=\sffamily},
fv/.style={rectangle,inner sep=1.5pt,fill=gray!20,draw,font=\sffamily}, 
node distance=12mm, >=stealth}
\begin{scope}
\node[rv, yshift=-1cm] (A0) {$A_1$};
\node[rv, right of=A0] (L) {$L$};
\node[rv, right of=L] (A1) {$A_2$};
\node[rv, right of=A1] (Y) {$Y$};
\node[lv, below of=A1, yshift=4mm] (U) {$U$};
\draw[->,very thick, color=blue, bend left, out=40, in=140] (A0) to (A1);
\draw[->,very thick, color=blue] (A0) -- (L);
\draw[->,very thick, color=purple, dashed, out=40, in=140] (A0) to (Y);
\node[above of=A1, xshift=-4.5mm, yshift=-1mm] (l) {\textcolor{purple}{$\tau$}};
\draw[->,very thick, color=blue] (L) -- (A1);
\draw[->,very thick, color=blue] (A1) -- (Y);
\draw[->,very thick, color=blue] (U) -- (L);
\draw[->,very thick, color=blue] (U) -- (Y);
\node[below=7mm of L, xshift=5mm] {$\g$};
\end{scope} \begin{scope}[xshift=6cm]
\node[rv, yshift=-1cm] (A0) {$A_1$};
\node[rv, right of=A0] (L) {$L$};
\node[rv, right of=L] (A1) {$A_2$};
\node[rv, right of=A1] (Y) {$Y$};
\node[lv, below of=A1, yshift=4mm] (U) {$U$};
\draw[->,very thick, color=blue] (A0) -- (L);
\draw[->,very thick, color=purple, dashed, out=40, in=140] (A0) to (Y);
\node[above of=A1, xshift=-4.5mm, yshift=-1mm] (l) {\textcolor{purple}{$\tau$}};
\draw[->,very thick, color=blue] (A1) -- (Y);
\draw[->,very thick, color=blue] (U) -- (L);
\draw[->,very thick, color=blue] (U) -- (Y);
\node[below=7mm of L, xshift=5mm] {$\g'$};
\end{scope} \end{tikzpicture}
\caption{Graph $\g$ depicts a sequentially randomised trial: $A_1$ is the first treatment, $L$ is the first outcome, $A_2$ is the second treatment, and $Y$ is the final outcome. Latent variable $U$ represents the underlying health status that affects both outcomes. The dashed edge represents the direct effect from $A_1$ on $Y$. 
Graph $\g'$ represents the population where both $A_1$ and $A_2$ are completely randomised. In this case, if $A_1$ has no direct effect on $Y$ (i.e., the dashed edge is absent), then we can observe $A_1 \indep Y$.}
\label{fig:verma}
\end{figure}
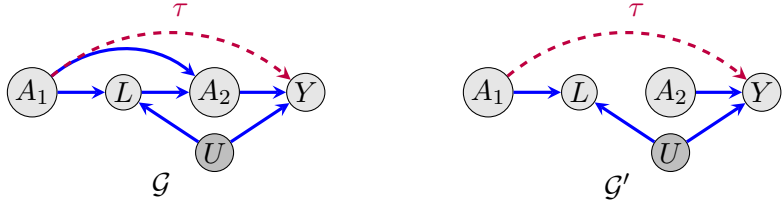

To learn $\tau$, it is useful to imagine another trial, drawn as $\g'$ in \cref{fig:verma}, where both treatments $A_1$ and $A_2$ are completely randomised so that $A_1$ only affects $Y$ directly. There, we can easily learn $\tau$ by regressing $Y$ on $A_1$ in the $\g'$ distribution. Although we did not carry out the $\g'$ trial, its data distribution can be approximated by reweighting our data obtained from the $\g$ trial according to the inverse propensity of $A_2$ given $A_1$ and $L$. Further, we can even artificially \emph{simulate} data from the $\g'$ trial by resampling (e.g., importance resampling or rejection sampling) our data from the $\g$ trial. In other words, reweighting or resampling provides access to our distribution of interest. This idea underlies a general approach known as the g-methods \citep{robins1986new,naimi2017introduction}.

Specifically, let us consider testing the sharp null hypothesis 
\begin{equation} \label{eqs:H0-verma}
H_0: \text{$A_1$ has no individual direct effect on $Y$},
\end{equation}
which is represented by the dashed edge from $A_1$ to $Y$ being absent from $\g$. As explained earlier, this amounts to $A_1 \indep Y$ in the population $Q$ represented by $\g'$, given by
\begin{equation} \label{eqs:verma-lr}
dQ / dP = q(A_2) / p(A_2 \mid A_1, L),
\end{equation}
where $q$ is an arbitrary positive distribution over $A_2$. Constraints as such, which prescribe independence or conditional independence in a reweighted distribution, are called generalised conditional independence or Verma constraints in the literature; see \citet{richardson2017nested} and the references therein.

Because the independence holds under $Q$ instead of $P$, the usual permutation test is not applicable; nor is it applicable through simple reweighting as employed by \citet{berrett2020conditional} for testing ordinary conditional independence.
Instead, the standard approach in causal inference is through inverse probability weighting (IPW).
Here, we consider an alternative approach: we resample our data to represent $Q$ according to \cref{eqs:verma-lr} and then use the resampled data to test $A_1 \indep Y$, e.g., by a permutation test or any off-the-shelf test for independence. 

More generally, as demonstrated by \citet{thams2023statistical}, such a test-after-subsampling procedure is applicable to testing any property under ``distributional shift''. 
Suppose we observe iid sample $X_1, \dots X_n \sim P$ but we are interested in testing a property of a different target distribution $Q$. Distribution $Q$ is related to $P$ through a density ratio $r = \dd Q / \dd P$, which is either known or can be estimated from $P$. 
Suppose we already have a suitable test for the property based on a test statistic $T_n := T_n(\tilde{X}_1, \dots, \tilde{X}_n)$ for iid sample $\tilde{X}_1, \dots, \tilde{X}_n \sim Q$. Assume $r$ is bounded from above by a constant $C$. Then, we can test the property in two steps:
\begin{enumerate}
\item Obtain a sample from $Q$ through rejection sampling;
\item Test the property with the test statistic computed from the accepted sample.
\end{enumerate}

The next result shows that the test statistic computed as such inherits the desired asymptotic distribution while permitting the use of an estimated density ratio $\hat{r}_n$. 

\begin{proposition}[Test after rejection sampling] \label{prop:test-rej}
Let $T_n(X_1, \dots, X_n)$ be a test statistic. Define $T_0 := 0$. With $X:=(X_1, \dots, X_n)$, suppose $T_n(X_1, \dots, X_n) \rightarrow_d T$ under $X \sim Q^n$. 
Let $P$ be a distribution and $C$ be a positive constant, such that $Q$ is absolute continuous with respect to $P$ and $r:= \dd Q / \dd P$ satisfies 
\[ r(x) \leq C, \quad \text{$P$-almost every $x$}. \]
Let $\hat{r}_n$ be a random sequence of approximate density ratios such that almost surely,
\[ \int \hat{r}_n \dd P =1, \quad 0 \leq \hat{r}_n(x) \leq C \text{ for $P$-almost every $x$}. \]
Also, suppose 
\begin{equation} \label{eqs:dens-ratio-tv}
\E \int |\hat{r}_n - r| \dd P = o(n^{-1}).
\end{equation}
With $U:=(U_1, \dots, U_n)$, consider $(X, U) \sim P^n \times \unif^{n}(0,1)$. For $i=1,\dots,n$, let the $i$-th sample be accepted if $U_i < \hat{r}_n(X_i) / C$. Let the set of accepted sample be denoted as $\tilde{X}^{n}_{1}, \dots, \tilde{X}^{n}_{\Gamma_n}$, where $0 \leq \Gamma_n \leq n$ is the number of acceptances. Then it holds that 
\[ T_{\Gamma_n}(\tilde{X}^{n}_1, \dots, \tilde{X}^{n}_{\Gamma_n}) \rightarrow_d T. \]
\end{proposition}

While the rate in \cref{eqs:dens-ratio-tv} seems to demand a large separate sample for estimating $r$, we expect that such rate can be relaxed under additional conditions on the density and the test statistic. For example, when the density ratio is parametrically specified and the test statistic is simple, see \cref{prop:test-rej-parametric} and the following remark in the supplementary material.

When the density ratio $r$ is unbounded, the rejection sampling can be replaced by the distinct replacement sampling (DRPL) introduced by \citet{thams2023statistical}, which selects a resampled dataset $(\tilde{X}_1 = X_{i_1}, \dots, \tilde{X}_k = X_{i_k})$ of given size $k$ by sampling distinct indices $(i_1, \dots, i_k)$ with probability proportional to $\prod_{l=1}^k r(X_{i_l})$. It is shown that, under a moment condition on $r$, by choosing $k = o(\sqrt{n})$, the statistic $T_n$ evaluated on the DRPL resampled data converges to its asymptotic distribution under $Q$ \citep[Theorem 1]{thams2023statistical}. However, DRPL produces an $o(\sqrt{n})$ sample while rejection sampling produces an $O(n)$ sample with high probability.

Both rejection sampling and DRPL lead to tests that are randomised and potentially of low power due to a reduced sample size. However, as we demonstrate here in the case of testing  \cref{eqs:H0-verma} in \cref{fig:verma}, our aggregation scheme is able to restore highly competitive levels of power.
For simplicity, we choose the post-rejection/DRPL-sampling test statistic $T_n$ as the permutation p-value for $|\cov_Q(A_1,Y)|$, which is uniformly distributed when $A_1 \indep Y$ holds under $Q$. As a benchmark, we also consider testing \cref{eqs:H0-verma} with the standard IPW version of the statistic \citet[p.~10]{robins1999testing}, which directly uses the sample under $P$. Let
\[ Z_i := \frac{Y_i (A_{1,i} - \E A_{1})}{p(A_{2,i} \mid L_i, A_{1,i})}, \]
where $\E A_{1}$ and $p(A_{2} \mid L, A_{1})$ can be replaced by their consistent, asymptotically linear estimators. 
Under $H_0$, we have
\begin{equation} \label{eqs:ipw-verma}
\chi_n := \frac{\sum_i Z_i}{\sqrt{\sum_i Z_i^2}} \rightarrow_d \N(0,1), 
\end{equation}
which leads to a two-sided test; see also \cref{prop:CAN} in the supplementary material. 

We consider the following data generating mechanism. We have binary treatments $A_1, A_2$ and real-valued outcomes $L, Y$. The latent $U$ is a 4-dimensional random vector. Let 
\[  A_1 \sim \ber(1/2),  \quad A_2 \sim \ber(\text{expit}(2 A_1 - L + 2)), \]
and 
\[ U \sim \N(0, \Sigma_{ij} = 2^{-|i-j|}), \quad L = A_1 + \beta_{U,L}^{\T} U + \epsilon_L, \quad Y = \tau A_1 - A_2 + \beta_{U,Y}^{\T} U + \epsilon_{Y}, \]
where $\beta_{U,L} = (1,1,-2,2)^{\T}$, $\beta_{U,Y} = (2,-1,3,-10)^{\T}$ and $\epsilon_L, \epsilon_Y \sim \N(0,1)$ independently. Parameter $\tau$ controls the effect size: $H_0$ in \cref{eqs:H0-verma} holds if and only if $\tau = 0$. 

\cref{fig:simu-verma} shows the result for sample size $n=1000$ under $P$. For rejection sampling, $q(A_2)$ in \cref{eqs:verma-lr} is chosen to maximise the acceptance rate; for DRPL, we choose $q(A_2) = p(A_2)$ and $k = \floor{\sqrt{n}}$.
We use \cref{alg:agg-test} to aggregate $L=20$ post-rejection/DRPL-sampling permutation p-values ($S$=avg).
Presumably because rejection sampling generates a larger sample than DRPL, rejection sampling tests are more powerful than DRPL tests here. For both rejection sampling and DRPL, `avg' considerably boosts the power while maintaining the correct level at $\tau = 0$. Perhaps surprisingly, the `avg' rejection sampling test outperforms the IPW test based on \cref{eqs:ipw-verma}, even though they target the same population quantity $\cov_{Q}(A_1,Y)$.

\begin{figure}[!htb]
\centering
\includegraphics[width=.8\textwidth]{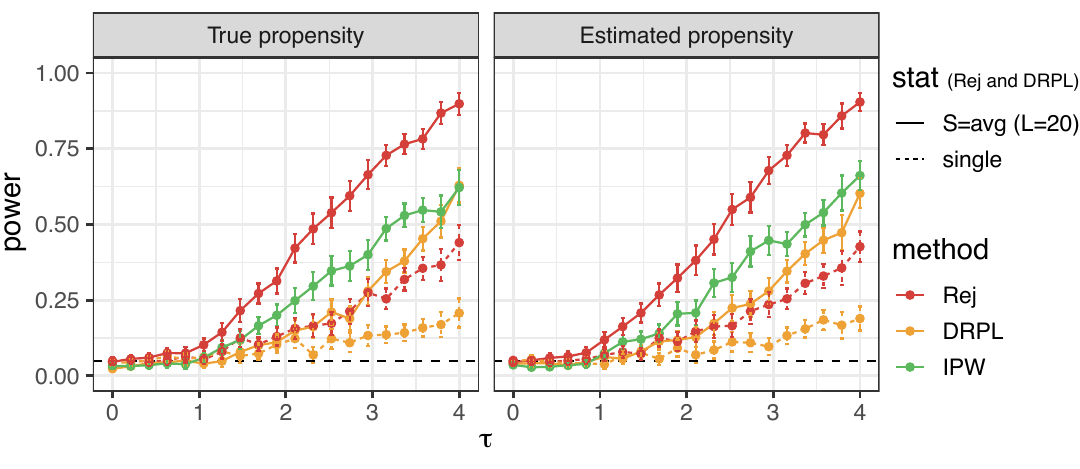}
\caption{Power for testing no individual direct effect of $A_1$ on $Y$ in graph $\g$ of \cref{fig:verma}, based on $n=1000$ sample at level 0.05 (dashed horizontal). The null hypothesis corresponds to $\tau = 0$. \texttt{Rej}: post-rejection-sampling permutation test for $|\cov_Q(A_1,Y)|$, \texttt{DRPL}: post-DRPL-sampling \citep{thams2023statistical} permutation test for $|\cov_Q(A_1,Y)|$, \texttt{IPW}: inverse probability weighted test based on $\cov_{Q}(A_1,Y) = 0$ given by \cref{eqs:ipw-verma}. \texttt{Rej} and \texttt{DRPL} are randomised tests (`single'). Their `$S$=avg' version is the aggregated test by taking the average of 20 p-values, calibrated by \cref{alg:agg-test}. Left panel: the true propensity $p(A_2 \mid A_1, L)$ is used for sampling; right panel: the propensity used for sampling is estimated using logistic regression. See also \cref{fig:verma-vs-conservative} in the supplementary material for comparison with conservatively merged p-values.}
\label{fig:simu-verma}
\end{figure}

Our strategy is particularly useful when it is difficult to detect the dependence between $A_1$ and $Y$ under $Q$ with simple, ``mean-like'' statistics such as covariance and a more complex statistic must be employed. In such cases, it is not always clear how to reweight and calibrate the statistic under $P$ using IPW. In contrast, post-resampling test works out of the box whenever a test under $Q$ can be constructed. We present such an example in \cref{app:verma}.

\subsection{Calibrating cross-fitted double machine learning} \label{sec:dml}
Double/debiased machine learning (DML) methods, also known as doubly robust targeted estimation,  employ flexible, machine learning methods to estimate low-dimensional target parameters. These methods are widely used for a variety of problems; see \citet{diaz2020machine,kennedy2022semiparametric} for recent reviews. In this section, we follow the  setup of \citet{chernozhukov2018dml} and focus on estimation with Neyman orthogonal scores, but our method can be extended to other settings such as those that are based on influence functions with a mixed bias property \citep{rotnitzky2021characterization}.

Suppose we have an estimating equation $\psi(W; \theta, \eta)$ for an unknown, real-valued target parameter $\theta_0$ such that $\E \psi(W; \theta_0, \eta_0) = 0$. 
Here $W$ denotes the data point and $\eta_0$ denotes two or more unknown nuisance parameters. Suppose $\psi(W; \theta, \eta)$ satisfies the Neyman orthogonality condition \citep[Definition 2.1]{chernozhukov2018dml} and takes the form
\[ \psi(W; \theta, \eta) = \psi^{a}(W; \eta) \, \theta + \psi^{b}(W; \eta). \]
The DML approach proposes to estimate $\eta_0$ using flexible machine learning methods, and employs sample-splitting to control bias that may otherwise be introduced by overfitting.
Suppose  iid data $W_1, \dots, W_n$ are split into $L$ equal-sized folds. For $l=1,\dots,L$, let $\hat{\eta}^{(-l)}$ be  the nuisance parameters learned from all the data except the $l$-th fold. Then, let $\hat{\theta}^{(l)}$ be the solution to 
\[ \mathbb{P}_n^{(l)} \psi(W; \theta, \eta^{(-l)}) = 0, \quad l=1,\dots,L, \]
where $\mathbb{P}_n^{(l)}$ denotes the empirical measure of data in the $l$-th fold. Under 
regularity conditions and a sufficiently fast rate of nuisance estimation \citep[Assumptions 3.1 \& 3.2]{chernozhukov2018dml}, it can be shown that the single-split statistic 
\begin{equation} \label{eqs:dml-single}
T_n^{(l)} := \frac{\sqrt{n/L} (\hat{\theta}_n^{(l)} - \theta_0)}{\sigma} \rightarrow_{d} \N(0,1), \quad l=1,\dots,L,
\end{equation}
where the asymptotic variance is
\begin{equation} \label{eqs:sigma-dml}
\sigma^2 = \E \psi^2(W; \theta_0, \eta_0) / \{\E \psi^a(W; \eta_0) \}^2. 
\end{equation}
The so-called cross-fitted DML-1 estimator is simply the average $\hat{\theta}^{\text{dml}} := (\hat{\theta}^{(1)} + \dots + \hat{\theta}^{(L)}) / L$. Under the same conditions as above, it holds that
\begin{equation} \label{eqs:dml-clt}
\sqrt{L} S_n = \frac{1}{\sqrt{L}} \sum_{l=1}^T T_n^{(l)} = \frac{\sqrt{n} (\hat{\theta}_n^{\text{dml}} - \theta_0)}{\sigma} \rightarrow_d \N(0,1), 
\end{equation}
where $S_n$ is the average of exchangeable statistics $T_n^{(1)}, \dots, T_n^{(L)}$. Provided with a consistent estimator of $\sigma$, standard large-sample DML confidence intervals are constructed from \cref{eqs:dml-clt}.

In view of \cref{eqs:dml-single}, \cref{eqs:dml-clt} effectively ignores the between-fold correlations among statistics $T_n^{(1)}, \dots, T_n^{(L)}$ because their contribution is asymptotically negligible under the conditions. However, in practice, such correlations can be visible in finite sample and hence undermine the coverage of normal confidence intervals constructed from \cref{eqs:dml-clt}. When the between-fold correlation is $\rho$, the variance of $\sqrt{n}(\hat{\theta}_n^{\text{dml}} - \theta)$ is roughly $\sigma^2(1 + \rho(L-1))$ instead of $\sigma^2$.  We propose to use \cref{alg:agg-test} to account for such correlation and construct confidence intervals with better finite-sample coverage. In this context, $\hat{H}_{b,l}$ in \cref{alg:agg-test} is the subsample counterpart of \cref{eqs:dml-single}, which involves unknown $\sigma$ and $\theta_0$. Nevertheless, by monotonicity, the rank-transformed $\tilde{\mathbf{H}}$ can be computed from the ranks of the subsampling counterparts of $\hat{\theta}_n^{(l)}$ directly as
\[
\tilde{H}_{b,l} = \Phi^{-1}\left(\mathbb{F}_{\{\hat{\theta}_{m,b}^{(l)}\}} (\hat{\theta}_{m,b}^{(l)}) -1 / (2BL) \right), \quad b=1,\dots,B, \quad l=1,\dots,L.
\]
Here $\hat{\theta}_{m,b}^{(l)}$ is the $l$-th single-split estimator computed from the $b$-th subsample. Further, because $\tilde{H}_{b,l} \approx \hat{H}_{b,l} = \sqrt{m/L} (\hat{\theta}_{m,b}^{(l)} - \theta_0) / \sigma$ with high probability, we get the following estimator for $\sigma$ as a by-product:
\[
\hat{\sigma}_{\text{ls}} := \sqrt{m/L} / \hat{\beta}_1.
\]
Here $\hat{\beta}_1$ the slope from least squares $\tilde{H}_{b,l} \sim \beta_1 \hat{\theta}_{m,b}^{(l)} +\beta_0$.
We perform this least squares regression only using those points for which
$\mathbb{F}_{\{\hat{\theta}_{m,b}^{(l)}\}} (\hat{\theta}_{m,b}^{(l)}) \in (\epsilon/2,1-\epsilon/2) $ for some small $\epsilon > 0$ to avoid high-leverage points influencing the fit too much; in our numerical experiments we use $\epsilon = 0.1$.
\begin{proposition} \label{prop:sigma}
Suppose $\hat{\theta}_n^{(1)}, \dots \hat{\theta}_n^{(L)}$ are exchangeable and \cref{eqs:dml-single} holds. Then, $\hat{\sigma}_{\text{ls}} \rightarrow_p \sigma$. 
\end{proposition}

Consequently, our rank-transformed subsampling confidence interval for $\theta_0$ is 
\begin{equation} \label{eqs:rank-dml-CI}
\left[\hat{\theta}_n^{\text{dml}} - \sqrt{L / n}\, \hat{\sigma}_{\text{ls}}\, \tilde{G}_n^{-1}(1-\alpha/2),\: \hat{\theta}_n^{\text{dml}} - \sqrt{L / n}\, \hat{\sigma}_{\text{ls}}\, \tilde{G}_n^{-1}(\alpha/2) \right].
\end{equation}

As a simple example, consider inferring $\theta_0$ in a partially linear model 
\begin{align*}
D = m_0(X) + V, \quad & \E[V \mid X] = 0,\\
Y = D \theta_0 + g_0(X) + \xi, \quad & \E[\xi \mid D, X] = 0
\end{align*}
using  $\psi$  given by Robinson's score function \citep[\S 4.1]{chernozhukov2018dml}
\[ \psi(W; \theta, \eta) = \{Y - l(X) - \theta (D - m(X)) \}(D - m(X)),\quad l(X):=\E[Y \mid X],\: m(X) = \E[D \mid X].\]
In our numerical example, we choose heteroscedastic errors $V$ and $\xi$ (see \cref{app:dml} for details) and fit nuisance parameters $\eta := (l, m)$ with random forests. \cref{tab:simu-dml} compares the coverage of our confidence interval \cref{eqs:rank-dml-CI} with the coverage of standard DML normal confidence interval based on \cref{eqs:dml-clt}. For the latter, we use a plugin estimate of $\sigma$ that replaces expectations in \cref{eqs:sigma-dml} by their empirical counterparts.
While DML confidence intervals tend to undercover for smaller samples, where $\rho(L-1)$ is larger, our confidence intervals are well calibrated. 

When one of the two nuisance functions above is misspecified, under assumptions, $\hat{\theta}_n^{\text{dml}}$ can still be consistent and asymptotically normal due to the double robustness of the estimating equation. However, the estimators $\hat{\theta}_n^{(1)}, \dots \hat{\theta}_n^{(L)}$ are correlated, even asymptotically.
In such a case, the standard DML confidence interval no longer has the desired asymptotic coverage but our confidence interval \cref{eqs:rank-dml-CI} should still work. See \cref{app:dml} for such an example; see also \citet{benkeser2017doubly} for related methods. 

\begin{table}[!htb]
\begin{center}
\caption{Coverage of nominal $95\%$ confidence intervals (brackets: median width of intervals; $\rho$ is the between-fold correlation of $T_n^{(l)}$)} \label{tab:simu-dml}
\begin{tabular}{cccccccccc}
\toprule[1.2pt]
\multicolumn{1}{c}{}&&\multicolumn{2}{c}{$n=500$}&&\multicolumn{2}{c}{$n=1000$}&&\multicolumn{2}{c}{$n=2000$}\\ 
\cline{3-4}\cline{6-7}\cline{9-10}
method&& $L=2$ & $L=5$ && $L=2$ & $L=5$ && $L=2$ & $L=5$\\ 
\midrule
$\rho(L-1)$ && 0.46 & 0.31 && 0.36 &  0.18 &&  0.25 &  0.14 \\ 
\midrule
\cref{eqs:rank-dml-CI} &&0.94 [0.21]&0.93 [0.19]&&0.95 [0.15]&0.95 [0.13]&&0.96 [0.10]&0.95 [0.09]\\ 
DML &&0.86 [0.15]&0.88 [0.15]&&0.88 [0.11]&0.92 [0.11]&&0.91 [0.08]&0.92 [0.08]\\ 
\bottomrule[1.2pt]
\end{tabular}
\end{center}
\end{table}

\section{Discussion} \label{sec:discuss}
Rank-transformed subsampling provides a framework that properly aggregates results from multiple applications of a randomised statistical test. Backed by this framework, we are free to design explicitly randomised tests that employ data splitting, resampling or any other random processing to solve difficult problems. In particular, data splitting and resampling can often reduce a complex hypothesis to a simpler one, for which an off-the-shelf test can be repurposed in a ``plug-and-play'' fashion.
Even though the resulting ``single-split'' randomised procedure may have high variability or low power, as we have demonstrated, these aspects can be significantly improved by our aggregation algorithms.

There are several aspects worth further investigation. Firstly, it is of interest to go beyond \cref{alg:agg-test-adaptive} and study how to learn an optimal aggregation function. Secondly, with a large $L$ and $B$, constructing $\hat{\mathbf{H}}$ can be computationally intensive, although it is embarrassingly parallelisable across both rows and columns. It is desirable to develop faster approximations to $\hat{\mathbf{H}}$. Thirdly, we find $m = \floor{n / \log n}$ works surprisingly well in our numerical studies,
in spite of only guaranteeing around 7 independent subsamples for $n=1000$. Essentially, rank-transformed subsampling is able to approximate the null sampling distribution far better than the theory might suggest, and it would be interesting to investigate why this might be the case, and whether there is a better choice of the subsample size (e.g., using the rule of \citealp{bickel2008choice}).

\paragraph{Acknowledgements} The authors thank Richard Moulange and Sach Mukherjee for help with the TCGA data, and Richard Samworth, Qingyuan Zhao and four anonymous referees for their helpful comments.
Both authors were funded by EPSRC grants EP/N031938/1 and EP/R013381/1.

\bibliographystyle{plainnat}

\newpage
\appendix

\section{Gaussian location experiment} \label{app:gauss-loc}
Here we derive expressions for the power and probability of replication relating to \cref{ex:gauss-loc}. The sequence $T^{(1)}, T^{(2)}, \dots$ follows a Gaussian process with mean and covariance
\[ \mu_i = \mu, \quad \Sigma_{\rho} = \begin{cases} 1, &\quad i=j \\ \rho, & \quad i\neq j \end{cases}.\]
It follows that 
\[ \begin{pmatrix} \bar{T}_{1:L} \\ \bar{T}_{(1+L):2L} \end{pmatrix} := 
\begin{pmatrix} L^{-1} \sum_{l=1}^L T^{(l)} \\ L^{-1} \sum_{l=L+1}^{2L} T^{(l)} \end{pmatrix} \sim \N\left(\begin{pmatrix} \mu \\ \mu \end{pmatrix}, \begin{pmatrix} 1/L + \rho (L-1)/L & \rho \\ \rho & 1/L + \rho (L-1)/L \end{pmatrix} \right).
\]
The critical value is
\[ \bar{t}_{\alpha} = \sqrt{1/L + \rho (L-1) / L} \, z_{\alpha}, \]
which reduces to $z_{\alpha}$ when $L=1$. 
The power function of the aggregated test given by
\begin{align*}
\E \phi_{\bar{T}_{1:L}} = \pr(\bar{T}_{1:L} > \bar{t}_{\alpha}) &= \Phi\left( \frac{\mu}{\sqrt{1/L + \rho(L-1)/L}} - z_{\alpha} \right) \\
&= \alpha + \phi(z_{\alpha}) \frac{\mu}{\sqrt{\frac{1}{L} + \frac{L-1}{L} \rho}} + o(\mu)
\end{align*}
as $\mu \searrow 0$. To derive the expression for $\pr(\phi_{\bar{T}_{1:L}} \neq \phi_{\bar{T}_{(L+1):2L}})$, consider the representation
\[  \begin{pmatrix} \bar{T}_{1:L} \\ \bar{T}_{(1+L):2L} \end{pmatrix} =_d  \begin{pmatrix} \mu + \sqrt{\rho} Z_0 + \sqrt{(1-\rho) / L} Z_1 \\  \mu + \sqrt{\rho} Z_0 + \sqrt{(1-\rho) / L} Z_2 \end{pmatrix}, \]
where $Z_0, Z_1, Z_2$ are independent standard normals. Then we can write
\begin{align*}
& \quad \pr(\bar{T}_{1:L} > \bar{t}_{\alpha}, \bar{T}_{(L+1):2L} > \bar{t}_{\alpha}) \\
& = \E \pr(\sqrt{(1-\rho) / L} Z_1 > \bar{t}_{\alpha} - \mu - \sqrt{\rho} Z_0, \sqrt{(1-\rho) / L} Z_2 > \bar{t}_{\alpha} - \mu - \sqrt{\rho} Z_0 \mid Z_0) \\
&= \E \left\{ 1 - \Phi\left(\frac{\bar{t}_{\alpha} - \mu - \sqrt{\rho} Z_0}{\sqrt{(1-\rho) / L}} \right) \right \}^2 = \E \Phi^2 \left(\frac{\mu + \sqrt{\rho} Z_0 - \bar{t}_{\alpha}}{\sqrt{(1-\rho) / L}} \right),
\end{align*}
and similarly
\[ \pr(\bar{T}_{1:L} \leq \bar{t}_{\alpha}, \bar{T}_{(L+1):2L} \leq \bar{t}_{\alpha}) =  \E \Phi^2 \left(- \frac{\mu + \sqrt{\rho} Z_0 - \bar{t}_{\alpha}}{\sqrt{(1-\rho) / L}} \right). \]
Thus
\begin{gather*}
\pr\left(\phi_{\bar{T}_{1:L}} \neq \phi_{\bar{T}_{(L+1):2L}} \right) = 1 - \int \left\{\Phi^2(y) + \Phi^2(-y) \right\} \phi(x) \dd x, \\
y=\frac{\mu + \sqrt{\rho} z - \sqrt{1/L + \rho (L-1) / L} \, z_{\alpha}}{\sqrt{(1-\rho)/L}},
\end{gather*}
where $\phi(x)$ is the standard Gaussian density evaluated and $x$. $\pr(\phi_{T^{(1)}} \neq \phi_{T^{(2)}})$ is given by the same expression with $L=1$.

\section{Tail bound for the average statistic} \label{app:aggregation}
Let $T^{(1)}, \dots, T^{(L)}$ be exchangeable statistics with $T^{(1)} \sim F_0$ under the null.
In this Appendix, we consider aggregation rules of the form 
\begin{equation} \label{eqs:avg-S}
S := c \, \bar{T} = c \, (T^{(1)} + \dots + T^{(L)}) / L,
\end{equation}
that guarantee finite sample type-I error control when a test is constructed by comparing $S$ to the $\alpha$ quantile of $F_0$. No assumption is made on the copula of $(T^{(1)}, \dots, T^{(L)})$. We specifically consider two cases: $F_0 = \unif(0,1)$ and $F_0 = \N(0,1)$. We summarise results from the literature on the former case, and also develop new results on the latter case. 

\subsection{The $F_0 = \unif(0,1)$ case}
It is well-known that when each $T^{(l)}$ is a valid p-value, \cref{eqs:avg-S} with $c=2$ is also a valid p-value; see \citet{ruschendorf1982random,meng1994posterior,vovk2021values}. It has been recently shown that the constant $c=2$ cannot be improved under an exchangeability condition on $(T^{(1)}, \dots, T^{(L)})$ \citep{choi2022avg}.  

\subsection{The $F_0 = \N(0,1)$ case}
Without loss of generality, suppose the null hypothesis is rejected whenever $S > \Phi(1-\alpha)$. 

\begin{theorem}[Averaging for $Z$-statistics] \label{thm:avg-Z}
Let $T^{(1)}, \dots, T^{(L)}$ be exchangeable statistics with $T^{(1)} \sim \N(0,1)$ under the null. Then, for $0 < \alpha \leq 1 - \Phi(1) \approx 0.159$, rejecting the null whenever $S = \bar{T} / 2 > \Phi(1-\alpha)$ controls the type-I error below $\alpha$. 
\end{theorem}

We will show \cref{thm:avg-Z} using properties of convex ordering on random variables. 
\begin{definition}[Convex order]
For random variables $X$ and $Y$, we say $X \leqcx Y$ if $\E f(X) \leq \E f(Y)$ for every convex $f$.
\end{definition}

For a random variable $X$, let its integrated survival function be defined as
\[ \pi_X(t) = \int_{t}^{\infty} \bar{F}_X(s) \dd s =\int_{t}^{\infty} (1 - F_X(s)) \dd s = \E (X-t)_{+}. \]
It can be shown that  $\pi_X(t)$ is decreasing and convex. 

\begin{lemma}[{\citet[\S1.5]{muller2002comparison}}] \label{lem:cx-survival}
$X \leqcx Y$ if and only if $\pi_X(t) \leq \pi_Y(t)$ for every $t$ and $\E X = \E Y$. 
\end{lemma}

\begin{lemma}[Mills' ratio] \label{lem:mills}
For $\lambda>0$, $\frac{\lambda}{\lambda^2+1} \phi(\lambda) < 1 - \Phi(\lambda) < \frac{1}{\lambda} \phi(\lambda)$. 
\end{lemma}

The next lemma derives a tail bound for a random variable dominated by a standard Gaussian in the convex order; see \citet[Lemma 1]{meng1994posterior} for a similar result on convex order with respect to a $\unif(0,1)$ random variable. 
\begin{lemma} \label{lem:tail-cx-normal}
Suppose $\bar{Z} \leqcx Z$ for $Z \sim \N(0,1)$. Then we have
\[ \mathbb{P}(\bar{Z} > s) = \bar{F}_{\bar{Z}}(s) < \frac{1 - \Phi(\beta)}{\beta (s-\beta)}, \quad 0 < \beta < s. \]
\end{lemma}
\begin{proof}
Because $\pi_{\bar{Z}}$ is convex, we have
\[ \pi_{\bar{Z}}(\beta) \geq \pi(s) + \pi_{\bar{Z}}'(s) (\beta -s) = \pi_{\bar{Z}}(s) + \bar{F}_{\bar{Z}}(s) (s-\beta), \]
which gives
\[ \bar{F}_{\bar{Z}}(s) \leq \frac{\pi_{\bar{Z}}(\beta) - \pi_{\bar{Z}}(s)}{s-\beta}, \quad \beta < s.\]
By $\bar{Z} \leqcx Z$ and \cref{lem:cx-survival}, we have $\pi_{\bar{Z}}(\beta) \leq \pi_{Z}(\beta)$ and it follows that
\[ \bar{F}_{\bar{Z}}(s) \leq \frac{\pi_{Z}(\beta) - \pi_{\bar{Z}}(s)}{s-\beta}, \quad \beta < s.\]
Now we bound $\pi_{Z}(\beta)$ for $\beta > 0$. Using \cref{lem:mills}, we have
\begin{equation*}
\begin{split}
\pi_{Z}(\beta) < \int_{\beta}^{+\infty} t \phi(t) \dd t = \int_{\beta}^{+\infty} t \dd \Phi(t) &= [t \Phi(t)]|_{\beta}^{+\infty} - \int_{\beta}^{+\infty} \Phi(t) \dd (1/t) \\
&= - \Phi(\beta) / \beta + \int_{\beta}^{+\infty} \Phi(t) t^{-2} \dd t \\
&\leq - \Phi(\beta) / \beta + \int_{\beta}^{+\infty} t^{-2} \dd t = \frac{1-\Phi(\beta)}{\beta}.
\end{split}
\end{equation*}
Hence, for $0 < \beta < s$, we have
\[ \bar{F}_{\bar{Z}}(s) < \frac{(1-\Phi(\beta)) / \beta - \pi_{\bar{Z}}(s)}{s-\beta} \leq \frac{1 - \Phi(\beta)}{\beta (s-\beta)}.\]
\end{proof}

\begin{proof}[of \cref{thm:avg-Z}]
We first claim that $\bar{T} \leqcx T^{(1)} \sim \N(0,1)$. Indeed, given any convex function $f$, we have 
\[ \E f(\bar{T}) \leq \E L^{-1} \sum_{l=1}^L f(T^{(l)}) = \E f(T^{(1)}) \]
if $\E f(T^{(1)}) < \infty$, where we used Jensen's inequality and exchangeability. 
Now, take $\beta = s/2$ for $s \geq 2$ in \cref{lem:tail-cx-normal}. We have 
\begin{equation*}
\mathbb{P}(\bar{T} / 2 > s / 2) = \mathbb{P}(\bar{T} > s ) < \frac{1 - \Phi(s/2)}{(s/2)^2} \leq 1 - \Phi(s/2). 
\end{equation*}
\end{proof}

From our experience, the aggregation rule \cref{eqs:avg-S} with $c=1/2$ is often quite conservative; see, for example, \cref{fig:power-KR}. It is sometimes argued that using $c=1$ seems to work in practice \citep{wang2020debiased,lundborg2022projected}. \cref{prop:2-alpha} below shows that in certain cases using $c=1$ can lead to type-I error inflated by a factor of two. Further, \cref{prop:tail-cx-normal} shows that, when $\alpha$ is small, the actual type-I error is at most $e \alpha$. However, determining the optimal $c$, which lies between $1/2$ and $1$, remains an open question.

\begin{proposition} \label{prop:2-alpha}
Under the assumption of \cref{thm:avg-Z}, using $S = \bar{T}$ can lead to size $2 \alpha$ when $L$ is even.  
\end{proposition}
\begin{proof}
Let $0 \leq \alpha \leq 1/2$ and consider \cref{fig:UZ}. Let $U_1 \sim \unif(0,1)$ and let 
\[ U_2 = \begin{cases} U_1, \quad &U_1 \leq 1 - 2 \alpha \\ 2 - 2\alpha - U_1, \quad & U_1 > 1- 2\alpha \end{cases}.\]
Pair $U_1, U_2$ are exchangeable with uniform margin.
Then, let $T^{(1)} = \Phi^{-1}(U_1)$ and $T^{(2)} = \Phi^{-1}(U_2)$. It follows that $T^{(1)}, T^{(2)}$ are exchangeable with standard normal margin. Further, with probability $2\alpha$, the graph of $(T^{(1)}, T^{(2)})$ is the red curve, on which the minimum value of $\bar{T}$ is $\Phi^{-1}(1-\alpha)$, when $U_1 = U_2 = 1-\alpha$. Hence, 
\[ \pr(\bar{T} > \Phi^{-1}(1-\alpha)) = 2 \alpha, \]
or equivalently, 
\[ \pr(\bar{T} > t) = 2(1 - \Phi(t)). \]
The construction is specific to $t$ for $t > 0$. 

It is easy to extend this construction for $L=2k$, $k \geq 1$. To construct $T^{(1)}, \dots, T^{(2k)}$, randomly choose half of them to be equal to $\Phi^{-1}(U_1)$ and the rest to be equal to  $\Phi^{-1}(U_2)$. Observe that $T^{(1)}, \dots, T^{(2k)}$ are exchangeable with standard normal margin. Further, $\bar{T} = (\Phi^{-1}(U_1) + \Phi^{-1}(U_2))/2$ so the bound remains valid. 
\end{proof}

\begin{proposition}[Tight tail bound up to a constant] \label{prop:tail-cx-normal}
Given $\bar{Z} \leqcx \N(0,1)$, it holds that
\[ \pr(\bar{Z} > s) < \frac{1 - \Phi(s-1/s)}{1-1/s^2}, \quad s>1, \]
where $\mathrm{RHS} \sim e (1 - \Phi(s))$ as $s \rightarrow \infty$. 
\end{proposition}
\begin{proof}
For $s>1$, choose $\beta = s - 1/s > 0$ in \cref{lem:tail-cx-normal} and we get
\[ \pr(\bar{Z} > s) <  \frac{1 - \Phi(s-1/s)}{1-1/s^2}. \] 
By L'H\^opital's rule, it is easy to see that 
\[ \lim_{s \rightarrow +\infty} \frac{1 - \Phi(s-1/s)}{1-1/s^2} / (1 - \Phi(s)) = \lim_{s \rightarrow +\infty} \frac{\phi(s-1/s)}{\phi(s)} = e. \]
\end{proof}

\begin{remark}
This means for small $\alpha$, comparing $\bar{Z}$ to a standard normal inflates size by $e$. 
\end{remark}

\begin{figure}[!htb]
\centering
\includegraphics[width=0.7\textwidth]{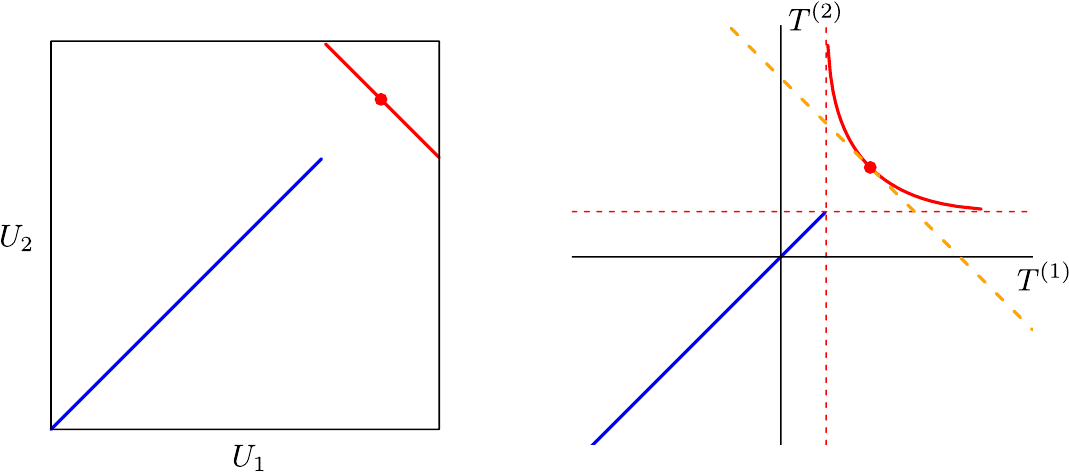}
\caption{Construction for exchangeable $(T^{(1)},T^{(2)})$ with standard normal margin such that $\pr(\bar{T} > t) = 2(1-\Phi(t))$.}
\label{fig:UZ}
\end{figure}

\section{Proofs of theoretical results} \label{sec:app:proofs}
In this Appendix, we prove results for a generic $m=m(n)$ sequence such that $1<m<n$, $m \nearrow \infty$ and $n/m \rightarrow \infty$, although we choose $m = \floor{n / \log n}$ throughout the paper. We write $B_n$ in lieu of $B$ to highlight its dependence on $n$.
\subsection{Finite-sample bounds of subsampling} \label{sec:app:subsample}
The standard consistency result is given by \citet[Theorem 2.2.1]{politis1999subsampling}. The following is an extension that allows multiple exchangeable statistics. Throughout, let $\hat{F}_n := \mathbb{F}_{\hat{\mathbf{H}}}$, where $\hat{\mathbf{H}}$ is given by \cref{eqs:H-mat}. Let $F_{n,P}$ be the distribution function of $T_n^{(1)}$.

\begin{lemma}[Uniform Hoeffding bound] \label{lem:hoeffding}
Let $T_n^{(1)}, \dots, T_n^{(L)}$ be exchangeable statistics under $(X, \Omega) \sim P^n \times P_{\Omega}$.  
We have $\E \hat{F}_{n}(x) = F_{m,P}(x)$ for every $x$ and 
\[ \pr(\|\hat{F}_{n} - F_{m,P}\|_{\infty} > t) \leq 2 \exp(-2 \floor{n/m} t^2 / c_0^2), \quad t \geq 0, \] 
where $c_0$ is a universal constant. 
\end{lemma}
\begin{proof}
By the construction of $\mathcal{B}$ via \cref{alg:gen-tuple} and noting that $|\mathcal{B}| = B = \JJ \floor{n/m}$, we have
\begin{equation} \label{eqs:proof-subsample-F}
\hat{F}_n(x) - F_{m,P}(x) = \frac{1}{\JJ L} \sum_{l=1}^L \sum_{\pi \in \{\pi_1, \dots, \pi_{\JJ} \}} \left[\hat{F}_{\pi,n}^{(l)}(x) - F_{m,P}(x) \right],
\end{equation}
where 
\begin{multline*}
 \hat{F}_{\pi,n}^{(l)}(x) := \frac{1}{\floor{n/m}} \biggl[ \I\left\{ T_{m}^{(l)}(X_{\pi(1)}, \dots, X_{\pi(m)}) \leq x \right\} + \\
 \I\left\{ T_{m}^{(l)}(X_{\pi(m+1)}, \dots, X_{\pi(2m)}) \leq x \right\} + \dots + \I\left\{ T_{m}^{(l)}(X_{\pi((\floor{n/m}-1) m + 1)}, \dots, X_{\pi(\floor{n/m} m)}) \leq x \right\} \biggr].
\end{multline*}
The summands in the square brackets correspond to $\floor{n/m}$ disjoint subsamples of data
\[(X_{\pi(1)}, \dots, X_{\pi(m)}), \quad (X_{\pi(m+1)}, \dots, X_{\pi(2m)}), \quad \dots \quad, (X_{\pi((\floor{n/m}-1) m + 1)}, \dots, X_{\pi(\floor{n/m} m)})\]
and they are therefore independent.

Fix $\pi$ and $l$. By definition, we know $\E \hat{F}_{\pi,n}^{(l)}(x) = F_{m,P}(x)$ and hence $\E \hat{F}_n(x) = F_{m,P}(x)$. In view of $\hat{F}_{\pi,n}^{(l)}$ as an empirical distribution function, the inequality due to \citet{dvoretzky1956asymptotic} and \citet{massart1990tight} states that
\[ \pr\left(\|\hat{F}_{\pi,n}^{(l)} - F_{m,P} \|_{\infty} > t \right) \leq 2 \exp\left(-2 \floor{n/m} t^2 \right), \quad t \geq 0.  \]
This inequality does not require $F_{m,P}$ to be continuous; see \citet[Comment 2(iii)]{massart1990tight}.
The inequality implies that $\|\hat{F}_{\pi,n}^{(l)} - F_{m,P} \|_{\infty}$ is $c_0 / (2 \sqrt{\floor{n/m}})$-sub-Gaussian for a universal constant $c_0 > 0$ \citep[Proposition 2.5.2]{vershynin2018high}. From \cref{eqs:proof-subsample-F}, it follows that
\[ \|\hat{F}_n - F_{m,P}\|_{\infty} \leq \frac{1}{\JJ L} \sum_{l=1}^L \sum_{\pi \in \{\pi_1, \dots, \pi_{\JJ} \}} \|\hat{F}_{\pi,n}^{(l)} - F_{m,P} \|_{\infty}, \]
where we observe that every $\|\hat{F}_{\pi,n}^{(l)} - F_{m,P} \|_{\infty}$ is identically distributed. By Jensen's inequality, we have that the upper bound above is also $c_0 / (2 \sqrt{\floor{n/m}})$-sub-Gaussian and hence so is $\|\hat{F}_n - F_{m,P}\|_{\infty}$, from which the result follows. 
\end{proof}

Recall
that $\hat{G}_n(x) := \mathbb{F}_{\{\hat{S}_b\}}(x)$ is the natural subsampling estimation of $G_P$, the asymptotic distribution function of $S_n$. For $b=1,\dots,B$, $\hat{S}_b$ is the result of applying $S$ to the $b$-th row of $\hat{\mathbf{H}}$. Let $G_{n,P}$ be the distribution function of $S_n$. We define
\begin{equation} \label{eqs:K}
K_n := \sqrt{\frac{\floor{n/m}}{\log (n/m)}},
\end{equation}
which is a positive sequence that tends to infinity. 

\begin{lemma}[High probability bound on subsampling approximation] \label{lem:bound-unif-approx}
Let $T_n^{(1)}, \dots, T_n^{(L)}$ be exchangeable statistics under $(X, \Omega) \sim P^n \times P_{\Omega}$.
Given any distribution functions $F$ and $G$, with probability at least $1 - 4 (n/m)^{-2/c_0^2}$ (where $c_0$ is the universal constant from \cref{lem:hoeffding}), it holds that 
\begin{align*}
\| \hat{F}_{n} - F \|_{\infty} &\leq 1/K_n + \|F_{m,P} - F\|_{\infty}, \\
\| \hat{G}_{n} - G \|_{\infty} &\leq 1/K_n + \|G_{m,P} - G \|_{\infty}.
\end{align*}
\end{lemma}
\begin{proof}
Taking $t=1/K_n$ in \cref{lem:hoeffding}, by the triangle inequality, the first inequality above holds with probability at least $1 - 2 (n/m)^{-2/c_0^2}$. By the definition of $\hat{G}_n$, a tail bound of the same type as \cref{lem:hoeffding} holds for $\|\hat{G}_n - G_{m,P}\|_{\infty}$. Hence, the second inequality above also holds with probability at least $1 - 2 (n/m)^{-2/c_0^2}$. The result follows from a union bound. 
\end{proof}

\begin{corollary} \label{cor:sub-consistency}
Let $T_n^{(1)}, \dots, T_n^{(L)}$ be exchangeable statistics under $(X, \Omega) \sim P^n \times P_{\Omega}$ for $P \in \mathcal{P}_0$.
Under \cref{cond:pivotal,assump:stable-G}, we have 
\[ \| \hat{F}_{n} - F_0 \|_{\infty} \rightarrow_p 0, \quad \| \hat{G}_{n} - G_{P} \|_{\infty} \rightarrow_p 0 \]
as $n \rightarrow \infty$. 
\end{corollary}
\begin{proof}
Take $F=F_0$ and $G=G_P$ in \cref{lem:bound-unif-approx}. Because $F_0$ and $G_P$ are continuous, by \citet[Lem.~2.11]{van2000asymptotic} \cref{cond:pivotal} implies $\|F_{m,P} - F_0\|_{\infty} \rightarrow 0$ and \cref{assump:stable-G} implies $\|G_{m,P} - G_P\|_{\infty} \rightarrow 0$. Therefore, with probability tending to one, the bounds in \cref{lem:bound-unif-approx} tend to zero, which yields the result. 
\end{proof}

\subsection{Rank transform} \label{sec:app:rank}
In this section, we show that the rank transform does not affect the validity of subsampling approximation under the null. Recall that $\tilde{G}_n$ is the empirical distribution of $\{\tilde{S}_b\}$, where $\tilde{S}_b$ is the result of applying aggregation $S$ to the $b$-th row of the rank-transformed matrix $\tilde{\mathbf{H}}$. 

\begin{lemma} \label{lem:empirical-bound}
Given $x_1, \dots, x_n \in \R$, $x_1', \dots, x_n' \in \R$ and $\varepsilon > 0$ such that $|x_i - x_i'| \leq \epsilon$ for all $i$, for every $x \in R$ it holds that 
\[ |\mathbb{F}_{\{x_i\}}(x) - \mathbb{F}_{\{x_i'\}}(x)| \leq \mathbb{F}_{\{x_i\}}(x + \epsilon) - \mathbb{F}_{\{x_i\}}(x - \epsilon).\]
\end{lemma}
\begin{proof}
As $|\mathbb{F}_{\{x_i\}}(x) - \mathbb{F}_{\{x_i'\}}(x)| = n^{-1} \sum_{i=1}^n (\I\{x_i \leq x \} - \I\{x_i' \leq x \})$, we know $|\mathbb{F}_{\{x_i\}}(x) - \mathbb{F}_{\{x_i'\}}(x)|$ is bounded above by the fraction of $i$ such that $x$ is between $x_i$ and $x_i'$. Due to the fact that $|x_i - x_i'| \leq \epsilon$, this fraction is further bounded above by the fraction of $x_i$ such that $x \in [x_i-\epsilon, x_i+\epsilon)$, or equivalently, $x_i \in (x-\epsilon, x+\epsilon]$, which is the right-hand side of the display above.
\end{proof}

\begin{lemma}[High probability bound on rank-transformed subsampling] \label{lem:rank-approx}
Let $T_n^{(1)}, \dots, T_n^{(L)}$ be exchangeable statistics under $(X, \Omega) \sim P^n \times P_{\Omega}$ for $P \in \mathcal{P}_0$.
Suppose $S$ is chosen such that \cref{cond:lip} holds.
Let 
\begin{align*}
& \epsilon^{F,P}_{n} := 1/K_n + \|F_{m,P} - F_0\|_{\infty} + 1/(2\,B_n\,L), \\
& \epsilon^{G,P}_{n} := 1/K_n + \|G_{m,P} - G_{P}\|_{\infty},
\end{align*}
where $B_n = \JJ \floor{n/m}$.
Then, under \cref{cond:pivotal,assump:stable-G}, we have the following results.
\begin{enumerate}[(1)]
\item For $F_0 = \unif(0,1)$, with probability tending to one, it holds that
\[ \|\tilde{G}_{n} - G_P\|_{\infty} \leq 2 \, g_{\max,P}\, \epsilon^{F,P}_n + 3 \,\epsilon^{G,P}_n. \]

\item For $F_0 = \N(0,1)$, with probability tending to one, it holds that
\[ \|\tilde{G}_{n} - G_P\|_{\infty} \leq \frac{2 \, g_{\max,P}}{|\Phi^{-1}(\epsilon^{F,P}_n)|} + 4 L \,\epsilon^{F,P}_n  + 3 \, \epsilon^{G,P}_n + \frac{1}{B_n}. \]
\end{enumerate}
\end{lemma}
\begin{proof}
In what follows, let $\mathcal{A}_n$ be the event that the bounds in \cref{lem:bound-unif-approx} hold with $F = F_0$ and $G = G_P$. Note that $\pr_{P}(\mathcal{A}_n) \rightarrow 1$.
\begin{enumerate}[(1)]
\item Under $F_0 = \unif(0,1)$, the rank transform \cref{eqs:rank} becomes
\[ \tilde{H}_{b,l} = \hat{F}_{n}(\hat{H}_{b,l}) - \frac{1}{2 B_n L}, \quad b=1,\dots,B_n,\quad l=1,\dots,L. \]
On $\mathcal{A}_n$, we have 
\[ \sup_{x \in [0,1]} |\hat{F}_n(x) - x| \leq 1/K_n + \|F_{m,P} - F_0\|_{\infty} \]
 and hence, by the rank transform,
\[ \max_{b,l} |\tilde{H}_{b,l} - \hat{H}_{b,l} | \leq \epsilon^{F,P}_n. \]
By \cref{cond:lip}, it implies
\[ \max_{b} |\tilde{S}_{b} - \hat{S}_b | \leq \epsilon^{F,P}_n.\]
Recall that on $\mathcal{A}_n$, $ \| \hat{G}_{n} - G_P \|_{\infty} \leq \epsilon^{G,P}_n$. Note that $\tilde{G}_n$ and $\hat{G}_{n}$ are respectively the distribution functions of $\{\tilde{S}_b\}$ and $\{\hat{S}_b\}$. Applying \cref{lem:empirical-bound}, we have that for every $x$,
\begin{align*}
|\tilde{G}_{n}(x) - \hat{G}_{n}(x)| &\leq \hat{G}_{n}(x + \epsilon^{F,P}_n) - \hat{G}_{n}(x - \epsilon^{F,P}_n) \\
& \leq G_P(x + \epsilon^{F,P}_n) - G_P(x - \epsilon^{F,P}_n) + 2 \epsilon^{G,P}_n \\
& \leq  2 g_{\max,P}\, \epsilon^{F,P}_n + 2\,\epsilon^{G,P}_n \quad \text{on } \mathcal{A}_n,
\end{align*}
where the last step follows from \cref{assump:stable-G} and the mean value theorem. Note that the final bound above does not depend on $x$. Consequently, on $\mathcal{A}_n$, we have 
\begin{align*}
 \|\tilde{G}_{n} - G_P\|_{\infty}  &\leq \|\tilde{G}_{n} - \hat{G}_{n}\|_{\infty} +  \|\hat{G}_{n} - G_P \|_{\infty} \\
 & \leq  2 \, g_{\max,P}\, \epsilon^{F,P}_n + 3 \, \epsilon^{G,P}_n,
\end{align*}
which occurs with probability tending to one.
\item Under $F_0 = \Phi$, the rank transform in \cref{eqs:rank} becomes
\[ \tilde{H}_{b,l} = \Phi^{-1}\left(\hat{F}_n(\hat{H}_{b,l}) - 1/(2 B_n L) \right). \]
By a Taylor expansion of the outer function at $\Phi(\hat{H}_{b,l})$, 
\begin{equation} \label{eqs:pf-taylor}
\tilde{H}_{b,l} =  \hat{H}_{b,l} + \frac{1}{\phi(\Phi^{-1}(\xi))} \left[\hat{F}_n(\hat{H}_{b,l}) - 1/(2B_n L) - \Phi(\hat{H}_{b,l})  \right]
\end{equation}
for some $\xi$ between $\hat{F}_n(\hat{H}_{b,l}) - 1/(2 B_n L)$ and $\Phi(\hat{H}_{b,l})$. Let
\[ \Gamma_{\epsilon} := \left\{(b,l): \hat{F}_n(\hat{H}_{b,l}) - 1/(2 B_n L) \in [0, 2 \epsilon^{F,P}_n) \cup (1-2 \epsilon^{F,P}_n, 1]  \right\}. \]
Because $\{\hat{F}_n(\hat{H}_{b,l}) \}$ is the set of normalised ranks, we have 
\begin{equation} \label{eqs:pf-size-bound}
|\Gamma_{\epsilon}| \leq 4 \epsilon^{F,P}_n B_n L + 1.
\end{equation}

Observe that $ \| \hat{F}_n - 1/(2B_n L) - \Phi \|_{\infty} \leq \epsilon^{F,P}_n$ on $\mathcal{A}_n$. Then on $\mathcal{A}_n$, for every $(b,l) \in \Gamma_{\epsilon}^{c}$, we know $\xi$ in \cref{eqs:pf-taylor} is in $[\epsilon^{F,P}_n,1-\epsilon^{F,P}_n]$. By Mills' ratio (\cref{lem:mills}), we have
\[
\phi(\Phi^{-1}(\xi)) \geq \phi(\Phi^{-1}(1-\varepsilon_{F,P})) > \epsilon^{F,P}_n |\Phi^{-1}(\epsilon^{F,P}_n)|.
\]
It then follows from \cref{eqs:pf-taylor} that, on $\mathcal{A}_n$, 
\begin{equation} \label{eqs:pf-T-close}
|\tilde{H}_{b,l} - \hat{H}_{b,l}| \leq \frac{1}{|\Phi^{-1}(\epsilon^{F,P}_n)|}, \quad (b,l) \in \Gamma_{\epsilon}^c.
\end{equation}
Define
\[ \Xi_{\epsilon} := \left\{b: (b, l) \in \Gamma_\epsilon \text{ for some } l \in \{1,\dots,L\} \right\}. \]
By \cref{cond:lip}, on $\mathcal{A}_n$, \cref{eqs:pf-T-close} implies
\begin{equation} \label{eqs:pf-delta}
 |\tilde{S}_{b} - \hat{S}_b| \leq \delta_n := \frac{1}{|\Phi^{-1}(\epsilon^{F,P}_n)|}, \quad b \in \Xi_{\epsilon}^c.
\end{equation}
In light of $|\Xi_{\epsilon}| \leq |\Gamma_{\epsilon}|$ and \cref{eqs:pf-size-bound}, observe that
\begin{equation} \label{eqs:pf-fraction-bound}
\rho_{\epsilon} := |\Xi_{\epsilon}| / B_n \leq 4 \epsilon^{F,P}_n L + 1/B_n. 
\end{equation}
We now bound the distance between the empirical distributions of the rank-transformed and untransformed statistics. We have
\begin{equation} \label{eqs:pf-bound-emp-1}
\begin{split}
\|\tilde{G}_{n} - \hat{G}_{n}\|_{\infty} &= \left\|\mathbb{F}_{\{\tilde{S}_b\}} - \mathbb{F}_{\{\hat{S}_b\}} \right\|_{\infty} \\
&= \left\|(1-\rho_{\epsilon}) \mathbb{F}_{\tilde{S}_{\Xi_{\epsilon}^c}} + \rho_{\epsilon} \mathbb{F}_{\tilde{S}_{\Xi_{\epsilon}}} - (1-\rho_{\epsilon}) \mathbb{F}_{\hat{S}_{\Xi_{\epsilon}^c}} - \rho_{\epsilon} \mathbb{F}_{\hat{S}_{\Xi_{\epsilon}}} \right\|_{\infty} \\
& \leq (1-\rho_{\epsilon}) \left\|\mathbb{F}_{\tilde{S}_{\Xi_{\epsilon}^c}} - \mathbb{F}_{\hat{S}_{\Xi_{\epsilon}^c}} \right\|_{\infty} + \rho_{\epsilon} \left\| \mathbb{F}_{\tilde{S}_{\Xi_{\epsilon}}} - \mathbb{F}_{\hat{S}_{\Xi_{\epsilon}}}  \right\|_{\infty} \\
& \leq (1-\rho_{\epsilon}) \left\|\mathbb{F}_{\tilde{S}_{\Xi_{\epsilon}^c}} - \mathbb{F}_{\hat{S}_{\Xi_{\epsilon}^c}} \right\|_{\infty} + \rho_{\epsilon},
\end{split}
\end{equation}
where the last step uses the fact that Kolmogorov distance is bounded by one. Given \cref{eqs:pf-delta}, by \cref{lem:empirical-bound}, at any $x$, we have
\begin{equation*}
\begin{split}
\left|\mathbb{F}_{\tilde{S}_{\Xi_{\epsilon}^c}}(x) - \mathbb{F}_{\hat{S}_{\Xi_{\epsilon}^c}}(x) \right| &\leq \frac{\left|\{b \in \Xi_{\epsilon}^c: \hat{S}_b \in(x-\delta_n, x+\delta_n] \right|}{|\Xi_{\epsilon}^c|} \\
& \leq \frac{\left|\{b: \hat{S}_b\in(x-\delta_n, x+\delta_n] \right|}{|\Xi_{\epsilon}^c|} \\
&= \frac{\left|\{b: \hat{S}_b\in(x-\delta_n, x+\delta_n] \right|}{B (1 - \rho_{\epsilon})} \\
&= \frac{1}{1-\rho_{\epsilon}} \left[\hat{G}_{m,n,P}(x+\delta_n) - \hat{G}_{m,n,P}(x-\delta_n) \right],
\end{split}
\end{equation*}
where $\rho_{\epsilon}$ is defined in \cref{eqs:pf-fraction-bound}. Further, on $\mathcal{A}_n$, we also have $\| \hat{G}_{n} - G_P \|_{\infty} \leq \epsilon^{G,P}_n$ and hence 
\[ \left|\mathbb{F}_{\tilde{S}_{\Xi_{\epsilon}^c}}(x) - \mathbb{F}_{\hat{S}_{\Xi_{\epsilon}^c}}(x) \right| \leq \frac{1}{1-\rho_{\epsilon}} \left[ G_P(x+\delta_n) - G_P(x-\delta_n) + 2 \epsilon^{G,P}_n \right]. \]
Taking supremum over $x$ and using \cref{assump:stable-G}, we have
\[ \left\|\mathbb{F}_{\tilde{S}_{\Xi_{\epsilon}^c}} - \mathbb{F}_{\hat{S}_{\Xi_{\epsilon}^c}} \right\|_{\infty} \leq \frac{1}{1-\rho_{\epsilon}}(2 g_{\max,P} \delta_n + 2 \epsilon^{G,P}_n). \]
Hence, by substituting the above, \cref{eqs:pf-delta} and \eqref{eqs:pf-fraction-bound} into \cref{eqs:pf-bound-emp-1}, we arrive at
\begin{align*}
\|\tilde{G}_{n} - \hat{G}_{n}\|_{\infty} &\leq 2 g_{\max,P} \delta_n + 2 \epsilon^{G,P}_n  + 4  L \epsilon^{F,P}_n + 1/B_n\\
&= \frac{2 g_{\max,P}}{|\Phi^{-1}(\epsilon^{F,P}_n)|} + 2 \epsilon^{G,P}_n + 4 L \epsilon^{F,P}_n + \frac{1}{B_n} \quad \text{on } \mathcal{A}_n.
\end{align*}
Finally, combining the above and $\| \hat{G}_{n} - G_P \|_{\infty} \leq \epsilon^{G,P}_n$, we conclude that on $\mathcal{A}_n$,
\[ \|\tilde{G}_{n} - G_{P}\|_{\infty} \leq \frac{2 g_{\max,P}}{|\Phi^{-1}(\epsilon^{F,P}_n)|} + 4 L \epsilon^{F,P}_n + 3 \epsilon^{G,P}_n + 1/B_n, \]
whose probability tends to one.
\end{enumerate}
\end{proof}

\begin{remark}
The normal case needs a separate proof because when $S$ is Lipschitz, $(t_1, \dots, t_L) \mapsto S(\Phi^{-1}(t_1),\dots,\Phi^{-1}(t_L))$ need not be Lipschitz.
\end{remark}

\subsection{Proof of \cref{thm:main-null}}
\begin{proof}
Observe that with $m = \floor{n / \log n}$ and $B_n = \JJ \floor{n/m}$, the bounds on $\|\tilde{G}_{n} - G_P\|_{\infty}$ in \cref{lem:rank-approx} tend to zero as $n \rightarrow \infty$. We have $\|\tilde{G}_{n} - G_P\|_{\infty} \rightarrow_p 0$. 

\begin{enumerate}[(i)]
\item This follows from the proof of (ii) by taking $\mathcal{P}_0 = \{P\}$. 

\item Note first that by \cref{lem:G_P_unif} and \citet[Lem.~2.11]{van2000asymptotic}, we have
\[
\sup_{P \in \mathcal{P}_0} \|F_{n,P} - F_0 \|_\infty \rightarrow 0, \quad \sup_{P \in \mathcal{P}_0} \|G_{n,P}- G_P \|_\infty \rightarrow 0.
\]
By \cref{lem:rank-approx} and the given assumptions, there exists a sequence $\gamma_n \rightarrow 0$ that does not depend on $P$, such that 
\[ \sup_{P \in \mathcal{P}_0} \epsilon^{F,P}_n \leq \gamma_n, \quad \sup_{P \in \mathcal{P}_0} \epsilon^{G,P}_n \leq \gamma_n, \]
where $\epsilon^{F,P}_n$ and $\epsilon^{G,P}_n$ are defined in \cref{lem:rank-approx}. Also, because the probability in \cref{lem:bound-unif-approx} does not depend on $P$, there exists sequences $\lambda_n \rightarrow 0$ and $\delta_n \rightarrow 0$, which do not depend on $P$, such that for every $P \in \mathcal{P}_0$,
\[ \mathcal{D}_{n,P} := \{\|\tilde{G}_n - G_P\|_{\infty} < \lambda_n\} \]
occurs with probability at least $1 - \delta_n$. Let $\mathcal{B}_n := \left\{S_n > \tilde{G}_n^{-1}(1-\alpha) \right\}$. 

We first show that $\limsup_{n} \sup_{P \in \mathcal{P}_0} \mathbb{P}_{P}(\mathcal{B}_n) \leq \alpha$. 
Observe that 
\[ 1-\alpha \leq \tilde{G}_n\left(\tilde{G}_n^{-1}(1-\alpha)\right) \leq G_P\left(\tilde{G}_n^{-1}(1-\alpha)\right) + \|\tilde{G}_n - G_P\|_{\infty} \]
using \cref{lem:df}, 
which implies
\begin{equation}
G_P\left(\tilde{G}_n^{-1}(1-\alpha)\right) \geq 1 - \alpha - \|\tilde{G}_n - G_P\|_{\infty}.
\end{equation}
It follows that on $\mathcal{B}_n$,
\[ G_P(S_n) \geq G_P\left(\tilde{G}_n^{-1}(1-\alpha) \right) \geq 1 - \alpha - \|\tilde{G}_n - G_P\|_{\infty}.\]
Hence on $ \mathcal{B}_n \cap \mathcal{D}_n$,
\[ S_n > G_P^{-1}(1-\alpha-\lambda_n). \]
We have
\begin{align*}
\mathbb{P}_P(\mathcal{B}_n) &\leq \mathbb{P}_P(\mathcal{B}_n \cap \mathcal{D}_n) + \mathbb{P}_P(\mathcal{D}_n^c)\\
& \leq \mathbb{P}_P\left\{ S_n > G_P^{-1}(1-\alpha-\lambda_n) \right\} + \delta_n \\
&= 1 - G_{n,P}\left(G_P^{-1}(1-\alpha-\lambda_n) \right) + \delta_n \\
&\leq 1 - G_P\left(G_P^{-1}(1-\alpha-\lambda_n) \right) + \|G_P - G_{n,P}\|_{\infty} + \delta_n \\
\text{(by \cref{lem:df})} \quad & \leq \alpha+\lambda_n+\|G_P - G_{n,P}\|_{\infty} + \delta_n,
\end{align*}
where $G_{n,P}$ is the distribution function of $S_n$. Because $\lambda_n$ and $\delta_n$ do not depend on $P$, it follows that
\[ \sup_{P \in \mathcal{P}_0} \mathbb{P}_P(\mathcal{B}_n) \leq \alpha+\lambda_n+\sup_{P \in \mathcal{P}_0} \|G_P - G_{n,P}\|_{\infty} + \delta_n,\]
and hence $\limsup_{n} \sup_{P \in \mathcal{P}_0} \mathbb{P}_P(\mathcal{B}_n) \leq \alpha$ because $\lambda_n \rightarrow 0$, $\delta_n \rightarrow 0$ and $\sup_{P \in \mathcal{P}_0} \|G_P - G_{n,P}\|_{\infty} \rightarrow 0$.

We now show that $\liminf_n \inf_{P \in \mathcal{P}_0} \mathbb{P}_P(\mathcal{B}_n) \geq \alpha$, which will yield the final conclusion. Let 
\begin{equation} \label{eqs:proof-y}
 y_{n,P}:=\sup\left\{x: G_P(x) \leq 1 - \alpha + \lambda_n + 1/B_n \right\}
\end{equation}
and consider $\mathcal{F}_{n,P} := \{S_n > y_{n,P}\}$. Observe that $\mathcal{F}_{n,P}$ implies $G_P(S_n) - \lambda_n > 1 - \alpha + 1/B$. Then, $\mathcal{F}_{n,P} \cap \mathcal{D}_{n,P}$ implies $\tilde{G}_n(S_n) > 1 - \alpha + 1/B_n$, which further leads to 
\[ S_n > \tilde{G}_n^{-1}(1 - \alpha) \]
because $\tilde{G}_n$ is an empirical measure with $B$ atoms. Hence, we have $\mathcal{F}_{n,P} \cap \mathcal{D}_{n,P} \subseteq \mathcal{B}_{n} \cap \mathcal{D}_{n,P}$ and therefore
\begin{align*}
\mathbb{P}_P(\mathcal{B}_n) &\geq \mathbb{P}_P(\mathcal{B}_n \cap \mathcal{D}_{n,P}) \\
& \geq \mathbb{P}_P(\mathcal{F}_{n,P} \cap \mathcal{D}_{n,P}) \\
& \geq \mathbb{P}_P(\mathcal{F}_{n,P}) - \mathbb{P}_P(\mathcal{D}_{n,P}^c) \\
& = 1 - G_{n,P}(y_{n,P}) - \mathbb{P}_P(\mathcal{D}_{n,P}^c) \\
& \geq 1 - G_{n,P}(y_{n,P}) - \delta_{n} \\
& \geq 1 - G_{P}(y_{n,P}) - \|G_P - G_{n,P}\|_{\infty} - \delta_{n} \\
& \stackrel{(i)}{=} 1 - (1 - \alpha + \lambda_n + 1/B_n) - \|G_P - G_{n,P}\|_{\infty} - \delta_{n} \\
& =  \alpha - \lambda_n - 1/B_n - \|G_P - G_{n,P}\|_{\infty} - \delta_{n},
\end{align*}
where (i) uses \cref{eqs:proof-y} and the continuity of $G_P$. Because sequences $\lambda_n$, $B$ and $\delta_n$ do not depend on $P$, it then follows that 
\[ \inf_{P \in \mathcal{P}_0} \mathbb{P}_P(\mathcal{B}_n) \geq  \alpha - \lambda_n - 1/B - \sup_{P \in \mathcal{P}_0} \|G_P - G_{n,P}\|_{\infty} - \epsilon_{n}, \]
and hence $\liminf_{n} \inf_{P \in \mathcal{P}_0} \mathbb{P}_P(\mathcal{B}_n) \geq \alpha$ because $\lambda_n \rightarrow 0$, $B \rightarrow \infty$, $\delta_n \rightarrow 0$ and $\sup_{P \in \mathcal{P}_0} \|G_P - G_{n,P}\|_{\infty} \rightarrow 0$.
\end{enumerate}
\end{proof}

\begin{remark}
The proof above needs no assumption on $G_P$ being strictly increasing at $G_P^{-1}(1-\alpha)$; see also \citet[Remark 1.2.1]{politis1999subsampling}.
\end{remark}
\subsection{Proof of \cref{thm:main-null-adaptive}}
We first prove the following lemma.
\begin{lemma} \label{lem:adaptive}
Under the assumptions of \cref{thm:main-null-adaptive} and \cref{assump:stable-joint}, it holds that $\|\tilde{Q}_n - Q_P\|_{\infty} \rightarrow_p 0$, where $Q_P$ is the distribution function of $\max\{G_P^1(S^1), \dots, G_P^W(S^W)\}$.
\end{lemma}
\begin{proof}
We show the result by relating $R_n$ to 
\[ R_n^{0} := \max(G_{P}^{1}(S_n^{1}), \dots, G_{P}^{W}(S_n^{1})) \]
and relating $\tilde{Q}_n := \mathbb{F}_{\{\tilde{R}_b\}}$ to the subsampling estimate $\hat{Q}_n := \mathbb{F}_{\{\hat{R}_b\}}$ with $\{\hat{R}_b\}$ computed as follows. Applying $S^1, \dots, S^W$ to the $b$-th row of the untransformed matrix $\hat{\mathbf{H}}$, we get 
\[ \hat{S}_b^{1}:=S^1(\hat{H}_{b,1}, \dots,\hat{H}_{b,L}), \quad \dots, \quad \hat{S}_b^{W}:=S^W(\hat{H}_{b,1}, \dots,\hat{H}_{b,L}).  \]
Recall that for $w=1,\dots,W$, $G_P^w$ is the asymptotic distribution function of $S_n^{w}$. We define 
\[ \hat{R}_b := \max\left\{G_P^{1}(\hat{S}_b^1), \dots, G_P^{W}(\hat{S}_b^W) \right\}, \quad b=1,\dots,B_n; \]
compare this with the definition of $\tilde{R}_b$ in \cref{alg:agg-test-adaptive}. 

First note that $Q_P$ has a density that is bounded by $W$. Indeed, fix $1 \geq t > u \geq 0$ and set $\mathcal{T}_w := \{G_P^w(S_n^w) \leq t\}$ and $\mathcal{U}_w:= \{G_P^w(S_n^w) \leq u\}$ for $w=1,\ldots,W$. Then
\begin{align*}
	Q_P(t) - Q_P(u) &= \pr_P\left(\cap_w \mathcal{T}_w \right) - \pr_P( \cap_w \mathcal{U}_w) \\
	&\leq \pr_P \left\{\left(\cap_w \mathcal{T}_w\right) \setminus \left(\cap_w \mathcal{U}_w\right) \right\}.
\end{align*}
Now
\begin{align*}
	\left(\cap_w \mathcal{T}_w\right) \setminus \left(\cap_w \mathcal{U}_w\right)&= \cap_w \mathcal{T}_w \cap \left(\cap_w  \mathcal{U}_w\right)^c \\
	&= \cap_w \mathcal{T}_w  \cap \left(\cup_w  \mathcal{U}_w^c\right) \\
	&= \cup_w \left(\cap_w \mathcal{T}_w  \cap  \mathcal{U}_w^c \right) \\
	&\subseteq \cup_w \left\{u < G_P^w(S_n^w) \leq t \right\}.
\end{align*}
Thus by a union bound,
\[
Q_P(t) - Q_P(u) \leq W(t-u),
\]
proving the claim.

By the joint stability in \cref{assump:stable-joint}, the distribution of $R^0_n$ converges to $Q_P$. By definition, $\hat{Q}_n$ is the natural subsampling estimate of $Q_P$. Let 
\[ \mathcal{Q}_{n,P} := \left\{\|\hat{Q}_n - Q_P\|_{\infty} \leq 1/K_n + \|Q_{n,P} - Q_{P}\|_{\infty} \right\}, \]
where $K_n$ is defined in \cref{eqs:K} and $Q_{n,P}$ is the distribution function of $R_n^0$. By \cref{assump:stable-joint}, $\|Q_{n,P} - Q_P\|_{\infty} \rightarrow 0$. By an argument similar to \cref{lem:bound-unif-approx}, we know $\mathbb{P}_P(\mathcal{Q}_{n,P}) \rightarrow 1$. In what follows, let $\mathcal{E}_n^w$ be the event that the corresponding bound in \cref{lem:rank-approx} on $\|\tilde{G}^w_n - G_P^w\|$ holds, namely
\[ \mathcal{E}_{n,P}^w := \left\{ \|\tilde{G}_n^w - G_P^w\|_{\infty} \leq \Delta_{n,P}^w \right \}, \] 
where 
\begin{equation} \label{eqs:proof-Delta}
\Delta_{n,P}^w = \begin{cases} 2 \, g^w_{\max,P}\, \epsilon^{F,P}_n + 3 \,\epsilon^w_{G,P}, \quad & F_0 = \unif(0,1) \\
 2 \, g^w_{\max,P} / |\Phi^{-1}(\epsilon^{F,P}_n)| + 4 L \,\epsilon^{F,P}_n  + 3 \, \epsilon^w_{G,P} + 1/B_n, \quad & F_0 = \N(0,1)
\end{cases}
\end{equation}
with $\epsilon^{F,P}_n$ and $\epsilon^w_{G,P}$ given by \cref{lem:rank-approx}. From \cref{lem:rank-approx}, we know $\mathbb{P}_P(\mathcal{E}_{n,P}^w) \rightarrow 1$. We now prove $\|\tilde{Q}_n - Q_P\|_{\infty} \rightarrow_p 0$.

\medskip We first show the $F_0 = \unif(0,1)$ case. Let $\mathcal{A}_{n,P}$ be the event that the bound in \cref{lem:bound-unif-approx} holds, which implies $\| \hat{F}_{n} - F_0 \|_{\infty} \leq \epsilon^{F,P}_n$; we have $\mathbb{P}_P(\mathcal{A}_{n,P}) \rightarrow 1$. Because \cref{cond:lip} holds for $S^1, \dots, S^W$, we have that on $ \mathcal{A}_{n,P}$,
\[
|\tilde{S}_b^w - \hat{S}_b^w| \leq \epsilon^{F,P}_n, \quad b=1,\dots,B_n, \quad w=1,\dots,W.
\]
Then, we have for $w=1,\dots, W$, on the event $\mathcal{A}_{n,P} \cap \mathcal{E}_{n,P}^w$,
\begin{align*}
|\tilde{G}_n^w(\tilde{S}_b^w) - G_P^w(\hat{S}_b^w)| &\leq |\tilde{G}_n^w(\tilde{S}_b^w) - G_P^w(\tilde{S}_b^w)| + |G_P^w(\tilde{S}_b^w) - G_P^w(\hat{S}_b^w)| \\
& \leq \Delta_{n,P}^w + g_{P,\max}^w \, \epsilon^{F,P}_n, \quad b=1,\dots,B_n.
\end{align*}
Further, on $\mathcal{A}_{n,P} \cap (\cap_w \mathcal{E}_{n,P}^w)$, it holds that 
\[ |\tilde{G}_n^w(\tilde{S}_b^w) - G_P^w(\hat{S}_b^w)| \leq (\max_w \Delta_{n,P}^w) + (\max_w g_{P,\max}^w) \, \epsilon^{F,P}_n, \quad b=1,\dots,B_n, \quad w=1,\dots,W, \]
and, because $\max(\cdot)$ is 1-Lipschitz in $\|\cdot\|_{\infty}$, that
\[ |\tilde{R}_b - \hat{R}_b| \leq (\max_w \Delta_{n,P}^w) + (\max_w g_{P,\max}^w) \, \epsilon^{F,P}_n, \quad b=1,\dots,B_n. \]
Because $\tilde{Q}_n := \mathbb{F}_{\{\tilde{R}_b\}}$ and $\hat{Q}_n := \mathbb{F}_{\{\hat{R}_b\}}$, using \cref{lem:empirical-bound} one can show that on $\mathcal{A}_{n,P} \cap (\cap_w \mathcal{E}_{n,P}^w)$,
\[
\|\tilde{Q}_n - \hat{Q}_n\|_{\infty} \leq 2 q_{\max,P} \left[(\max_w \Delta_{n,P}^w) + (\max_w g_{P,\max}^w) \, \epsilon^{F,P}_n \right] + 2 \|\hat{Q}_n - Q_P\|_{\infty}.
\]
It follows that, on $\mathcal{Q}_{n,P} \cap \mathcal{A}_{n,P} \cap (\cap_w \mathcal{E}_{n,P}^w)$, 
\begin{multline} \label{eqs:proof-bound-Q}
\|\tilde{Q}_n - Q_P\|_{\infty} \leq 2 W \left[(\max_w \Delta_{n,P}^w) + (\max_w g_{P,\max}^w) \, \epsilon^{F,P}_n \right] \\
+ 3\left(1/K_n + \|Q_{n,P} - Q_{P}\|_{\infty}\right). 
\end{multline}
The bound vanishes because
$\|Q_{n,P} - Q_P\|_{\infty} \rightarrow 0$ (by \cref{assump:stable-joint}), $\epsilon^{F,P}_n \rightarrow 0$, $\max_w \epsilon^w_{G,P} \rightarrow 0$ (by \cref{lem:rank-approx}, \cref{cond:pivotal,assump:stable-joint}), $\max_w \Delta_{n,P}^w \rightarrow 0$ (by \cref{eqs:proof-Delta}) and $K \rightarrow \infty$ (by definition). Since $\mathbb{P}_P(\mathcal{Q}_{n,P} \cap \mathcal{A}_{n,P} \cap (\cap_w \mathcal{E}_{n,P}^w)) \rightarrow 1$ by the union bound, we conclude that $\|\tilde{Q}_n - Q_P\|_{\infty} \rightarrow_p 0$.

\medskip The $F_0 = \N(0,1)$ case can be shown by an argument similar to the above and the argument used for proving the  corresponding case in \cref{lem:rank-approx}. 
\end{proof}

\begin{proof}[of \cref{thm:main-null-adaptive}]
We continue to use the notation for proving the last lemma. 
\begin{enumerate}[(i)]
\item This follows from the proof of (ii) by taking $\mathcal{P}_0 = \{P\}$. 

\item 
Note first that by \cref{lem:G_P_unif} and \citet[Lem.~2.11]{van2000asymptotic}, we have
\[
\sup_{P \in \mathcal{P}_0} \|F_{n,P} - F_0 \|_\infty \rightarrow 0, \quad \max_w \sup_{P \in \mathcal{P}_0} \|G^w_{n,P}- G^w_P \|_\infty \rightarrow 0, \quad \sup_{P \in \mathcal{P}_0} \|Q_{n,P} - Q_P \|_\infty \rightarrow 0.
\]
Now fix $P \in \mathcal{P}_0$. Observe that on the event $ \cap_w \mathcal{E}_{n,P}^w$, we have
\[|R_n^0 - R_n| \leq \max_w \Delta_{n,P}^w,\]
where each $\Delta_{n,P}^w$ in \cref{eqs:proof-Delta} tends to zero and $\mathbb{P}_P(\cap_w \mathcal{E}_{n,P}^w ) \rightarrow 1$. 
Now let $P$ vary in $\mathcal{P}_0$. 
Under the given assumptions, observe that there exists a sequence $\delta_n \rightarrow 0$, which does not depend on $P$, such that $\sup_{P \in \mathcal{P}_0} \max_w \Delta_{n,P}^w \leq \delta_n$. Further, because the probability in \cref{lem:bound-unif-approx} does not depend on $P$, also observe that 
\[ \sup_{P \in \mathcal{P}_0} \mathbb{P}_{P} \left\{\cup_{w} (\mathcal{E}_{n,P}^w)^c \right\} \rightarrow 0\]
holds under the given assumptions; see also the proof of \cref{lem:rank-approx}. 
Hence, with $\delta_n \rightarrow 0$ that does not depend on $P$, we have
\[ \sup_{P \in \mathcal{P}_0} \mathbb{P}_{P}\left(|R_n - R_n^0|  > \delta_n \right) \rightarrow 0. \]

Also, observe from the proof of (1), e.g., \cref{eqs:proof-bound-Q} for the case of $F_0 = \unif(0,1)$, that there exists sequence $\lambda_n \rightarrow 0$, which does not depend on $P$, such that 
\[ \sup_{P \in \mathcal{P}_0} \mathbb{P}_{P}\left(\|\tilde{Q}_n - Q_P\|_{\infty} > \lambda_n \right) \rightarrow 0. \]
The result then follows from an argument similar to that for proving (ii) of \cref{thm:main-null}. 
\end{enumerate}
\end{proof}
\subsection{Proofs of power results}
\begin{proposition} \label{prop:most-powerful}
Let $(T_n^{(1)}, \dots, T_n^{(L)})$ be $L$ exchangeable test statistics. Suppose $\phi_n(T_n^{(1)}, \dots, T_n^{(L)})$ is a most powerful at test level $\alpha$ for testing certain $P_0 \in \mathcal{P}_0$ against certain $P_1 \in \mathcal{P} \setminus \mathcal{P}_0$. Then $\phi_n$ equals a symmetric function of $(T_n^{(1)},\dots, T_n^{(n)})$ $P_0$-almost everywhere. 
\end{proposition}
\begin{proof}
Let $p_0$ and $p_1$ be respectively the densities of $P_0$ and $P_1$ with respect to some dominating measure $\mu$. By the Neyman--Pearson lemma \citep[Theorem 3.2.1]{lehmann2005testing}, $\phi_n$ must be equal to $\I\{ p_1(T_n^{(1)}, \dots, T_n^{(L)}) > k p_0(T_n^{(1)}, \dots, T_n^{(L)})\}$ for some $k$, $\mu$-almost everywhere. By exchangeability, $p_1$ and $p_0$ are symmetric in $(T_n^{(1)}, \dots, T_n^{(L)})$. The result follows. 
\end{proof}

\begin{proof}[of \cref{thm:stable-copula}]
We prove the result by showing that $\|\tilde{G}_n - G_{P_0}\|_{\infty} \rightarrow_p 0$.

\medskip We first prove the result for the case when $F_0 = \unif(0,1)$. Recall that in subsampling, for each $b=1,\dots,B_n$ and each $l=1,\dots,L$, we have
\[ \hat{H}_{b,l} = T_m^{(l)}(X_{i_{1,b}}, \dots, X_{i_{m,b}}), \]
whose distribution function is $F_{m,P_n}$ because data is drawn iid from $P_n$. With $K_n$ defined in \cref{eqs:K}, let 
\[ \mathcal{A}_n := \left \{ \|\mathbb{F}_{\hat{\mathbf{H}}} - F_{m}^{P_n} \|_{\infty} \leq 1/K_n \right \}. \]
Using \cref{lem:hoeffding}, it holds that $\mathbb{P}_{P_n}(\mathcal{A}_n) \rightarrow 1$. 
Also recall from \cref{alg:agg-test} that 
\[ \tilde{H}_{b,l} = F_0^{-1}\left(\mathbb{F}_{\hat{\mathbf{H}}}(\hat{H}_{b,l}) - 1 / (2B_n L)  \right) = \mathbb{F}_{\hat{\mathbf{H}}}(\hat{H}_{b,l}) - 1 / (2B_n L), \quad b=1,\dots,B_n,\; l=1,\dots,L.\]
Let $\check{H}_{b,l} := F_{m,P_n}(\hat{H}_{b,l})$ accordingly. From the definitions of $\tilde{H}_{b,l}$, $\check{H}_{b,l}$ and $\mathcal{A}_n$, observe that on $\mathcal{A}_n$,
\[ \max_{b,l} |\tilde{H}_{b,l} - \check{H}_{b,l}| \leq 1 / K_n + 1 / (2B_n L) =: \epsilon_n, \]
where $\epsilon_n \rightarrow 0$.
Recall that $\tilde{S}_b := S(\tilde{H}_{b,1}, \dots \tilde{H}_{b,L})$ for $b=1,\dots,B_n$. Accordingly, let $\check{S}_b := S(\check{H}_{b,1}, \dots \check{H}_{b,L})$. By \cref{cond:lip} on $S$, it follows that on $\mathcal{A}_n$, 
\begin{equation} \label{eqs:proof-S-approx}
\max_{b} |\tilde{S}_b - \check{S}_b| \leq \epsilon_n.
\end{equation}

Because $P_n$ converges in copula to $P_0$, observe that for every $b$, it holds that
\begin{align*}
(\check{H}_{b,1},\dots,\check{H}_{b,L}) &= \left(F_{m,P_n}(T_m^{(1)}(X_{i_{1,b}}, \dots, X_{i_{m,b}})), \dots, F_{m,P_n}(T_m^{(L)}(X_{i_{1,b}}, \dots, X_{i_{m,b}})) \right) \\
& \rightarrow_d (C^{(1)}, \dots, C^{(L)}).
\end{align*}
By the continuous mapping theorem, it then follows that
\begin{equation} \label{eqs:proof-conv-null}
\check{S}_b \rightarrow_{d} S(C^{(1)}, \dots, C^{(L)}) \sim G_{P_0}, \quad b=1,\dots,B_n.
\end{equation}
Let 
\[ \mathcal{B}_n := \left\{\|\mathbb{F}_{\{\check{S}_b\}} - G_{P_0}\|_{\infty} \leq 1/K_n + \|G_{m,P_n} - G_{P_0}\|_{\infty} =: \delta_n \right\}, \]
where $G_{m,P_n}$ is the distribution function of $S(F_{m,P_n}(T_m^{(1)}), \dots, F_{m,P_n}(T_m^{(L)}))$ under $(X,\Omega) \sim P_n^m \times P_{\Omega}$. 
By an argument similar to the proof of \cref{lem:bound-unif-approx}, we have that $\mathbb{P}_{P_n}(\mathcal{B}_n) \rightarrow 1$. Further, in light of \cref{eqs:proof-conv-null} where $\check{S}_b \sim G_{m,P_n}$ and the continuity of $G_{P_0}$ by \cref{assump:stable-G}, we know $\|G_{m,P_n} - G_{P_0}\|_{\infty} \rightarrow 0$ and hence $\delta_n \rightarrow 0$. 

Recall that $\tilde{G}_n := \mathbb{F}_{\{\tilde{S}_b\}}$. Now, by definition of $\mathcal{B}_n$ and \cref{eqs:proof-S-approx}, on $\mathcal{A}_n \cap \mathcal{B}_n$, using \cref{lem:empirical-bound}, for every $x$ it holds that 
\begin{align*}
|\tilde{G}_n(x) - \mathbb{F}_{\{\check{S}_b\}}(x)| & \leq \mathbb{F}_{\{\check{S}_b\}}(x + \epsilon_n) - \mathbb{F}_{\{\check{S}_b\}}(x - \epsilon_n) \\
& \leq G_{P_0}(x + \epsilon_n) - G_{P_0}(x - \epsilon_n) + 2 \delta_n \\
& \leq 2 g_{\max,P_0} \epsilon_n + 2 \delta_n
\end{align*}
and consequently,
\[ \|\tilde{G}_n - G_{P_0} \|_{\infty} \leq 2 g_{\max,P_0} \epsilon_n + 3 \delta_n \rightarrow 0\]
since $\epsilon_n \rightarrow 0$, $\delta_n \rightarrow 0$ and $g_{\max,P_0} < \infty$ (by \cref{assump:stable-G}). 
Because $\mathbb{P}_{P_n}(\mathcal{A}_n \cap \mathcal{B}_n) \rightarrow 1$ by a union bound, our desired result that $\|\tilde{G}_{n} - G_{P_0}\|_{\infty} \rightarrow_p 0$ follows.  

\medskip The proof for the case when $F_0 = \N(0,1)$ follows analogously
using an argument similar to that employed to prove case (2) of \cref{lem:rank-approx}.
\end{proof}
\subsection{Test after rejection sampling}
\begin{proof}[of \cref{prop:test-rej}]
Let $\hat{Q}_n$ be the distribution defined through $\dd \hat{Q}_n / \dd P = \hat{r}_n$. Note that $\hat{Q}_n$ is random because $\hat{r}_n$ is random. Let $\gamma_n$ be a sequence of positive integers such that $\gamma_n / n \rightarrow 1/C$.
Consider a triangular array consisting of rows $(\tilde{X}^n_1, \dots, \tilde{X}^n_{\gamma_n})$ drawn iid from $\hat{Q}_n$ for $n=1,2,\dots$, i.e., the $n$-th row is drawn from $\hat{Q}_n^{\gamma_n}$. 

We first show that $T_{\gamma_n}(\tilde{X}^n_1, \dots, \tilde{X}^n_{\gamma_n}) \rightarrow_d T$. Let $\pi_n$ be the total variation coupling between $(\tilde{X}_1^n, \dots, \tilde{X}_{\gamma_n}^n) \sim \hat{Q}_n^{\gamma_n}$ and $(\tilde{X}_1^{\ast}, \dots, \tilde{X}_{\gamma_n}^{\ast}) \sim Q^{\gamma_n}$, where $\hat{Q}_n^{\gamma_n}$ and $Q^{\gamma_n}$ are respectively the $\gamma_n$-fold product measures of $\hat{Q}_n$ and $Q$. Under $\pi_n$ it holds that 
\begin{align*}
\pr\left(T_{\gamma_n}(\tilde{X}_1^n, \dots, \tilde{X}_{\gamma_n}^n) \neq T_{\gamma_n}(\tilde{X}_1^{\ast}, \dots, \tilde{X}_{\gamma_n}^{\ast}) \mid \hat{r}_n \right ) &\leq \pr \left((\tilde{X}_1^n, \dots, \tilde{X}_{\gamma_n}^n) \neq (\tilde{X}_1^{\ast}, \dots, \tilde{X}_{\gamma_n}^{\ast}) \mid \hat{r}_n \right ) \\
& = \TV(\hat{Q}_n^{\gamma_n},\, Q^{\gamma_n}) \\
& \leq \gamma_n \TV(\hat{Q}_n, Q) = \gamma_n \int |\hat{r}_n - r| \dd P,
\end{align*}
where we used total variation's tensorisation property. Taking expectation over $\hat{r}_n$, we have
\begin{align*}
\pr \left(T_{\gamma_n}(\tilde{X}_1^n, \dots, \tilde{X}_{\gamma_n}^n) \neq T_{\gamma_n}(\tilde{X}_1^{\ast}, \dots, \tilde{X}_{\gamma_n}^{\ast}) \right ) &\leq \gamma_n \E \int |\hat{r}_n - r| \dd P \\
&= (\gamma_n / n) \, n \E \int |\hat{r}_n - r| \dd P \rightarrow 0,
\end{align*}
where the last step used \cref{eqs:dens-ratio-tv} and $\gamma_n / n \rightarrow 1/C$. Again, using the optimal coupling formulation of total variation, it follows that the total variation between the law of $T_{\gamma_n}(\tilde{X}_1^n, \dots, \tilde{X}_{\gamma_n}^n)$ and the law of $T_{\gamma_n}(\tilde{X}_1^{\ast}, \dots, \tilde{X}_{\gamma_n}^{\ast})$ tends to zero. 
Because total variation is a strong metric and $T_{\gamma_n}(\tilde{X}_1^{\ast}, \dots, \tilde{X}_{\gamma_n}^{\ast}) \rightarrow_d T$ by our assumption, we have that $T_{\gamma_n}(\tilde{X}_1^n, \dots, \tilde{X}_{\gamma_n}^n) \rightarrow_d T$. 

By the properties of rejection sampling \citep{von195113}, conditioned on $\Gamma_n = \gamma_n$, the accepted sample $\tilde{X}_1^n, \dots, \tilde{X}_{\gamma_n}^n$ are iid from $\hat{Q}_n$. Further, $\Gamma_n \sim \bin(n, 1/C)$ and hence $\Gamma_n / n \rightarrow_p 1/C$. Consider the sequence $\tilde{T}_{n} := T_{\Gamma_n}(\tilde{X}_1^n, \dots, \tilde{X}_{\Gamma_n}^n)$. For every subsequence $n_i \nearrow \infty$, there is a further subsequence $n_{i_k}$ such that $\Gamma_{n_{i_k}} / n_{i_k} \rightarrow 1/C$ almost surely. This completes the proof because by the result we showed earlier $\tilde{T}_{n_{i_k}}$ has the same limit law for every such subsequence.  
\end{proof}

\begin{proposition}[Parametric density ratio and simple test statistic] \label{prop:test-rej-parametric}
Consider the setting of \cref{prop:test-rej} without requiring \cref{eqs:dens-ratio-tv} to hold. Suppose $r(x) = r(x; \beta_0)$ holds for some parametric family $r(x;\beta)$ differentiable in $\beta$ at $\beta_0$, such that $\|\beta - \beta_0\| \rightarrow 0$ implies $\int |r(x; \beta) - r(x; \beta_0)|^2 \dd P(x) \rightarrow 0$. Further, let $\hat{r}_n(x) := r(x; \hat{\beta}_n)$, where $\hat{\beta}_n$ is a consistent, asymptotically normal estimator of $\beta_0$ fitted on a separate sample of size $n$ from $P$. Let $f$ be a real-valued function of $X$ such that $\E_P f^4(X) < \infty$. Define $T_n := n^{-1/2} \sum_{i=1}^n \{f(X_i) - \E_Q f(X)\}$ and $T_0 := 0$. Then we have
\[ T_{\Gamma_n}(\tilde{X}^{n}_1, \dots, \tilde{X}^{n}_{\Gamma_n}) \rightarrow_d \N\left(0,\, \var_Q f(X) + v_{\beta_0}^{\T} \Sigma v_{\beta_0} / C \right),\]
where $\Sigma$ is the asymptotic covariance of $\hat{\beta}_n$ and $v_{\beta_0} = \nabla_{\beta} \E_P f(X) r(X; \beta)\mid_{\beta=\beta_0}$.
\end{proposition}

\begin{proof}
Let $\hat{Q}_n$ be the distribution defined through $\dd \hat{Q}_n / \dd P = \hat{r}_n$. Note that $\hat{Q}_n$ is random because $\hat{r}_n$ is random. Let $\gamma_n$ be a sequence of positive integers such that $\gamma_n / n \rightarrow 1/C$.
Consider a triangular array consisting of rows $(\tilde{X}^n_1, \dots, \tilde{X}^n_{\gamma_n})$ drawn iid from $\hat{Q}_n$ for $n=1,2,\dots$, i.e., the $n$-th row is drawn iid from $\hat{Q}_n$. By the proof of \cref{prop:test-rej}, it suffices to show that 
\[ \tilde{T}_n:= T_{\gamma_n}(\tilde{X}^{n}_1, \dots, \tilde{X}^{n}_{\gamma_n}) \rightarrow_d \N\left(0,\, \var_Q f(X) + v_{\beta_0}^{\T} \Sigma v_{\beta_0} / C \right). \]
By L\'evy's continuity theorem, to show the above, it suffices to show that the characteristic function of $\tilde{T}_n$ converges to the characteristic function of the limit distribution pointwise. Using 
\[ T_n = \gamma_n^{-1/2} \sum_{i=1}^{\gamma_n} \left\{f(\tilde{X}^{n}_i) - \E_{\hat{Q}_n} f(\tilde{X}_1^n)  \right\} + \sqrt{\gamma_n} \left\{\E_{\hat{Q}_n} f(X) - \E_Q f(X) \right\},\]
we can write the characteristic function of $\tilde{T}_n$ as
\begin{equation} \label{eqs:proof-cf}
\varphi_n(t) := \E \exp(i t \tilde{T}_n) = \E \left [ A_n(t) \exp\left(i t \sqrt{\gamma_n} \{\E_{\hat{Q}_n} f(X) - \E_Q f(X) \} \right) \right ],
\end{equation}
where 
\[ A_n(t):= \E\left[\exp\left(i t \gamma_n^{-1/2} \sum_{i=1}^{\gamma_n} \{f(\tilde{X}^{n}_i) - \E_{\hat{Q}_n} f(X) \} \right) \mid \{\hat{Q}_n\}_{n=1}^{\infty} \right]. \]

By the Lyapunov CLT, we have
\begin{equation} \label{eqs:proof-clt}
 \gamma_n^{-1/2} \sum_{i=1}^{\gamma_n} \frac{f(\tilde{X}^{n}_i) - \E_{\hat{Q}_n} f(X)}{\sqrt{\var_{\hat{Q}_n} f(X)}} \mid \{\hat{Q}_n\}_{n=1}^{\infty} \rightarrow_{d} \N(0, 1).
\end{equation}
To see this, note that $\E_P f^4(X) < \infty$ and $\hat{r}_n \leq C$ imply $\E_{\hat{Q}_n} |f(X)|^3 = \int \hat{r}_n |f^3| \dd P < \infty$, which ensures $\E_{\hat{Q}_n}|f(X) - \E_{\hat{Q}_n} f(X)|^3 < \infty$. 
Further, observe that 
\begin{align*}
 |\E_{\hat{Q}_n} f^2(X) - \E_{Q} f^2(X)| &\leq \int |r(x; \hat{\beta}_n) - r(x; \beta_0)| f^2(x) \dd P(x) \\
 & \leq \left(\int |r(x; \hat{\beta}_n) - r(x; \beta_0)|^2 \dd P(x) \right)^{1/2} (\E_P f^4(X))^{1/2} \rightarrow_p 0
\end{align*}
by $\hat{\beta}_n \rightarrow_p \beta_0$ and our assumption on $r(x; \beta)$ and $f$. This implies $\var_{\hat{Q}_n} f(X) \rightarrow_p \var_{Q} f(X)$. Let $n'$ be an arbitrary subsequence of $n$. Then there exists a further subsequence $n''$ along which $\var_{\hat{Q}_{n''}} f(X) \rightarrow_{a.s.} \var_{Q} f(X)$. By Slutsky's theorem, along the subsequence $n''$, \cref{eqs:proof-clt} implies
\[  \gamma_{n''}^{-1/2} \sum_{i=1}^{\gamma_{n''}} \frac{f(\tilde{X}^{n''}_i) - \E_{\hat{Q}_{n''}} f(\tilde{X}_1^{n''})}{\sqrt{\var_{Q} f(X)}} \mid \{\hat{Q}_n\}_{n=1}^{\infty} \rightarrow_{d} \N(0, 1) \]
and hence
\[  \gamma_{n''}^{-1/2} \sum_{i=1}^{\gamma_{n''}} \left\{ f(\tilde{X}^{n''}_i) - \E_{\hat{Q}_{n''}} f(\tilde{X}_1^{n''}) \right\} \mid \{\hat{Q}_n\}_{n=1}^{\infty} \rightarrow_{d} \N(0, \var_{Q} f(X)). \]
It follows from L\'evy's continuity theorem that at every $t$, $A_{n''}(t) \rightarrow_{a.s.} \exp\left\{- \var_{Q} f(X)\, t^2 / 2\right\}$ along $n''$.
Because $n'$ is an arbitrary subsequence of $n$, we conclude that 
\[A_{n}(t) \rightarrow_{p} \exp\left\{- \var_{Q} f(X)\, t^2 / 2\right\}\]
along the original sequence $n$.
Rewrite \cref{eqs:proof-cf} as 
\begin{multline*}
\varphi_n(t) = \exp\left\{- \var_{Q} f(X)\, t^2 / 2\right\} \E \exp\left(i t \sqrt{\gamma_n} \{\E_{\hat{Q}_n} f(X) - \E_Q f(X) \} \right) + \\
\E \left\{ \left[A_n(t) - \exp\left\{- \var_{Q} f(X)\, t^2 / 2\right\} \right] \exp\left(i t \sqrt{\gamma_n} \{\E_{\hat{Q}_n} f(X) - \E_Q f(X) \} \right)  \right\}.
\end{multline*}
For every $t$, $(A_n(t) - \exp\left\{- \var_{Q} f(X)\, t^2 / 2\right\}) \rightarrow_p 0$. Because the moduli of both $A_n(t)$ and $\exp\left(i t \sqrt{\gamma_n} \{\E_{\hat{Q}_n} f(X) - \E_Q f(X) \} \right)$  are bounded by one, by DCT, we have 
\[ \E \left\{ \left[A_n(t) - \exp\left\{- \var_{Q} f(X)\, t^2 / 2\right\} \right] \exp\left(i t \sqrt{\gamma_n} \{\E_{\hat{Q}_n} f(X) - \E_Q f(X) \} \right)  \right\} \rightarrow 0 \]
and hence 
\[\varphi_n(t) \rightarrow \exp\left\{- \var_{Q} f(X)\, t^2 / 2\right\} \, \lim_{n} \E \exp\left(i t \sqrt{\gamma_n} \{\E_{\hat{Q}_n} f(X) - \E_Q f(X) \} \right). \]

Because $\sqrt{n} (\hat{\beta}_n - \beta_0) \rightarrow_d \N(0, \Sigma)$, by the delta method, we have
\[ \sqrt{n} \{\E_{\hat{Q}_n} f(X) - \E_Q f(X) \} \rightarrow_d \N(0, v_{\beta_0}^{\T} \Sigma v_{\beta_0}) \]
and hence at every $t$,
\[ \E \exp\left(i t \sqrt{\gamma_n} \{\E_{\hat{Q}_n} f(X) - \E_Q f(X) \} \right) \rightarrow \exp\left[- (v_{\beta_0}^{\T} \Sigma v_{\beta_0} / C) t^2 / 2 \right]. \]
Finally, at every $t$, $\varphi_n(t)$ converges to the characteristic function of $\N(0,\, \var_Q f(X) + v_{\beta_0}^{\T} \Sigma v_{\beta_0} / C )$ and this completes the proof.
\end{proof}

\begin{remark} \label{rem:test-rej-parametric}
In addition to the variance of $f(X)$ under $Q$, the estimation of the density ratio also contributes to the asymptotic variance of the CLT in \cref{prop:test-rej-parametric}. While it may be difficult to form a self-normalised $Z$-statistic using the accepted sample, we expect other forms of normalisation, such as using a permutation test (as employed in our numerical studies), to behave properly. For permutations, the additional asymptotic variance equally affects $T_{\Gamma_n}$ and its permutation counterparts.
\end{remark}

\subsection{Variance estimator from rank-transformed subsampling} 
\begin{proof}[of \cref{prop:sigma}]
Fix $\epsilon > 0$. Let $\Pi := \left\{(b,l): \mathbb{F}_{\hat{\theta}_{m,b}^{(l)}}(\hat{\theta}_{m,b}^{(l)}) \in (\epsilon /2, 1-\epsilon /2) \right\}$. Observe that $|\Pi|>(1-\epsilon) B_n L \nearrow \infty$. Let $\tilde{\beta}$ be the least squares estimate of the slope from linearly regressing $\tilde{H}_{b,l}$ on $\hat{H}_{b,l} = \sqrt{m/L} (\hat{\theta}_{m,b}^{(l)} - \theta_0) / \sigma$ and an intercept term using those points in $\Pi$. It suffices to show that $\tilde{\beta} \rightarrow_p 1$. 
By inspecting the proof for case (2) of \cref{lem:rank-approx}, because $\epsilon^{F,P}_n \rightarrow 0$, observe that when $n$ is sufficiently large, we have $2 \epsilon^{F,P}_n < \epsilon$ and, with probability tending to one, $|\tilde{H}_{b,l} - \hat{H}_{b,l}| \leq 1 / |\Phi^{-1}(\epsilon^{F,P}_n)|$ for every $(b,l) \in \Pi$. Let $\hat{h}_0 := |\Pi|^{-1} \sum_{(b,l) \in \Pi} \hat{H}_{b,l}$ and note that
\[
\sum_{(b,l) \in \Pi} \hat{H}_{b,l}(\hat{H}_{b,l} - \hat{h}_0) = \sum_{(b,l) \in \Pi} (\hat{H}_{b,l} - \hat{h}_0)^2.
\]
Thus
\begin{equation*}
\tilde{\beta} - 1 = \frac{\sum_{(b,l) \in \Pi} (\tilde{H}_{b,l} - \hat{H}_{b,l}) (\hat{H}_{b,l} - \hat{h}_0)}{\sum_{(b,l) \in \Pi} (\hat{H}_{b,l} - \hat{h}_0)^2} \leq \frac{1}{|\Phi^{-1}(\epsilon^{F,P}_n)|} \frac{1}{\sqrt{|\Pi|^{-1} \sum_{(b,l) \in \Pi} (\hat{H}_{b,l} - \hat{h}_0)^2}}
\end{equation*}
with probability tending to one. Noting that $|\Phi^{-1}(\epsilon^{F,P}_n)| \rightarrow \infty$ and that $|\Pi|^{-1} \sum_{(b,l) \in \Pi} (\hat{H}_{b,l} - \hat{h}_0)^2$ tends to a positive constant by \cref{eqs:dml-single}, it follows that $\tilde{\beta} \rightarrow_p 1$. 
\end{proof}
\subsection{Auxiliary lemmas} \label{sec:app:aux}
The following standard result follows directly from the definition of $F^{-1}$.
\begin{lemma} \label{lem:df}
Given any distribution function $F$, we have $ F^{-1}(F(x)) \leq x$ and $F(F^{-1}(\alpha)) \geq \alpha$. 
\end{lemma}

\subsection{Other results} \label{app:extra}
\begin{proposition} \label{prop:tight}
	Suppose $T_n^{(1)}, \dots, T_n^{(L)}$ are exchangeable and $S$ satisfies \cref{cond:lip}. Then, under \cref{cond:pivotal}, the sequence $(S_n)$ is uniformly tight. 
\end{proposition}
\begin{proof}
	We know from Prohorov's theorem \citep[Thm.~2.4]{van2000asymptotic} and \cref{cond:pivotal} that the sequence $(T_n^{(1)})$ is uniformly tight, so given $\epsilon > 0$, there exists $M \geq 0$ such that $\sup_n \pr_P(|T_n^{(1)}| > M) < \epsilon/L$.
	
Now from \cref{cond:lip} we have that $|S_n - S(0, \ldots, 0)| \leq \|T_n\|_\infty$. Thus,
\begin{align*}
	\sup_n \pr_P \left( |S_n| > M + |S(0, \ldots, 0)| \right) &\leq \sup_n \pr_P \left( \|T_n\|_\infty > M \right) \\
	&\leq L \sup_n \pr_P\left(|T_n^{(1)}| > M\right) < \epsilon,
\end{align*}
applying a union bound and appealing to exchangeability in the final line.
\end{proof}

\begin{lemma} \label{lem:uniform_dist}
	Suppose that for all $P \in \mathcal{P}$, distribution functions $H_{n, P}$ converge uniformly to a continuous distribution function $H_P$, i.e., for all $x \in \R$,
	\[
	\sup_{P \in \mathcal{P}} |H_{n, P}(x) - H_P(x)| \to 0.
	\]
	Assume that each $H_P$ has density $h_P$ and $\sup_{P \in \mathcal{P}} \|h_P\|_\infty =: h_{\max} < \infty$.
	Further, suppose $\{H_P: P \in \mathcal{P}\}$ is tight. 
	Then we have
	\[
	\sup_{P \in \mathcal{P}} \|H_{n,P} - H_P\|_\infty \to 0.
	\]
\end{lemma}
\begin{proof}
	Given $\epsilon > 0$, by tightness of $\{H_P: P \in \mathcal{P}\}$, there exists  $M \geq 0$ such that $\sup_{P \in \mathcal{P}} H_P(-M) < \epsilon$ and $\sup_{P \in \mathcal{P}} \{1-H_P(M)\} < \epsilon$.
	Next set $\delta := \epsilon /  h_{\max}$ and $x_j = \delta (j-1) - M$ for $j=1, \ldots, \ceil{2M / \delta} =: J-1$. Additionally set $x_0 := -\infty$ and $x_J := \infty$. Fix $P \in \mathcal{P}$ and note that $H_P(x_j) - H_P(x_{j-1}) \leq \epsilon$ for all $j$.
	
	Now for each $x \in \R$, there exists $j$ with $x_{j-1} \leq x \leq x_j$, and
	\begin{align*}
		H_{n, P}(x) - H_P(x) &\leq H_{n, P}(x_j) - H_P(x_{j-1}) \leq H_{n, P}(x_j) - H_P(x_{j}) +\epsilon \\
		H_{n, P}(x) - H_P(x) &\geq H_{n, P}(x_{j-1}) - H_P(x_j) \geq H_{n, P}(x_{j-1}) - H_P(x_{j-1})  - \epsilon,
	\end{align*}
so
\[
\|H_{n, P} - H_P\|_{\infty} \leq \max_{j=0,\ldots,J} |H_{n, P}(x_j) - H_P(x_{j})| +\epsilon.
\]
Thus
\[
\limsup_{n \to \infty} \sup_{P\in \mathcal{P}} \|H_{n, P} - H_P\|_{\infty} \leq \limsup_{n \to \infty} \sup_{P \in \mathcal{P}} \max_{j=0,\ldots,J} |H_{n, P}(x_j) - H_P(x_{j})| +\epsilon = \epsilon.
\]
As $\epsilon$ was arbitrary, we have the desired result.
\end{proof}

\begin{lemma} \label{lem:G_P_unif}
	Let $G_{n, p}$ and $G_P$ for $P \in \mathcal{P}_0$ be as in (3) of \cref{thm:main-null}. Then
\[
\sup_{P \in \mathcal{P}_0} \|G_{n,P} - G_P\|_\infty \to 0.
\]
\end{lemma}
\begin{proof}
	By \cref{lem:uniform_dist} it suffices to show that for all $\epsilon > 0$, there exists $M \geq 0$ such that $\sup_{P \in \mathcal{P}_0} G_P(-M) < \epsilon$ and $\sup_{P \in \mathcal{P}_0} \{1-G_P(M)\} < \epsilon$.
	
	We have that given $\epsilon > 0$, there exists $M' \geq 0$ such that $F_0(-M') + F_0(M') < \epsilon / L$. Then, arguing similarly to the proof of \cref{prop:tight}, for all $P \in \mathcal{P}_0$,
	\[
	 \pr_P(|S_n| > M + |S(0, \ldots, 0)|)   \leq L  \pr_P\left(|T_n^{(1)}| > M\right) .
	\]
	Taking suprema over $P \in \mathcal{P}_0$ and limits as $n \to \infty$,
	\[
	\sup_{P \in \mathcal{P}_0} \{ G_P(-M') + G_P(M')\} < \epsilon,
	\]
	as required.
\end{proof}

\begin{proposition} \label{prop:CAN}
Let $W_1,\dots, W_n$ be iid from $P$. Suppose for $f(W; \beta_0)$ it holds that $\E f(W; \beta_0) =0$, $\sigma^2 := \E f^2(W; \beta_0) < \infty$ and that $f(W; \beta)$ permits a second order expansion in $\beta$ in a neighbourhood of $\beta_0$. Let $\hat{\beta}$ be a consistent, asymptotically linear estimator of $\beta_0$ from $W_1, \dots, W_n$. Then it holds that 
\[ \frac{\sum_i f(W_i; \hat{\beta})}{\sqrt{\sum_i f^2(W_i; \hat{\beta})}} \rightarrow_d \N(0,1). \]
\end{proposition}
\begin{proof}
In view of 
\[ \frac{\sum_i f(W_i; \hat{\beta})}{\sqrt{\sum_i f^2(W_i; \hat{\beta})}} = \frac{n^{-1/2} \sum_i f(W_i; \hat{\beta})}{\sqrt{n^{-1}\sum_i f^2(W_i; \hat{\beta})}}, \]
it suffices to show that $n^{-1} \sum_i f(W_i; \hat{\beta})$ is a consistent, asymptotically linear estimator of zero, i.e.,
\[ n^{-1/2} \sum_{i} f(W_i; \hat{\beta}) = n^{-1/2} \sum_i \varphi(W_i) + o_p(1), \]
for some $\varphi$ such that $\E \varphi(W) = 0$, $\E \varphi^2(W) < \infty$. Let $h(W):= (\partial f(W; \beta) / \partial \beta)|_{\beta = \beta_0}$. Using 
\[ f(W; \hat{\beta}) = f(W; \beta_0) + h(W)^{\T} (\hat{\beta} - \beta_0) + O(\|\hat{\beta} - \beta_0\|^2), \]
we have
\begin{align*}
n^{1/2} \sum_{i} f(W_i; \hat{\beta}) &= n^{-1/2} \sum_i \left[f(W_i; \beta_0) + h(W_i)^{\T} (\hat{\beta} - \beta_0) + O(\|\hat{\beta} - \beta_0\|^2) \right] \\
&= n^{-1/2} \sum_i f(W_i; \beta_0) + \sqrt{n}(\hat{\beta} - \beta_0)^{\T} \left[n^{-1} \sum_i h(W_i) \right] + O(n^{1/2} \|\hat{\beta} - \beta_0\|^2) \\
&= n^{-1/2} \sum_i f(W_i; \beta_0) + \sqrt{n} \left[\E h^{\T}(W) \right](\hat{\beta} - \beta_0) + O_p(n^{-1/2}) \\
&= n^{-1/2} \sum_i \left\{f(W_i; \beta_0) + \left[\E h^{\T}(W)\right] \gamma(W_i) \right\} + o_p(1),
\end{align*}
where $\gamma(W)$ is the influence function for $\hat{\beta}$. Hence, we have $\varphi(W) = f(W; \beta_0) + (\E h)^{\T} \gamma(W)$ and it is easy to check that $\E \varphi(W) = 0$ and $\E \varphi^2(W) < \infty$. 
\end{proof}

\section{First-order power analysis} \label{app:power-1st-order}
This section of the appendix is dedicated to proving \cref{thm:rank-asymp}, our main theorem for the first-order power behaviour of rank-transformed subsampling; we also prove \cref{thm:crit-ordinary} for the first-order power behaviour of ordinary subsampling. We sketch the proof of our main theorem here. If we compare the copula of $(T_n^{(1)}, \dots, T_n^{(L)})$ before and after the rank transform, the rank transform can be viewed as a random perturbation applied to each coordinate of the copula. It follows from the standard empirical process that the perturbation, when scaled by $\sqrt{n/m}$, converges to an average over exchangeable copies of a Brownian bridge (\cref{lem:bb}). To propagate this random limit to the change in distribution of $S(T_n^{(1)}, \dots, T_n^{(L)})$ after the rank transform, we rely on the functional delta method. As a crucial precursory step, in \cref{app:hadmard-diff} we establish the Hadamard differentiability of the distribution function of a monotone function of a random vector with respect to the coordinatewise perturbation of the random vector (\cref{prop:hada-aggregate}). This is achieved by firstly establishing the result for a single random variable (\cref{lem:hada-coor-cond}) and then extending it to a random vector through a conditioning argument (\cref{prop:hada-aggregate}). In \cref{app:rank-first-order}, we fully state and prove \cref{thm:rank-asymp} , which builds on \cref{cor:hada-aggregate}, itself a specialised version of \cref{prop:hada-aggregate} under symmetry. Then, in \cref{app:sub-first-order}, we prove \cref{thm:crit-ordinary} about the first-order behaviour of ordinary subsampling. In \cref{app:check-assump}, we check that the regularity assumptions required by \cref{thm:rank-asymp} are met by cases considered in this paper when the underlying copula is Gaussian. We illustrate our theory along with numerical results using a simple example in \cref{app:power-numerical}. Finally, various auxiliary results used through the arguments are stated and proved in \cref{app:aux-hadamard}.

\paragraph{Notation} For any real number $a$, we use the notation $a_+ = a \I\{a \geq 0\}$ and $a_- = -a\I\{a<0\}$, and similarly for real-valued functions.
Let $D[a,b]$ be the space of c\`adl\`ag functions on $[a,b]$ equipped with $\|\cdot\|_{\infty}$. 
For a function $f$, its modulus of continuity is $\Delta f(\delta):= \sup_{x,x': \|x-x'\|\leq \delta} |f(x) - f(x')|$ for $\delta \geq 0$. We use $a \vee b$ and $a \wedge b$ to denote $\max(a,b)$ and $\min(a,b)$, respectively. We use $\Id$ to denote the distribution function of $\unif(0,1)$. We use $\TV(X,Y)$ to denote the total variation distance between $X$'s distribution and $Y$'s distribution. 

\paragraph{Equicontinuity} Let $f_{\gamma}(x)$ be a class of real-valued functions defined on a Euclidean domain, indexed by $\gamma$. We say $\{f_{\gamma}(x): \gamma \in \Gamma\}$ is equicontinuous at $x_0$ if for any $\epsilon > 0$, there exists $\delta > 0$ such that $|f_{\gamma}(x) - f_{\gamma}(x_0)| < \epsilon$ holds for every $\gamma \in \Gamma$ whenever $\|x - x_0\|<\delta$.

\subsection{Hadamard differentiability of additive coordinatewise perturbation} \label{app:hadmard-diff}

\subsubsection{Univariate result} 

\begin{lemma} \label{lem:hada-coor-cond}
Let $X$ be a real-valued random variable that admits a density $f(x)$ with respect to the Lebesgue measure. Let $A$ be a subset of $\mathbb{R}$. Suppose $f_{\max} := \sup_{x \in A} f(x) < \infty$. Suppose $f(x)$ is uniformly continuous on $A$. Let $\mathcal{H}$ be the set of $\mathbb{R} \rightarrow \mathbb{R}$ functions such that the function takes value zero outside $A$. Then, for any $t \rightarrow 0$, $h_t, h \in \mathcal{H}$ such that $\|h_t - h\|_{\infty} \rightarrow 0$ and $h$ is uniformly continuous, it holds that 
\[\sup_{x \in A} \left|t^{-1}\left\{\pr\left(X + t h_t(X) \leq x \right) - \pr\left(X \leq x \right) \right\} + f(x)\,h(x) \right| \rightarrow 0.\]
\end{lemma}
\begin{proof}
Without loss of generality, we assume $t \searrow 0$ and $\|h_t\| > 0$.
Let $m_t := t\,\|h_t\|_{\infty}$, which is a positive sequence converging to zero by our assumption.
In light of 
\begin{equation} \label{eqs:event-decomp}
\begin{split}
\{X + t h_t(X) \leq x\} &= \{X \leq x - m_t\} \cup \left(\{x - m_t < X \leq x \} \cap \{X \leq x - t h_t(X)\} \right) \\
& \qquad \cup \left(\{x < X \leq x + m_t \} \cap \{X \leq x - t h_t(X)\} \right),  \\
\{X \leq x\} &= \{X \leq x - m_t\} \cup \{x - m_t < X \leq x\},
\end{split}
\end{equation}
we can write
\begin{multline}
\frac{\pr(X + t\,h_t(X) \leq x ) - \pr(X \leq x)}{t} = \frac{1}{t}\, \pr\left(\{x < X \leq x + m_t\} \cap \{X \leq x - t h_t(X)\} \right) \\
-\frac{1}{t}\, \pr\left(\{x - m_t < X \leq x\} \cap \{X > x - t h_t(X) \} \right).
\end{multline}

We claim that the first and second terms on the RHS converge to $f(x) \, h_{-}(x)$ and $- f(x) h_{+}(x)$ respectively, uniformly over $x \in A$. 
Using the fact that $h = h_{+} - h_{-}$, combining the two yields $-f(x) h(x)$ as desired.

We now prove the claim for the first term; the argument for the second term is similar. For any $x \in A$, we have
\[
\frac{1}{t}\, \pr\left(\{x < X \leq x + m_t\} \cap \{X \leq x - t h_t(X)\} \right) = \frac{1}{t} \int_{x}^{x+m_t} \I\{h_t(y) \leq (x-y)/t \} f(y) \dd y.\]
By a change of variable $y = x + m_t u$ for $u \in (0,1]$, we get
\begin{align*}
& \quad \frac{1}{t}\, \pr\left(\{x < X \leq x + m_t\} \cap \{X \leq x - t h_t(X)\} \right) \\
&= \frac{1}{t} \int_0^1 \I\{h_t(x + m_t u) \leq - (m_t / t) u \} f(x + m_t u ) \, m_t \dd u \\
&= \left(\frac{m_t}{t}\right) \int_0^1 \I\{\|h_t\|_{\infty} u \leq - h_t(x + m_t u) \} f(x + m_t u) \dd u \\
&= \|h_t\|_{\infty} \int_0^1 \I\{\|h_t\|_{\infty} u \leq - h_t(x + m_t u) \} f(x + m_t u) \dd u \\
&= \underbrace{\|h_t\|_{\infty} \int_0^1 \I\{\|h\|_{\infty} u \leq -h(x) \} f(x) \dd u}_{\text{(I)}} + \underbrace{\|h_t\|_{\infty} \int_0^1 \I\{\|h_t\|_{\infty} u \leq - h_t(x + m_t u) \} \left[ f(x + m_t u ) -  f(x) \right] \dd u}_{\text{(II)}} \\
& \qquad + \underbrace{\|h_t\|_{\infty} \int_0^1 \left[\I\{\|h_t\|_{\infty} u \leq - h_t(x + m_t u)\} - \I\{\|h\|_{\infty} u \leq - h(x)\} \right] f(x) \dd u}_{\text{(III)}}.
\end{align*}
The result follows from analyzing (I), (II) and (III) separately, detailed as below. 

\begin{description}
\item[Term (I):] When $\|h\|_{\infty}=0$, observe that 
\[ |\text{(I)}| \leq \|h_t\|_{\infty} f(x) \int_0^1 \dd u = \|h_t\|_{\infty} f(x) \leq \|h_t\|_{\infty} f_{\max} \rightarrow 0. \]
When $\|h\|_{\infty}>0$, we have
\begin{align*}
 \text{(I)} = \|h_t\|_{\infty} f(x) \left( \int_0^{h_{-}(x) / \|h\|_{\infty}} \dd u \right) &= f(x) \, h_{-}(x) \|h_t\|_{\infty} / \|h\|_{\infty} \\
&= f(x) \, h_{-}(x) + f(x) \, h_{-}(x) \left( \|h_t\|_{\infty} / \|h\|_{\infty} - 1 \right), 
\end{align*}
where the second term uniformly converges to zero because $h$ is bounded and $f(x) \leq f_{\max}$. Thus, in either case, (I) converges to $f(x) \, h_{-}(x)$ uniformly over $x \in A$. 

\item[Term (II):] We show that (II) converges to zero uniformly. Because $\|h_t\|_{\infty} \rightarrow \|h\|_{\infty}$ irrespective of $x$, it suffices to show 
\[ \underbrace{\int_0^1 \I\{\|h_t\|_{\infty} u \leq - h_t(x + m_t u) \} \left| f(x + m_t u) -  f(x) \right| \dd u}_{\text{(II)}'} \rightarrow 0 \]
uniformly over $x \in A$. If $x + m_t u \notin A$, $\text{(II)}' = 0$ by $h_t \in \mathcal{H}$; if $x + m_t u \in A$, we have 
\[\text{(II)}' \leq \int_0^1 \left| f(x + m_t u) -  f(x) \right| \dd u \leq \sup_{x,x'\in A: |x - x'| \leq m_t} |f(x) - f(x')| \rightarrow 0\]
by uniform continuity of $f$ on $A$. 

\item[Term (III):] We now show that (III) converges to zero uniformly. The case when $\|h\|_{\infty} = 0$ is straightforward: we have
\[ |\text{(III)}| \leq 2  \|h_t\|_{\infty} f(x) \int_0^1 \dd u \leq 2 f_{\max} \|h_t\|_{\infty} \rightarrow 0. \]
Now we suppose $\|h\|_{\infty} > 0$. 
We have that
\begin{align*}
\text{(III)} &= \|h_t\|_{\infty} \, f(x) \int_0^1 \left[\I\left\{u \leq - \frac{h_t(x + m_t u)}{\|h_t\|_{\infty}}\right\} - \I\left\{u \leq - \frac{h(x)}{\|h\|_{\infty}}\right\} \right] \dd u \\
&= \|h_t\|_{\infty} \, f(x) \underbrace{\int_0^1 \left[\I\left\{u \leq \left(\frac{h_t(x + m_t u)}{\|h_t\|_{\infty}}\right)_{-}\right\} - \I\left\{u \leq \left(\frac{h(x)}{\|h\|_{\infty}}\right)_{-}\right\} \right] \dd u}_{\text{(III)}'}.
\end{align*}

Because $|f(x)| \leq f_{\max}$ for $x \in A$, it suffices to show that $\text{(III)}'$ uniformly converges to zero. 
Observe that
\begin{align*}
\left|\left(\frac{h_t(x)}{\|h_t\|_{\infty}} \right)_{-} - \left(\frac{h(x)}{\|h\|_{\infty}} \right)_{-} \right| &\leq \left|\left(\frac{h_t(x)}{\|h_t\|_{\infty}} \right) - \left(\frac{h(x)}{\|h\|_{\infty}} \right) \right| \\
&\leq \|h\|_{\infty} (1/\|h_t\|_{\infty} - 1 / \|h\|_{\infty}) + \|h_t - h\|_{\infty} / \|h_t\|_{\infty}  \rightarrow 0,
\end{align*}
uniformly over $x$. Therefore, because $|m_t u| \leq m_t = t \|h_t\|_{\infty} \rightarrow 0$, given any $\epsilon > 0$, there exists $T'_{\epsilon} > 0$ such that for every $t < T'_{\epsilon}$, 
\[ \left\|\left(\frac{h_t}{\|h_t\|_{\infty}} \right)_{-} - \left(\frac{h}{\|h\|_{\infty}} \right)_{-} \right\|_{\infty} < \epsilon. \]
It follows that for $t < T'_{\epsilon}$,
\[ \text{(III)}' \leq \int_0^1 \left[\I\left\{u \leq \epsilon + \left(\frac{h(x + m_t u)}{\|h\|_{\infty}}\right)_{-}\right\} - \I\left\{u \leq \left(\frac{h(x)}{\|h\|_{\infty}}\right)_{-}\right\} \right] \dd u. \]
Further, because $h$ is uniformly continuous, there exists $T''_{\epsilon} > 0$ such that $|h(x+ m_t) - h(x)| < \epsilon \|h\|_{\infty} $ for all $t < T''_{\epsilon}$. Hence, for all $t < T'_{\epsilon} \wedge T''_{\epsilon}$, we have
\[ \text{(III)}' \leq \int_0^1 \left[\I\left\{u \leq 2\epsilon + \left(\frac{h(x)}{\|h\|_{\infty}}\right)_{-}\right\} - \I\left\{u \leq \left(\frac{h(x)}{\|h\|_{\infty}}\right)_{-}\right\} \right] \dd u \leq 2 \epsilon.\]
An analogous lower bound can be derived in a similar way. Because $T'_{\epsilon}, T''_{\epsilon}$ do not depend on $x$ and $\epsilon > 0$ is arbitrary, we conclude that $\text{(III)}'$ converges to zero uniformly as desired.
\end{description}

\end{proof}

\subsubsection{Multivariate result}
\begin{proposition}[Hadamard differentiability of additive coordinatewise perturbation] \label{prop:hada-aggregate}
Let $R(x_1, \dots, x_L): [0,1]^L \rightarrow \mathbb{R} \cup \{\pm \infty\}$ be a function that is non-decreasing in each coordinate. Let $[a,b]$ be an interval in the range of $R$. For $l=1,\dots,L$, define
\[ R_l^{-1}(x_{-l}; r) :=  \sup \{x_l \in [0,1]: R(x) \leq r\}, \quad (x_{-l}, r) \in [0,1]^{L-1} \times [a,b]. \]

Let $X = (X_1, \dots, X_L)$ be a random vector in $[0,1]^L$ with a distribution absolutely continuous with respect to the Lebesgue measure. Let $f(x_l \mid x_{-l}) := f(x_l \mid X_{-l} = x_{-l})$ for $l=1,\dots,L$ be the corresponding conditional densities with respect to the Lebesgue measure.

Let $D_{\varphi}$ be the subset of $D[0,1]$ consisting of $g$ such that $0 \leq u + g(u) \leq 1$ for $u \in [0,1]$ and $g(0) = g(1) = 0$.
For $g \in D_{\varphi}$, consider the following map
\[ \varphi: D_{\varphi} \rightarrow D[a,b], \quad \varphi(g) = r \mapsto \pr\left(R(X_1 + g(X_1), \dots, X_L + g(X_L)) \leq r \right). \]
Then, under \cref{assump:f-cont,assump:f-equi} below, the map $\varphi$ is Hadamard differentiable at $g=0$ tangentially to any uniformly continuous $h \in D[0,1]$ satisfying $h(0) = h(1) = 0$ with derivative
\[ \varphi'_{g=0}(h) = r \mapsto  -\sum_{l=1}^L \E\left[f(R_l^{-1}(X_{-l}; r) \mid X_{-l}) \, h(R_l^{-1}(X_{-l}; r))\right], \quad r \in [a,b]. \]
Further, the derivative is continuous on $[a,b]$. 
\end{proposition}

\begin{assumption} \label{assump:f-cont} 
For $l=1,\dots,L$ and Lebesgue almost every $x_{-l} \in [0,1]^{L-1}$, the map $x_l \mapsto f(x_l \mid x_{-l})$ is continuous on $[0,1]$.  
\end{assumption}

\begin{assumption} \label{assump:f-equi}
There exists $\delta > 0$ such that for $l=1,\dots,L$, $\{x_l \mapsto f(x_l \mid x_{-l}): x_{-l}\}$ is equicontinuous on $A_{\delta}^l := \{(x_l, x_{-l}): x_l \in A_{\delta,x_{-l}}^l, \; x_{-l} \in [0,1]^{L-1}\}$, where 
\[ A_{\delta, x_{-l}}^l := \left[\left(R_l^{-1}(x_{-l} + \delta; a) - \delta\right) \vee 0,\; \left(R_l^{-1}(x_{-l} - \delta; b) + \delta\right) \wedge 1 \right], \] 
where $x_{-l} \pm \delta$ is applied coordinatewise. 
\end{assumption}

The next corollary specialises \cref{prop:hada-aggregate} to the case where $R$ is a symmetric function and the random vector is exchangeable. 
\begin{corollary}[Under symmetry] \label{cor:hada-aggregate}
Let $R(c_1, \dots, c_L): [0,1]^L \rightarrow \mathbb{R} \cup \{\pm \infty\}$ be a symmetric function that is non-decreasing in each coordinate. Let $[a,b]$ be an interval in the range of $R$. Define
\[ R^{-1}(c_{-1}; r) :=  \sup \{c_1 \in [0,1]: R(c) \leq r\}, \quad (c_{-1}, r) \in [0,1]^{L-1} \times [a,b]. \]

Let $C = (C_1, \dots, C_L)$ be an exchangeable random vector in $[0,1]^L$ with a distribution absolutely continuous with respect to the Lebesgue measure. Let $f(c_1 \mid c_{-1}) := f(c_1 \mid C_{-1} = c_{-1})$ be the corresponding conditional density with respect to the Lebesgue measure.

Let $D_{\varphi}$ be the subset of $D[0,1]$ consisting of $g$ such that $0 \leq u + g(u) \leq 1$ for $u \in [0,1]$ and $g(0) = g(1) = 0$.
For $g \in D_{\varphi}$, consider the following map
\[ \varphi: D_{\varphi} \rightarrow D[a,b], \quad \varphi(g) = r \mapsto \pr\left(R(C_1 + g(C_1), \dots, C_L + g(C_L)) \leq r \right). \]
Then, under \cref{assump:f-cont-symm,assump:f-equi-symm} below, the map $\varphi$ is Hadamard differentiable at $g=0$ tangentially to any uniformly continuous $h \in D[0,1]$ satisfying $h(0) = h(1) = 0$ with derivative
\[ \varphi'_{g=0}(h) = r \mapsto  -L \E\left[f(R^{-1}(C_{-1}; r) \mid C_{-1}) \, h(R^{-1}(C_{-1}; r)) \right], \quad r \in [a,b]. \]
Further, the derivative is continuous on $[a,b]$. 
\end{corollary}
\begin{assumption} \label{assump:f-cont-symm}
For Lebesgue almost every $c_{-1} \in [0,1]^{L-1}$, the map $c_1 \mapsto f(c_1 \mid c_{-1})$ is continuous on $[0,1]$.
\end{assumption}
\begin{assumption} \label{assump:f-equi-symm}
There exists $\delta > 0$ such that $\{c_1 \mapsto f(c_1 \mid c_{-1}): c_{-1}\}$ is equicontinuous on $A_{\delta} := \{(c_1, c_{-1}): c_1 \in A_{\delta,c_{-1}}, \; c_{-1} \in [0,1]^{L-1}\}$, where 
\[ A_{\delta, c_{-1}} := \left[\left(R^{-1}(c_{-1} + \delta; a) - \delta\right) \vee 0,\; \left(R^{-1}(c_{-1} - \delta; b) + \delta\right) \wedge 1 \right], \] 
where $c_{-1} \pm \delta$ is applied coordinatewise.  
\end{assumption}

We now prove \cref{prop:hada-aggregate}, the key proposition that underpins the first-order power result. The proof relies several auxiliary results which can be found in \cref{app:aux-hadamard}.

\begin{proof}[of \cref{prop:hada-aggregate}]
First note that the claimed $\varphi'_{g=0}(h)$ is linear in $h$. 
Further, we claim that the linear map is continuous because $\|\varphi'_{g=0}(h)\|_{\infty} \leq L f_{\max} \, \|h\|_{\infty}$, where we define
\begin{equation} \label{eqs:pf-fmax}
 f_{\max} := \max_{1 \leq l \leq L} \sup_{x \in A_{\delta}^l} f(x_l \mid x_{-l}),  
\end{equation}
based on the region defined in \cref{assump:f-equi}. Using \cref{assump:f-equi} and the fact that every $A_{\delta,x_{-l}}^l \subset [0,1]$ is compact, we have $f_{\max} < \infty$.
To see the $f(R_l^{-1}(X_{-l}; r) \mid X_{-l}) \leq f_{\max}$, note that the monotonicity of $R$ implies the following property of $A_{\delta, x_{-l}}^l$ ($l=1,\dots,L$):
\begin{multline} \label{eqs:pf-region}
x_l \in [0,1],\; |x_l - R_l^{-1}(x_{-l}'; r)|\leq \delta \text{ for some $x_{-l}', r$ such that $\|x_{-l}' - x_{-l}\|_{\infty} \leq \delta$, $r \in [a,b]$} \\
 \implies x_l \in A_{\delta, x_{-l}}^l.
\end{multline}
Clearly, $R_l^{-1}(x_{-l}; r) \in A_{\delta, x_{-l}}^l$ for any $r \in [a,b]$.  

\medskip Fix any uniformly continuous $h \in D[0,1]$ with $h(0)=h(1)=0$. Consider $t \searrow 0$ and let $h_t$ be a sequence in $D[0,1]$ such that $\|h_t - h\|_{\infty} \rightarrow 0$ and $t h_t \in D_{\varphi}$, which implies $h_t(0) = h_t(1) = 0$. Without loss of generality, we can assume $\|h_t\|_{\infty} > 0$. Because $t \|h_t\|_{\infty} \rightarrow 0 $, fix any $T_{\delta} > 0$ with the following property, which will be useful later:
\begin{equation} \label{eqs:pf-T-delta}
t < T_{\delta} \implies t \|h_t\|_{\infty} < \delta.
\end{equation}

We want to show that 
\begin{multline*}
\sup_{a \leq r \leq b} \bigg|t^{-1}\left\{\pr\left(R(X_1 + t h_t(X_1), \dots, X_L + t h_t(X_L) \right) \leq r) - \pr\left(R(X_1, \dots, X_L) \leq r \right) \right\} \\
+ \sum_{l=1}^L \E\left[f(R_l^{-1}(X_{-l}; r) \mid X_{-l}) \, h(R_l^{-1}(X_{-l}; r)) \} \right]  \bigg| \rightarrow 0.
\end{multline*}
In light of the following telescoping sum of $L$ quotients
\begin{align*}
& \quad t^{-1}\left\{\pr\left(R(X_1 + t h_t(X_1), \dots, X_L + t h_t(X_L)) \leq r \right) - \pr\left(R(X_1, \dots, X_L) \leq r \right) \right\} \\
&= t^{-1}\big\{\pr\left(R(X_1 + t h_t(X_1), \dots, X_L + t h_t(X_L)) \leq r \right) \\
& \qquad \qquad - \pr\left(R(X_1, X_2 + t h_t(X_2), \dots, X_L + t h_t(X_L)) \leq r \right) \big\} \\
& \quad + t^{-1} \big\{\pr\left(R(X_1, X_2 + t h_t(X_2), \dots, X_L + t h_t(X_L)) \leq r \right) \\
& \qquad \qquad - \pr\left(R(X_1, X_2, X_3 + t h_t(X_3), \dots, X_L + t h_t(X_L)) \leq r \right)  \big\} \\
& \qquad +  \dots + t^{-1} \big\{\pr\left(R(X_1, \dots, X_{L-1}, X_L + t h_t(X_L)) \leq r \right) - \pr\left(R(X_1, \dots, X_{L-1}, X_L) \leq r \right) \big\},
\end{align*}
we prove our result by showing that, for $l=1,\dots,L$, 
\begin{multline*}
\sup_{a \leq r \leq b} \bigg|t^{-1}\big\{\pr\left(R(X_1, \dots, X_{l-1}, X_l + t h_t(X_l), \dots, X_L + t h_t(X_L)) \leq r \right) \\
- \pr\left(R(X_1, \dots, X_{l-1}, X_{l}, X_{l+1} + t h_t(X_{l+1}), \dots, X_L + t h_t(X_L)) \leq r \right) \big\} \\
+ \E\left[f(R_l^{-1}(X_{-l}; r) \mid X_{-l}) \, h(R_l^{-1}(X_{-l}; r))  \right]  \bigg| \rightarrow 0,
\end{multline*}
or, by conditioning on $X_{-l}:=(X_1, \dots, X_{l-1}, X_{l+1}, \dots, X_L)$ and using the monotonicity of $R$, equivalently, 
\begin{multline} \label{eqs:pf-desired}
\sup_{a \leq r \leq b} \bigg|\E \big[ t^{-1}\left\{\pr\left( X_l + t h_t(X_l) \leq R_l^{-1}(X_{-l}^{t}; r) \mid X_{-l} \right) - \pr\left(X_l \leq R_l^{-1}(X_{-l}^t; r) \mid X_{-l} \right) \right\} \\
+ f(R_l^{-1}(X_{-l}; r) \mid X_{-l}) \, h(R_l^{-1}(X_{-l}; r)) \} \big] \bigg| \rightarrow 0.
\end{multline}
In above, we define
\begin{equation} \label{eqs:pf-xlt}
X_{-l}^{t}:= \left(X_1, \dots, X_{l-1}, X_{l+1} + t h_t(X_{l+1}), \dots, X_L + t h_t(X_L) \right),
\end{equation}
which is a measurable function of $X_{-l}$ and satisfies
\begin{equation} \label{eqs:pf-xdiff}
\|X_{-l}^t - X_{-}\|_{\infty} \leq t \|h_t\|_{\infty}.
\end{equation}
We prove our desired \cref{eqs:pf-desired} by showing both (I) and (II) below converge to zero uniformly over $r \in [a,b]$: 
\begin{multline*}
\mathrm{(I)} := \E \big[ t^{-1}\left\{\pr\left( X_l + t h_t(X_l) \leq R_l^{-1}(X_{-l}^{t}; r) \mid X_{-l} \right) - \pr\left(X_l \leq R_l^{-1}(X_{-l}^t; r) \mid X_{-l} \right) \right\} \\
+ f(R_l^{-1}(X_{-l}^t; r) \mid X_{-l}) \, h(R_l^{-1}(X_{-l}^t; r)) \big], 
\end{multline*}
\[\mathrm{(II)} := \E \big[ f(R_l^{-1}(X_{-l}^t; r) \mid X_{-l}) \, h(R_l^{-1}(X_{-l}^t; r)) -  f(R_l^{-1}(X_{-l}; r) \mid X_{-l}) \, h(R_l^{-1}(X_{-l}; r)) \big]. \]

\paragraph{Bounding (I)} We write $\text{(I)} = \E J_{t,r}^l(X_{-l})$ with 
\begin{multline} \label{eqs:pf-J}
J^l_{t,r}(X_{-l}) := t^{-1}\left\{\pr\left( X_l + t h_t(X_l) \leq R_l^{-1}(X_{-l}^{t}; r) \mid X_{-l} \right) - \pr\left(X_l \leq R_l^{-1}(X_{-l}^t; r) \mid X_{-l} \right) \right\} \\
+ f(R_l^{-1}(X_{-l}^t; r) \mid X_{-l}) \, h(R_l^{-1}(X_{-l}^t; r)).
\end{multline}
It holds that
\begin{multline} \label{eqs:pf-Jt-tilde}
\sup_{r \in [a,b]} |J_{t,r}^l(X_{-l})| \leq \sup_{x_l \in [0,1]} \big| t^{-1}\left\{\pr\left( X_l + t h_t(X_l) \leq x_l \mid X_{-l} \right) - \pr\left(X_l \leq x_l \mid X_{-l} \right) \right\} \\
+ f(x_l \mid X_{-l}) \, h(x_l) \big| =: \tilde{J}_{t}^l(X_{-l}).
\end{multline}
We claim that 
\[ \sup_{r \in [a,b]} |J_{t,r}^l(X_{-l})| \leq \tilde{J}_{t}^l(X_{-l}) \rightarrow 0 \quad \text{a.e.}, \]
which follows from observing that $R_l^{-1}(X_{-l}^{t}; r)$ is a measurable function of $X_{-l}$ and applying \cref{lem:hada-coor-cond} with $A=[0,1]$ conditional on $X_{-l}$. 
To see that \cref{lem:hada-coor-cond} can be applied, note that for almost every $x_{-l}$, the map $x_l \mapsto f(x_l \mid x_{-l})$ is continuous on $A$ (\cref{assump:f-cont}), which, by compactness of $A$, further implies that $x_l \mapsto f(x_l \mid x_{-l})$ is uniformly continuous and is bounded. 

In light of 
\[ \sup_{r \in [a,b]} |\text{(I)}| = \sup_{r \in [a,b]} |\E J_{t,r}^l(X_{-l})| \leq \E \sup_{r \in [a,b]} |J_{t,r}^l(X_{-l})|, \]
to show that LHS converges to zero through showing $\E \sup_{r \in [a,b]} |J_{t,r}^l(X_{-l})| \rightarrow 0$, by DCT, it remains to show that 
$\sup_{r \in [a,b]} |J_{t,r}^l(X_{-l})|$ is dominated by an integrable random variable for all sufficiently small $t$. We have
\begin{multline*}
\sup_{r \in [a,b]} |J_{t,r}^l(X_{-l})| \leq \underbrace{\sup_{r \in [a,b]} t^{-1}\left|\pr\left( X_l + t h_t(X_l) \leq R_l^{-1}(X_{-l}^{t}; r) \mid X_{-l} \right) - \pr\left(X_l \leq R_l^{-1}(X_{-l}^t; r) \mid X_{-l} \right) \right|}_{\text{(I-a)}} \\
+ \underbrace{\sup_{r \in [a,b]} |f(R_l^{-1}(X_{-l}^t; r) \mid X_{-l}) \, h(R_l^{-1}(X_{-l}^t; r))|}_{\text{(I-b)}}.
\end{multline*}
First, by \cref{eqs:pf-T-delta,eqs:pf-xdiff,eqs:pf-region}, whenever $t < T_{\delta}$, we have $R_l^{-1}(X_{-l}^t; r) \in A_{\delta, x_{-l}}^l$ and hence $\text{(I-b)} \leq f_{\max} \|h\|_{\infty}$ by \cref{eqs:pf-fmax}. We now argue that $\text{(I-a)}$ is also bounded when $t < T_{\delta}$. To see the upper bound, observe that  
\begin{align*}
& \quad \sup_{r \in [a,b]} t^{-1}\left\{\pr\left( X_l + t h_t(X_l) \leq R_l^{-1}(X_{-l}^{t}; r) \mid X_{-l} \right) - \pr\left(X_l \leq R_l^{-1}(X_{-l}^t; r) \mid X_{-l} \right) \right\} \\
& \leq \sup_{r \in [a,b]} t^{-1}\left\{\pr\left( X_l \leq R_l^{-1}(X_{-l}^{t}; r) + t \|h_t\|_{\infty} \mid X_{-l} \right) - \pr\left(X_l \leq R_l^{-1}(X_{-l}^t; r) \mid X_{-l} \right) \right\} \\
& = \sup_{r \in [a,b]} f(X_l^{\ast}(r) \mid X_{-l}) \, \|h_t\|_{\infty} \leq f_{\max} \|h_t\|_{\infty}, \quad \forall\, t < T_{\delta},
\end{align*}
where we used the mean value theorem and $X_l^{\ast}(r) \in [R_l^{-1}(X_{-l}^{t}; r), (R_l^{-1}(X_{-l}^{t}; r) + t \|h_t\|_{\infty}) \wedge 1]$ is a random variable that depends on $X_{-l}$ and $r$. The final upper bound follows from the fact that $X_l^{\ast}(r) \in A_{\delta, X_{-l}}^l$ (see \cref{assump:f-equi}) almost surely for any $r \in [a,b]$ and every $t < T_{\delta}$. By a similar argument, $\text{(I-a)}$ is also lower bounded when $t < T_{\delta}$. Hence, we can apply DCT and conclude that (I) converges to zero uniformly over $r \in [a,b]$. 

\paragraph{Bounding (II)} We write $\text{(II)} = \E V_{t,r}^l(X_{-l})$ with  
\begin{equation} \label{eqs:pf-II-V}
V_{t,r}^l(X_{-l}) := f(R_l^{-1}(X_{-l}^t; r) \mid X_{-l}) \, h(R_l^{-1}(X_{-l}^t; r)) -  f(R_l^{-1}(X_{-l}; r) \mid X_{-l}) \, h(R_l^{-1}(X_{-l}; r)).
\end{equation}
We first show the pointwise convergence
\begin{equation} \label{eqs:pf-II-pointwise}
\E V_{t,r}^l(X_{-l}) \rightarrow 0, \quad r \in [a,b].
\end{equation}
Fix any $r \in [a,b]$. Recall that whenever $t < T_{\delta}$, we have $R_l^{-1}(X_{-l}^t; r) \in A_{\delta, x_{-l}}^l$ and hence $|f(R_l^{-1}(X_{-l}^t; r) \mid X_{-l})| \leq f_{\max}$, $|V_{t,r}^l(X_{-l})| \leq 2 f_{\max} \|h\|_{\infty}$.
By DCT, it remains to show that $V_{t,r}^l(X_{-l}) \rightarrow 0$ for almost every $X_{-l}$. 
By monotonicity of $R$, for every fixed $r \in [a,b]$, $x_{-l} \mapsto -R_l^{-1}(x_{-l}; r)$ is coordinatewise non-decreasing. By \cref{prop:monotone-cont-ae}, we conclude that for every $r \in [a,b]$, $x_{-l} \mapsto R_l^{-1}(x_{-l}; r)$ is continuous at Lebesgue almost every $x_{-l} \in [0,1]^{L-1}$. Hence, together with the continuity of $x_l \mapsto f(x_l \mid X_{-l})$ (\cref{assump:f-cont}) and the continuity of $h$, we have
\[ f(R_l^{-1}(X_{-l}^t; r) \mid X_{-l}) \rightarrow f(R_l^{-1}(X_{-l}; r) \mid X_{-l}), \quad h(R_l^{-1}(X_{-l}^t; r)) \rightarrow h(R_l^{-1}(X_{-l}; r)), \quad \text{a.e.}.\]

\medskip We now argue that \cref{eqs:pf-II-pointwise} can be strengthened to uniform convergence over $r \in [a,b]$. By \cref{lem:equicont-compact}, it suffices to show that, for $K'$ a dense subset of $[a,b]$,
\begin{equation} \label{eqs:pf-equi-V}
\left|\E V_{t,r_t}^l(X_{-l}) - \E V_{t,r}^l(X_{-l}) \right| \rightarrow 0, \quad \text{for any $r \in K'$ and any $[a,b] \ni r_t \rightarrow r$}.
\end{equation}
Fix an arbitrary $r \in K'$ and $r_t \rightarrow r$. We have 
\begin{equation} \label{eqs:pf-II-ab}
\begin{split}
& \quad \left|\E V_{t,r_t}^l(X_{-l}) - \E V_{t,r}^l(X_{-l}) \right| \\
&\leq \underbrace{\E \left|f(R_l^{-1}(X_{-l}^t; r_t) \mid X_{-l}) \, h(R_l^{-1}(X_{-l}^t; r_t)) - f(R_l^{-1}(X_{-l}^t; r) \mid X_{-l}) \, h(R_l^{-1}(X_{-l}^t; r))  \right|}_{\text{(II-a)}} \\
& \quad + \underbrace{\E \left|f(R_l^{-1}(X_{-l}; r_t) \mid X_{-l}) \, h(R_l^{-1}(X_{-l}; r_t)) - f(R_l^{-1}(X_{-l}; r) \mid X_{-l}) \, h(R_l^{-1}(X_{-l}; r))  \right|}_{\text{(II-b)}}. 
\end{split}
\end{equation}
First, we claim that (II-b) tends to zero, which follows from $R_l^{-1}(X_{-l}; r_t) \rightarrow R_l^{-1}(X_{-l}; r)$ for a.e. $X_{-l}$ (implied by \cref{prop:monotone-cont-ae} and monotonicity of $R_l^{-1}$; see also the proof of \cref{lem:pr-equi}) and DCT. Then, for $t < T_{\delta}$, note that (II-a) is further bounded by 
\begin{multline*}
\text{(II-a)} \leq \|h\|_{\infty}\, \E \left|f(R_l^{-1}(X_{-l}^t; r_t) \mid X_{-l}) - f(R_l^{-1}(X_{-l}^t; r) \mid X_{-l}) \right| \\
+ f_{\max}\, \E \left|h(R_l^{-1}(X_{-l}^t; r_t)) - h(R_l^{-1}(X_{-l}^t; r)) \right|.
\end{multline*}
We shall prove \cref{eqs:pf-equi-V} by showing that
\begin{align}
\E \left|f(R_l^{-1}(X_{-l}^t; r_t) \mid X_{-l}) - f(R_l^{-1}(X_{-l}^t; r) \mid X_{-l}) \right| &\rightarrow 0, \label{eqs:pf-equi-f-expectation} \\
\E \left|h(R_l^{-1}(X_{-l}^t; r_t)) - h(R_l^{-1}(X_{-l}^t; r)) \right| &\rightarrow 0. \label{eqs:pf-equi-h-expectation}
\end{align}

We now show \cref{eqs:pf-equi-f-expectation}.  Recall that $X_{-l}^t$ is given by \cref{eqs:pf-xlt} and the expectation in \cref{eqs:pf-equi-f-expectation} is taken over $X_{-l}$. For any $\epsilon > 0$, let 
\[\mathcal{E}_{\epsilon,t} := \{|R_l^{-1}(X_{-l}^t; r_t) - R_l^{-1}(X_{-l}^t; r)| > \epsilon \}.\] Considering the integral over $\mathcal{E}_{\epsilon,t}$ and its complement, for $t < T_{\delta}$, we have 
\begin{align*}
 \E \left|f(R_l^{-1}(X_{-l}^t; r_t) \mid X_{-l}) - f(R_l^{-1}(X_{-l}^t; r) \mid X_{-l}) \right| &\leq 2 f_{\max} \pr(\mathcal{E}_{\epsilon,t}) + \E\left[\Delta f_l(\epsilon \mid X_{-l}) \mid \mathcal{E}_{\epsilon,t}^c \right] \pr(\mathcal{E}_{\epsilon,t}^c) \\
& \leq 2 f_{\max} \pr(\mathcal{E}_{\epsilon,t}) + \E\left[\Delta f_l(\epsilon \mid X_{-l}) \mid \mathcal{E}_{\epsilon,t}^c \right],
\end{align*}
where $\Delta f_l(\cdot \mid x_{-l})$ is the modulus of continuity of $x_l \mapsto f(x_l \mid x_{-l})$ on $A_{\delta, x_{-l}}^l$. Applying the reverse Fatou's lemma, we derive 
\begin{align*}
& \quad \limsup \E \left|f(R_l^{-1}(X_{-l}^t; r_t) \mid X_{-l}) - f(R_l^{-1}(X_{-l}^t; r) \mid X_{-l}) \right| \\
& \leq 2 f_{\max} \limsup \pr(\mathcal{E}_{\epsilon,t}) + \limsup  \E\left[\Delta f_l(\epsilon \mid X_{-l}) \mid \mathcal{E}_{\epsilon,t}^c \right] \\
&= \limsup  \E\left[\Delta f_l(\epsilon \mid X_{-l}) \mid \mathcal{E}_{\epsilon,t}^c \right] \\
&\leq \Delta f_l(\epsilon), \quad \text{for all sufficiently small $\epsilon > 0$}.
\end{align*}
For the penultimate step, we used \cref{lem:pr-equi} to conclude $\pr(\mathcal{E}_{\epsilon,t}) \rightarrow 0$. For the final step, we used the equicontinuity of $f$ on $A_{\delta}^l$ (\cref{assump:f-equi}) to conclude that $\Delta f_l(\epsilon \mid x_{-l})$ for $x_{-l} \in [0,1]^{L-1}$ admits an upper bound $\Delta f_l(\epsilon)$ for all sufficiently small $\epsilon > 0$. Because $\Delta f_l(\epsilon)$ can be made arbitrarily small by choosing $\epsilon \searrow 0$, \cref{eqs:pf-equi-f-expectation} is proven. 

Finally, \cref{eqs:pf-equi-h-expectation} can be shown by a similar argument, which together with \cref{eqs:pf-equi-f-expectation} proves \cref{eqs:pf-equi-V}. To complete the proof, observe that the continuity of the derivative $\varphi'_{g=0}(h)$ on $[a,b]$ follows from the fact that (II-b) in \cref{eqs:pf-II-ab} tends to zero. 
\end{proof}

\subsection{First-order behaviour of rank-transformed subsampling} \label{app:rank-first-order}
For technical reasons, \cref{thm:rank-asymp} is established for $\tilde{G}_n$ computed from a variant of \cref{alg:agg-test} that has two independent copies of the data, $X$ and $X'$. The variant first uses $X$ to define the rank transform, and then uses $X'$ to approximate the distribution of the rank-transformed statistic, whereas the original \cref{alg:agg-test} uses the same $X$ for both purposes. In this variant described below, the randomness pertaining to the first and the second subsampling are independent and this greatly simplifies our analysis. 

\begin{description}
\item[First subsampling (for defining the rank transform)]
Recall that we choose $B = \JJ\, \floor{n/m}$ for a fixed integer $\JJ$. Also recall that $\hat{\mathbf{H}} = (\hat{H}_{b,l})$ is a $B \times L$ matrix consisting of rows 
\[\hat{\mathbf{H}}_{b,\cdot} := \left( T_m^{(1)}(X_{i_{1,b}}, \dots, X_{i_{m,b}}), \quad \dots \quad , T_m^{(L)}(X_{i_{1,b}}, \dots, X_{i_{m,b}}) \right), \quad b=1,\dots,B. \]
Suppose the $B$ rows are arranged into $J$ blocks of $\floor{n/m}$ rows so that $b=(j-1) \floor{n/m} + 1, \dots, j \floor{n/m}$ consists of $\floor{n/m}$ independent subsamples. Let $\mathbb{F}_{\hat{\mathbf{H}}}$ be the empirical distribution function of $\{\hat{H}_{b,l}: b=1,\dots,B,\,l=1,\dots, L\}$. The rank-transform is the following map from $T_m$ to $\tilde{T}_m$:
\begin{equation} \label{eqs:rank-trans}
\tilde{T}_m = F_0^{-1}\left(\frac{BL}{BL+1} \mathbb{F}_{\hat{\mathbf{H}}}(T_m) + \frac{1/2}{BL + 1}\right),
\end{equation}
where the normalised rank is bounded away from zero and one by $(1/2) / (BL+1)$. 

\item[Second subsampling (for estimating the distribution of rank-transformed statistic)]
We study the version of the procedure where the second subsampling is performed on $X'$, an independent copy of the data. The rank-transformed matrix $\tilde{\mathbf{H}} = (\tilde{H}_{b,l})$ consists of rows 
\[\tilde{\mathbf{H}}_{b,\cdot} := \left( \tilde{T}_m^{(1)}(X'_{i_{1,b}}, \dots, X'_{i_{m,b}}), \quad \dots \quad , \tilde{T}_m^{(L)}(X'_{i_{1,b}}, \dots, X'_{i_{m,b}}) \right), \quad b=1,\dots,B, \]
where $\tilde{T}_m^{(l)}(X'_{i_{1,b}}, \dots, X'_{i_{m,b}})$ is the image of $T_m^{(l)}(X'_{i_{1,b}}, \dots, X'_{i_{m,b}})$ under the rank transform. 
Passing them through the aggregation function $S$, we get
\[ \tilde{S}_b := S(\tilde{H}_{b,1}, \dots, \tilde{H}_{b,L}), \quad b=1,\dots,B \]
Let $\tilde{G}_{n}(x) := \mathbb{F}_{\{\tilde{S}_b\}}(x)$. For $\alpha \in (0,1)$, the rank-transformed critical value is $\tilde{G}_n^{-1}(1-\alpha)$. We compare this critical value to the oracle critical value $G_{P_0}^{-1}(1-\alpha)$, where $G$ is the asymptotic null distribution function of $S$. 

\end{description}

The full version of \cref{thm:rank-asymp} is stated as follows. 

\begin{theorem}[Full statement of \cref{thm:rank-asymp}] \label{thm:rank-asymp-full}
Suppose \cref{cond:pivotal} holds and $(T_n^{(1)}, \dots, T_n^{(L)})$ is exchangeable.
Consider a sequence $P_n \in \mathcal{P}$ that converges in copula to some $P_0 \in \mathcal{P}_0$ in the sense of \cref{def:copula-conv} such that 
\begin{equation} \label{eqs:copula-conv-tv}
\TV(U_m, C) = o(\sqrt{m/n}).
\end{equation}
Suppose the distribution of $C$ is absolutely continuous with respect to the Lebesgue measure. 
Let $S$ be a symmetric aggregation function that is non-decreasing in each coordinate. 
Let $\tilde{G}_n$ denote the rank-transformed subsampling distribution function obtained with a variant of \cref{alg:agg-test} that uses two independent copies of the data under $(X,X',\Omega) \sim P_n^n \times P_n^n \times P_{\Omega}$.

Suppose \cref{assump:stable-G} holds and fix $\alpha \in (0,1)$ such that the density $G'_{P_0}$ is strictly positive and continuous in a neighbourhood of $G_{P_0}^{-1}(1-\alpha)$.
With $c:=(c_1, \dots, c_L)$, define
\[ R(c_1,\dots,c_L):=S(F_0^{-1}(c_1), \dots, F_0^{-1}(c_L)), \quad c \in [0,1]^L. \]
Then, under \cref{assump:f-equi-symm,assump:f-cont-symm}, it holds that 
\begin{equation} \label{eqs:crit-val-asymp-rank}
\E \left[-M \vee \sqrt{n/m} \, \left(\tilde{G}_n^{-1}(1-\alpha) - G_{P_0}^{-1}(1-\alpha) \right) \wedge M \right] \rightarrow 0.
\end{equation}

Further, let $G_{n,P_n}$ be the distribution function of $S(T_n^{(1)}, \dots, T_n^{(L)})$ under $P_n^n \times P_{\Omega}$. Suppose $\|G_{n,P_n} - G_{\text{alt}}\|_{\infty} = o(\sqrt{m/n})$ holds for some distribution function $G_{\text{alt}}$ that is differentiable at $G_{P_0}^{-1}(1-\alpha)$. Then, under the same set of assumptions, for any $M>0$, we also have 
\begin{equation} \label{eqs:pow-asymp-rank}
\E \left[-M \vee \sqrt{n/m} \, \left(G_{n,P_n}(\tilde{G}_n^{-1}(1-\alpha)) - G_{n,P_n}(G_{P_0}^{-1}(1-\alpha)) \right) \wedge M \right] \rightarrow 0.
\end{equation}
\end{theorem}

\begin{proof}
By \cref{lem:lim-symm}, we prove \cref{eqs:crit-val-asymp-rank} by showing given any subsequence $n_k$ of $n$, $n_k$ admits a further subsequence $n_{\iota}$ such that $\sqrt{n_{\iota}/m_{\iota}} \, \left(\tilde{G}_{n_{\iota}}^{-1}(1-\alpha) - G_{P_0}^{-1}(1-\alpha)\right)$ converges to a symmetric law. Similarly, we establish \cref{eqs:pow-asymp-rank} by showing that 
\[\sqrt{n/m} \, \left(G_{n,P_n}(G_{P_0}^{-1}(1-\alpha)) - G_{n,P_n}(\tilde{G}_n^{-1}(1-\alpha)) \right)\] converges to a symmetric limit along the subsequence indexed by $\iota$. We will first prove \cref{eqs:crit-val-asymp-rank}, after which \cref{eqs:pow-asymp-rank} follows easily. 
In what follows, let $n_k$ be an arbitrary subsequence of $n$. 

\paragraph{First subsampling} Let 
\[ U_{b,l} := F_{m,P_n}(\hat{H}_{b,l}) = F_{m,P_n}\left(T(X_{i_{1,b}}, \dots, X_{i_{m,b}}; \Omega^{(b,l)}) \right), \quad b=1,\dots,B,\; l=1,\dots,L. \]
By construction, $\hat{H}_{b,l} = F_{m,P_n}^{-1}(U_{b,l})$. 
Because $\hat{H}_{b,l} =  F_{m,P_n}^{-1}(U_{b,l}) \leq x$ iff $U_{b,l} \leq F_{m,P_n}(x)$ for every $x$ \citep[Lemma 21.1]{van2000asymptotic}, $\mathbb{F}_{\hat{\mathbf{H}}}$ can be written as 
\begin{equation} \label{eqs:proof-U}
\mathbb{F}_{\hat{\mathbf{H}}}(x) = \mathbb{F}_{\mathbf{U},n} \circ F_{m,P_n}(x).
\end{equation}
With reindexing 
\[ U_{i,j,l} := U_{(j-1) \floor{n/m} + i,\, l}, \quad i=1,\dots,\floor{n/m}, \; j=1,\dots,J, \; l=1,\dots,L, \]
observe that for each $(j,l)$, $U_{1,j,l}, \dots, U_{\floor{n/m},j,l}$ are iid with a common distribution function, which we denote as $F_{U_{m}^{(1)}, P_n}$. Using \cref{eqs:copula-conv-tv}, we have
\[ \|F_{U_{m}^{(1)}, P_n} - \Id\|_{\infty} \leq \TV(U_m^{(1)}, C_1) \leq \TV(U_m, C) = o(\sqrt{m/n}). \]
Applying \cref{lem:bb} to the subsequence $n_k$, we conclude that $n_k$ admits a further a subsequence $n_j$ such that 
\begin{equation} \label{eqs:proof-subseq-1}
\sqrt{n_j/m_j} \left(\mathbb{F}_{\mathbf{U},n_j} - \Id \right) \rightsquigarrow \xi_1 \quad \text{in $D[0,1]$},
\end{equation}
where $\xi_1$ is almost surely a uniformly continuous with $\xi_1(0) = \xi_1(1) = 0$ and satisfies $\xi_1 =_d -\xi_1$.

\paragraph{Second subsampling, population version}
Define
\begin{equation} \label{eqs:pf-gn}
g_n(u) := (\mathbb{F}_{\mathbf{U},n} - \Id)(u) - \frac{1}{BL+1} \mathbb{F}_{\mathbf{U},n}(u) + \frac{1/2}{BL + 1},
\end{equation}
which is a random element in $D_{\varphi} \subset D[0,1]$ (see \cref{cor:hada-aggregate} for definition). 
For $l=1,\dots,L$, the rank transform can be written as $\tilde{T}_{m}^{(l)} = F_0^{-1}(\tilde{U}_m^{(l)})$ with 
\begin{align*}
\tilde{U}_m^{(l)}&=\frac{BL}{BL+1}\,\mathbb{F}_{\hat{\mathbf{H}}}(T_m^{(l)})  + \frac{1/2}{BL + 1}\\
&\stackrel{(i)}{=} \frac{BL}{BL+1}\, \mathbb{F}_{\mathbf{U},n} \circ F_{m,P_n} \circ F_{m,P_n}^{-1}(U_m^{(l)})  + \frac{1/2}{BL + 1} \\
&\stackrel{(ii)}{=}  \frac{BL}{BL+1}\, \mathbb{F}_{\mathbf{U},n}(U_m^{(l)})  + \frac{1/2}{BL + 1} \\
&= U_m^{(l)} + g_n(U_m^{(l)}),
\end{align*}
where (i) uses \cref{eqs:copula-m,eqs:proof-U}, (ii) uses the fact that $U_m^{(l)}$ is in the range of $F_{m,P_n}$ by construction \citep[Lemma 21.1]{van2000asymptotic}. In above, by the independence between $X$ and $X'$, we know $T_m^{(l)}$ (and hence $U_m^{(l)}$) is independent of $\mathbb{F}_{\mathbf{U},n}$ (and hence $g_n$). 
Further, the rank-transformed aggregated statistic 
\[ \tilde{S}_m = S(\tilde{T}_{m}^{(1)}, \dots, \tilde{T}_{m}^{(L)}) = R(\tilde{U}_m^{(1)}, \dots, \tilde{U}_m^{(L)}), \]
has its population distribution given by 
\begin{align*}
\tilde{G}_{\mathbb{F}_{\mathbf{U},n}}(r) &:= \pr\left(R(\tilde{U}_m^{(1)}, \dots, \tilde{U}_m^{(L)}) \leq r \mid \mathbb{F}_{\mathbf{U},n} \right) \\
&= \pr\left(R(U_m^{(1)} + g_n(U_m^{(1)}), \dots, U_m^{(L)} + g_n(U_m^{(L)})) \leq r \mid \mathbb{F}_{\mathbf{U},n} \right).
\end{align*}

Observe that $G_{P_0}$ is simply the distribution function of $S(F_0^{-1}(C_1), \dots, F_0^{-1}(C_L))$ under \cref{cond:pivotal,assump:stable-G}. 
Using the definition of $R(\cdot)$, we can write
\[ G_{P_0}(r) = \pr\left(R(C_1, \dots, C_L) \leq r \right). \]
Using \cref{eqs:proof-subseq-1} and $B = J \floor{n/m}$, we have
\begin{equation} \label{eqs:pf-gn-xi1}
 \sqrt{n_j / m_j} (g_{n_j} - 0) \rightsquigarrow \xi_1. 
\end{equation}
By our assumption that $S$ is symmetric and monotone, $R$ is symmetric and non-decreasing in each coordinate. Also, by our assumption, $C$ has a distribution that is absolutely continuous w.r.t. the Lebesgue measure. Given \cref{assump:f-cont-symm,assump:f-equi-symm}, we can apply \cref{cor:hada-aggregate} and conclude that the map 
\[ \varphi: D_{\varphi} \rightarrow D[a,b], \quad \varphi(g) = r \mapsto \pr\left(R(C_1 + g(C_1), \dots, C_L + g(C_L)) \leq r \right) \]
is Hadamard differentiable at $g=0$ tangentially to any uniformly continuously $h \in D[0,1]$ with $h(0) = h(1) = 0$. These conditions are satisfied almost surely by $\xi_1$ in \cref{eqs:pf-gn-xi1}. 
Applying the functional delta method \citep[Theorem 20.8]{van2000asymptotic} along the subsequence $n_j$, we have the following weak convergence in $D[a,b]$:
\begin{multline} \label{eqs:pf-eta-xi-1}
r \mapsto \sqrt{n_j / m_j} \left\{\pr(R\left(C_1 + g_{n_j}(C_1), \dots, C_L + g_{n_j}(C_L)) \leq r \right) - G_{P_0}(r) \right\} \\
\rightsquigarrow \eta(\xi_1) := \left( r \mapsto -L \E\left[f(R^{-1}(C_{-1}; r) \mid C_{-1}) \, \xi_1(R^{-1}(C_{-1}; r))  \mid \xi_1 \right] \right),
\end{multline}
where the notation $\E[\cdot \mid \xi_1]$ highlights that the expectation is only taken over $C$. In addition, by \cref{cor:hada-aggregate}, $\eta(\xi_1)$ is almost surely a continuous function of $r$ on $[a,b]$.

Note that the total variation bound in \cref{eqs:copula-conv-tv} implies 
\begin{multline*}
\sup_{r} \bigg|\pr\left(R(U_m^{(1)} + g_{n}(U_m^{(1)}), \dots, U_m^{(L)} + g_{n}(U_m^{(L)})) \leq r \mid \mathbb{F}_{\mathbf{U},n}  \right) \\
- \pr\left(R(C_1 + g_{n}(C_1), \dots, C_L + g_{n}(C_L)) \leq r \right) \bigg| = o(\sqrt{m/n})
\end{multline*}
and it follows that 
\begin{equation} \label{eqs:pf-xi-1}
\sqrt{n_j / m_j} \left(\tilde{G}_{\mathbb{F}_{\mathbf{U},n}} - G_{P_0} \right) \rightsquigarrow \eta(\xi_1) \quad \text{in $D[a,b]$},
\end{equation}
where $\eta(\xi_1) =_{d} -\eta(\xi_1)$ because $\xi_1 =_d -\xi_1$ and $\eta(\xi_1)$ is linear in $\xi_1$.

\paragraph{Second subsampling, sample version} Now consider the second subsample $\{\tilde{S}_b:b=1,\dots,B\}$ with $\tilde{S}_b = S(\tilde{H}_{b,1}, \dots, \tilde{H}_{b,L})$. For every $\tilde{S}_b$, draw a random $U_b'$ independently according to
\[ U_b' \mid \tilde{S}_b, \; \mathbb{F}_{\mathbf{U},n} \sim \unif[\tilde{G}_{\mathbb{F}_{\mathbf{U},n}}(S_b-), \; \tilde{G}_{\mathbb{F}_{\mathbf{U},n}}(S_b)]. \]
Then, by construction, 
\[  U_b' \sim \unif(0,1), \quad \tilde{S}_b = \tilde{G}^{-1}_{\mathbb{F}_{\mathbf{U},n}}(U_b'), \]
and consequently, the rank-transformed subsampling critical value can be expressed as
\begin{equation} \label{eqs:proof-quant-repre}
\tilde{G}_n^{-1}(1-\alpha) =  \tilde{G}^{-1}_{\mathbb{F}_{\mathbf{U},n}} \circ \mathbb{F}_{\mathbf{U}',n}^{-1}(1-\alpha), \quad \alpha \in (0,1),
\end{equation}
where we use $\mathbb{F}_{\mathbf{U}',n}$ to denote the empirical distribution function of $\{U_b'\}$. If we reindex $U_{i,j}' := U'_b$ for $b = (j-1) \floor{n/m} + i$, we see that $U_{1,j}', \dots, U_{\floor{n/m},j}'$ are iid $\unif(0,1)$ for every $j$ irrespective of $\mathbb{F}_{\mathbf{U},n}$. It follows from \cref{lem:bb,eqs:pf-xi-1} that there exists a further subsequence $n_{\iota}$ of $n_j$, such that
\begin{equation} \label{eqs:pf-joint}
 \sqrt{n_{\iota}/m_{\iota}} \left(\mathbb{F}_{\mathbf{U}',n_{\iota}} - \Id,\, \tilde{G}_{\mathbb{F}_{\mathbf{U},n_\iota}} - G_{P_0} \right) \rightsquigarrow (\xi_2, \,\eta(\xi_1)) \quad \text{in $D[0,1] \times D[a,b]$},
\end{equation}
where $\xi_2$ is a random element in $D[0,1]$ such that $\xi_2 =_{d} -\xi_2$ and $\xi_2$ is uniformly continuous almost surely. Further, because the second subsampling is performed on a separate sample $X'$, we have $\xi_1 \indep \xi_2$. 

Using the first limit in \cref{eqs:pf-joint}, by \citet[Lemma 21.3]{van2000asymptotic} we have 
\begin{equation} \label{eqs:proof-q-limit}
\sqrt{n_{\iota}/m_{\iota}} \left(\mathbb{F}^{-1}_{\mathbf{U}',n_{\iota}}(1-\alpha) - (1-\alpha) \right) \rightarrow_{d} -\xi_2(1-\alpha) 
\end{equation}
and consequently, 
\begin{equation} \label{eqs:proof-lim-2}
\sqrt{n_{\iota}/m_{\iota}} \left(G^{-1}(\mathbb{F}^{-1}_{\mathbf{U}',n_{\iota}}(1-\alpha)) - G_{P_0}^{-1}(1-\alpha) \right) \rightarrow_d \frac{-\xi_2(1-\alpha)}{g_{P_0}(G_{P_0}^{-1}(1-\alpha))}
\end{equation}
by the delta method. 

Further, by our assumption, because the density $g_{P_0}:=G_{P_0}'$ is strictly positive and continuous in a neighbourhood $[a,b] \supset
 [G_{P_0}^{-1}(1-\alpha) - \delta, G_{P_0}^{-1}(1-\alpha) + \delta]$ for some $\delta>0$, using \citet[Lemma 21.4]{van2000asymptotic} and the functional delta method, from the second limit in \cref{eqs:pf-joint} we derive
\begin{equation} \label{eqs:pf-lim-G-inv}
\sqrt{n_\iota / m_\iota}\, (\tilde{G}^{-1}_{\mathbb{F}_{\mathbf{U},n_\iota}} - G_{P_0}^{-1}) \rightsquigarrow -(\eta(\xi_1) / g_{P_0}) \circ G_{P_0}^{-1} \quad \text{in $\ell^{\infty}[1-\alpha-\delta',1-\alpha+\delta']$}
\end{equation}
for some $\delta'>0$. 

In fact, \cref{eqs:proof-q-limit,eqs:proof-lim-2,eqs:pf-lim-G-inv} hold jointly in light of \cref{eqs:pf-joint}, which implies
\[ \left(\sqrt{n_\iota / m_\iota}\, (\tilde{G}^{-1}_{\mathbb{F}_{\mathbf{U},n_\iota}} - G_{P_0}^{-1}),\,\mathbb{F}^{-1}_{\mathbf{U}',n_{\iota}}(1-\alpha)  \right) \rightsquigarrow \left(-(\eta(\xi_1) / g_{P_0}) \circ G_{P_0}^{-1},\, 1-\alpha \right) \]
in $\ell^{\infty}[1-\alpha-\delta',1-\alpha+\delta'] \times \mathbb{R}$. 
Observe that $(h,z) \mapsto h(z)$, as a $\ell^{\infty} \times \mathbb{R} \rightarrow \mathbb{R}$ map, is continuous when $h$ is continuous at $z$. Recall from \cref{eqs:pf-eta-xi-1} that $\eta(\xi_1)$ is almost surely a continuous function of $r$ on $[a,b]$, which verifies the continuity of $-(\eta(\xi_1) / g_{P_0}) \circ G_{P_0}^{-1}$ at $1-\alpha$. 
Hence, we can apply the continuous mapping theorem \citep[Theorem 18.11]{van2000asymptotic} and conclude 
\begin{equation} \label{eqs:proof-lim-1}
\sqrt{n_{\iota} / m_{\iota}}\, (\tilde{G}^{-1}_{\mathbb{F}_{\mathbf{U},n_{\iota}}} - G_{P_0}^{-1}) \circ \mathbb{F}^{-1}_{\mathbf{U}',n_{\iota}}(1-\alpha) \rightarrow_{d}  -(\eta(\xi_1) / g_{P_0}) \circ G_{P_0}^{-1}(1-\alpha).
\end{equation}
Finally, from \cref{eqs:proof-quant-repre,eqs:proof-lim-2,eqs:proof-lim-1} it follows that
\begin{equation} \label{eqs:pf-crit-val-limit}
\begin{split}
& \quad \sqrt{n_{\iota}/m_{\iota}}\, \left(\tilde{G}_{n_{\iota}}^{-1}(1-\alpha) - G_{P_0}^{-1}(1-\alpha)\right) \\
&= \sqrt{n_{\iota}/m_{\iota}} \left(\tilde{G}^{-1}_{\mathbb{F}_{\mathbf{U},n_{\iota}}} - G_{P_0}^{-1} \right) \circ \mathbb{F}_{\mathbf{U}',n_{\iota}}^{-1}(1-\alpha) + \sqrt{n_{\iota}/m_{\iota}}\, \left(G_{P_0}^{-1}(\mathbb{F}^{-1}_{\mathbf{U}',n_{\iota}}(1-\alpha)) - G_{P_0}^{-1}(1-\alpha) \right) \\
&\rightarrow_{d} -(\eta(\xi_1) / g_{P_0}) \circ G_{P_0}^{-1}(1-\alpha) -\frac{\xi_2(1-\alpha)}{g_{P_0}(G_{P_0}^{-1}(1-\alpha))}.
\end{split}
\end{equation}
Observe that the limit is symmetric because $\eta(\xi_1) =_{d} -\eta(\xi_1)$, $\xi_2 =_{d} -\xi_2$ and $\xi_1 \indep \xi_2$. As this limit holds for $n_{\iota}$, which is a subsequence of $n_j$ and hence of $n_k$, this establishes our first result \cref{eqs:crit-val-asymp-rank}.  

We now show \cref{eqs:pow-asymp-rank} under the additional assumption that $\|G_{n,P_n} - G_{\text{alt}}\|_{\infty} = o(\sqrt{m/n})$ for a distribution function $G_{\text{alt}}$ that is differentiable at $G_{P_0}^{-1}(1-\alpha)$. Again, by \cref{lem:lim-symm}, it suffices to show that 
\[\sqrt{n_{\iota}/m_{\iota}} \, \left(G_{{n_{\iota}},P_{n_{\iota}}}(G_{P_0}^{-1}(1-\alpha)) - G_{n_{\iota},P_{n_{\iota}}}(\tilde{G}_{n_{\iota}}^{-1}(1-\alpha)) \right)\]
converges to a symmetric limit law. This simply follows from 
\begin{align*}
& \quad \sqrt{n_{\iota}/m_{\iota}} \, \left(G_{{n_{\iota}},P_{n_{\iota}}}(G_{P_0}^{-1}(1-\alpha)) - G_{n_{\iota},P_{n_{\iota}}}(\tilde{G}_{n_{\iota}}^{-1}(1-\alpha)) \right)\\
& = \sqrt{n_{\iota}/m_{\iota}} \, \left(G_{\text{alt}}(G_{P_0}^{-1}(1-\alpha)) - G_{\text{alt}}(\tilde{G}_{n_{\iota}}^{-1}(1-\alpha)) \right) + O_p\left(\sqrt{n_{\iota}/m_{\iota}} \|G_{n_{\iota},P_{n_{\iota}}} - G_{\text{alt}}\|_{\infty} \right) \\
& = \sqrt{n_{\iota}/m_{\iota}} \, \left(G_{\text{alt}}(G_{P_0}^{-1}(1-\alpha)) - G_{\text{alt}}(\tilde{G}_{n_{\iota}}^{-1}(1-\alpha)) \right) + o_p(1) 
\end{align*}
and an application of the delta method based on \cref{eqs:pf-crit-val-limit}.
\end{proof}

\subsection{First-order behaviour of ordinary subsampling} \label{app:sub-first-order}
The result below formalises the first-order bias in approximating the oracle critical value with ordinary subsampling and the ensuing power loss; see also \cref{app:power-numerical} for numerical demonstration of the bias.  
\begin{theorem}[First-order behaviour of ordinary subsampling] \label{thm:crit-ordinary}
Consider a sequence $P_n \in \mathcal{P}$. Let $G_{m,P_n}$ be the distribution function of $S(T_m^{(1)}, \dots, T_m^{(L)})$ under $P_n^m \times P_{\Omega}$. Fix $\alpha \in (0,1)$ such that
\begin{equation} \label{eqs:G-cont-convergence}
 g_{m,P_n}(G_{m,P_n}^{-1}(x_n)) \rightarrow c > 0, \quad \text{for every sequence $x_n \rightarrow 1-\alpha$},
\end{equation}
where $g_{m,P_n} := G'_{m,P_n}$. 
Let $\hat{G}_{n} := \mathbb{F}_{\{\hat{S}_{b}: b=1,\dots,B\}}$ be the ordinary subsampling (without the rank transform; see \cref{sec:rank-transformed-sub}) estimate of $G_{P_0}$. Then, for every $M>0$, it holds that 
\begin{equation} \label{eqs:crit-val-asymp-ordinary}
\E \left[-M \vee \sqrt{n/m} \, \left(\hat{G}_n^{-1}(1-\alpha) - G_{m,P_n}^{-1}(1-\alpha) \right) \wedge M \right] \rightarrow 0,
\end{equation}
where the expectation is taken under $P_n^n \times P_{\Omega}$.

Further, suppose that for some $\beta \in (0, 1/2]$, we have $(n/m)^\beta(G_{m,P_n}^{-1}(1-\alpha) - G_{P_0}^{-1}(1-\alpha)) \rightarrow \tau > 0$.
Suppose $\|G_{n,P_n} - G_{\text{alt}}\|_{\infty} = o((m/n)^{\beta})$ holds for some distribution function $G_{\text{alt}}$ continuously differentiable at $G_{P_0}^{-1}(1-\alpha)$. Then, under the same set of assumptions, for any $M>0$, we also have 
\begin{multline} \label{eqs:power-asymp-ordinary}
\E \bigg\{-M \vee \bigg[(n/m)^{\beta} \, \left\{ G_{n,P_n}(\hat{G}_n^{-1}(1-\alpha)) - G_{n,P_n}(G_{P_0}^{-1}(1-\alpha)) \right\} \\
- \tau \, G'_{\text{alt}}(G_{P_0}^{-1}(1-\alpha)) \bigg] \wedge M \bigg\} \rightarrow 0.
\end{multline}
\end{theorem}
\begin{proof}
Using \cref{lem:lim-symm}, we prove \cref{eqs:crit-val-asymp-ordinary} by showing that given any subsequence $n_k$ of $n$, $n_k$ admits a further subsequence $n_{\iota}$ such that $\sqrt{n_{\iota}/m_{\iota}} \, \left(\hat{G}_{n_{\iota}}^{-1}(1-\alpha) - G_{m_{\iota},P_{n_{\iota}}}^{-1}(1-\alpha) \right)$ converges to a symmetric law. In what follows, let $n_k$ be an arbitrary subsequence of $n$. 

With $\hat{U}_{b,n} := G_{m,P_n}(\hat{S}_{b,n})$ for $b=1,\dots,B$, we have $\hat{S}_{b,n} = G_{m,P_n}^{-1}(\hat{U}_{b,n})$ and 
\begin{equation} \label{eqs:pf-G-hat}
\hat{G}_n^{-1}(1-\alpha) = G_{m,P_n}^{-1} \circ \mathbb{F}_{\{\hat{U}_{b,n}\}}^{-1}(1-\alpha), 
\end{equation}
where $\mathbb{F}_{\{\hat{U}_{b,n}\}}$ is the empirical distribution of $\{\hat{U}_{b,n}\}$. 
Further, if we reindex $\hat{U}_{i,j,n}' := \hat{U}_{b,n}$ for $b = (j-1) \floor{n/m} + i$, we see that $\hat{U}_{1,j,n}', \dots, U_{\floor{n/m},j,n}'$ are iid $\unif(0,1)$ for every $j=1,\dots,J$. By \cref{lem:bb}, there exists a further subsequence $n_{\iota}$ of $n_k$ such that
\[\sqrt{n_{\iota}/m_{\iota}} \left(\mathbb{F}_{\{\hat{U}_{b,n_\iota}\}} - \Id \right) \rightsquigarrow \xi \quad \text{in $D[0,1]$},\]
where $\xi =_{d} -\xi$. It then follows from \citet[Lemma 21.3]{van2000asymptotic} that 
\begin{equation} \label{eqs:pf-q-limit-2}
\sqrt{n_{\iota}/m_{\iota}} \left(\mathbb{F}^{-1}_{\{\hat{U}_{b,n_\iota}\}}(1-\alpha) - (1-\alpha) \right) \rightarrow_{d} -\xi(1-\alpha),
\end{equation}
where the limit law is symmetric. 

Using $(G_{m,P_n}^{-1})' = 1 / \left(g_{m,P_n} \circ G_{m,P_n}^{-1} \right)$ and the mean value theorem, from \cref{eqs:pf-G-hat} we derive 
\[ \sqrt{n_{\iota}/m_{\iota}} \left(\hat{G}_{n_\iota}^{-1}(1-\alpha) - G_{m_\iota,P_{n_\iota}}^{-1}(1-\alpha) \right) = C_{n_\iota}^{-1} \, \sqrt{n_\iota/m_\iota} \left[\mathbb{F}_{\{\hat{U}_{b,n_\iota}\}}^{-1}(1-\alpha) - (1-\alpha)\right] \]
for $C_{n_\iota} := g_{m_{\iota},P_{n_\iota}} \circ G_{m_{\iota},P_{n_{\iota}}}^{-1}(\zeta_{n_\iota})$ and some $\zeta_{n_\iota} \rightarrow_{p} 1-\alpha$. By our assumption \cref{eqs:G-cont-convergence}, we have $C_{n_\iota} \rightarrow_p c > 0$ by \citet[Theorem 18.11]{van2000asymptotic} and hence
\begin{equation} \label{eqs:pf-limit-ordinary}
\sqrt{n_{\iota}/m_{\iota}} \left(\hat{G}_{n_{\iota}}^{-1}(1-\alpha) - G_{m_{\iota},P_{n_{\iota}}}^{-1}(1-\alpha) \right) \rightarrow_{d} -c^{-1} \, \xi(1-\alpha),  
\end{equation}
where the limit law is symmetric. Because $n_{\iota}$ is a subsequence of $n_k$, this proves \cref{eqs:crit-val-asymp-ordinary}.

\medskip Finally, we show \cref{eqs:power-asymp-ordinary} under the given additional assumptions. Let $N$ be the neighbourhood of $G_{P_0}^{-1}(1-\alpha)$ on which $G_{\text{alt}}$ is differentiable. By the mean value theorem, we have that
\begin{equation} \label{eqs:pf-mean-val}
G_{\text{alt}} \circ \hat{G}_{n_{\iota}}^{-1}(1-\alpha) - G_{\text{alt}} \circ G_{m_{\iota},P_{n_{\iota}}}^{-1}(1-\alpha)  =  G_{\text{alt}}'(\xi_n) \left(\hat{G}_{n_{\iota}}^{-1}(1-\alpha) - G_{m_{\iota},P_{n_{\iota}}}^{-1}(1-\alpha) \right)
\end{equation}
on the sequence of events $\Omega_{n} := \{\xi_n \in N\}$, where $\xi_n$ is a random variable lying between $\hat{G}_{n_{\iota}}^{-1}(1-\alpha)$ and $G_{m_{\iota},P_{n_{\iota}}}^{-1}(1-\alpha)$. By~\cref{eqs:pf-limit-ordinary} and 
as
\begin{equation} \label{eqs:GP_n_G_P_0ass}
(n/m)^\beta(G_{m,P_n}^{-1}(1-\alpha) - G_{P_0}^{-1}(1-\alpha)) \rightarrow \tau > 0,
\end{equation}
 we have that $\xi_n \rightarrow_{p}  G_{P_0}^{-1}(1-\alpha)$. Thus $\mathbb{P}(\Omega_n) \to 1$ and by the continuous mapping theorem, $G_{\text{alt}}'(\xi_n) \rightarrow_{p} G_{\text{alt}}' \circ G_{P_0}^{-1}(1-\alpha)$. Thus, multiplying both sides of \cref{eqs:pf-mean-val} by $\sqrt{n_{\iota} / m_{\iota}}$, using Slutsky's lemma and \cref{eqs:pf-limit-ordinary} we have
\begin{equation} \label{eqs:pf-ord-limit}
\sqrt{\frac{n_{\iota}}{m_{\iota}}}  \left(G_{\text{alt}} \circ \hat{G}_{n_{\iota}}^{-1}(1-\alpha) - G_{\text{alt}} \circ G_{m_{\iota},P_{n_{\iota}}}^{-1}(1-\alpha)  \right) \rightarrow_{d}  -c^{-1} G_{\text{alt}}' \circ G_{P_0}^{-1}(1-\alpha) \, \xi(1-\alpha).
\end{equation}
Next, note that by our assumption \cref{eqs:GP_n_G_P_0ass} and the mean value theorem, we have that for all $n$ sufficiently large,
\begin{align*}
(n/m)^\beta\left\{G_{\text{alt}}(G_{m,P_{n}}^{-1}(1-\alpha)) - G_{\text{alt}}(G_{P_0}^{-1}(1-\alpha)) \right\} &= (n/m)^\beta \{ G_{m,P_{n}}^{-1}(1-\alpha) - G_{P_0}^{-1}(1-\alpha)\} G_{\text{alt}}'(\zeta_n) \\
&\to \tau \,G_{\text{alt}}'(G_{P_0}^{-1}(1-\alpha)),
\end{align*}
where $\zeta_n$ lies between $G_{P_0}^{-1}(1-\alpha)$ and $G_{m,P_{n}}^{-1}(1-\alpha)$. 
Using the fact that $\|G_{n,P_n} - G_{\text{alt}}\|_{\infty} = o((m/n)^\beta)$, we have
\begin{align*}
& \quad (n/m)^\beta \left\{ G_{n,P_n}(\hat{G}_n^{-1}(1-\alpha)) - G_{n,P_n}(G_{P_0}^{-1}(1-\alpha)) \right\} - \tau \,G_{\text{alt}}'(G_{P_0}^{-1}(1-\alpha)) \\
& = (n/m)^\beta \left\{ G_{\text{alt}}(\hat{G}_n^{-1}(1-\alpha)) - G_{\text{alt}}(G_{P_0}^{-1}(1-\alpha)) \right\} - \tau \,G_{\text{alt}}'(G_{P_0}^{-1}(1-\alpha)) + o(1) \\
&= (n/m)^\beta \left\{ G_{\text{alt}}(\hat{G}_n^{-1}(1-\alpha)) - G_{\text{alt}}(G_{m,P_n}^{-1}(1-\alpha)) \right\} \\
& \qquad + (n/m)^\beta \left\{ G_{\text{alt}}(G_{m,P_n}^{-1}(1-\alpha)) - G_{\text{alt}}(G_{P_0}^{-1}(1-\alpha)) \right\} - \tau \,G_{\text{alt}}'(G_{P_0}^{-1}(1-\alpha)) + o(1) \\
&= (n/m)^\beta \left\{ G_{\text{alt}}(\hat{G}_n^{-1}(1-\alpha)) - G_{\text{alt}}(G_{m,P_n}^{-1}(1-\alpha)) \right\}  + o_p(1).
\end{align*}
Using \cref{eqs:pf-ord-limit}, along the subsequence $n_{\iota}$, we have 
\begin{multline*}
(n_{\iota}/m_{\iota})^\beta \left\{ G_{n_{\iota},P_{n_{\iota}}}(\hat{G}_{n_{\iota}}^{-1}(1-\alpha)) - G_{{n_{\iota}},P_{n_{\iota}}}(G_{P_0}^{-1}(1-\alpha)) \right\} - \tau \,G_{\text{alt}}'(G_{P_0}^{-1}(1-\alpha)) \\
\rightarrow_{d} \begin{cases}
-c^{-1} G_{\text{alt}}' \circ G_{P_0}^{-1}(1-\alpha) \, \xi(1-\alpha), &\quad \beta = 1/2 \\
0, & \quad \beta < 1/2.
\end{cases}
\end{multline*}
The result then follows from applying \cref{lem:lim-symm}. 
\end{proof}

\subsection{Checking assumptions} \label{app:check-assump}
In this section, we check that for choices of $F_0$ and $S$ considered in this paper, \cref{assump:f-cont-symm,assump:f-equi-symm} are satisfied by a Gaussian copula. Consider a symmetric, Gaussian copula with a non-negative correlation (see \cref{fig:copula-normal}).
Observe that $c \mapsto f(c_1 \mid c_{-1})$ is continuous except for $(0,\dots,0)$ and $(1,\dots,1)$, the two points where $f$ diverges. This verifies \cref{assump:f-cont-symm}. Further, for $R^{-1}$ corresponding to choices of $F_0$ and $S$ considered in this paper (detailed below), it suffices to consider a region defined in \cref{assump:f-equi-symm} that is bounded away from $(0,\dots,0)$ and $(1,\dots,1)$. Because $c \mapsto f(c_1 \mid c_{-1})$ is uniformly continuous on this compact region, \cref{assump:f-equi-symm} holds.

We give $R^{-1}$ corresponding to choices of $F_0$ and $S$ considered in this paper. In terms of p-values ($F_0 = \Id$), the average p-value and the minimum p-value correspond to
\[R^{-1}_{\text{avg}}(c_{-1}; r) = 0 \vee \left(Lr - \sum_{i>1} c_i \right) \wedge 1, \quad R^{-1}_{\text{min}}(c_{-1}; r) = \begin{cases} 1, &\quad \min(c_{-1}) \leq r \\ r, & \quad \min(c_{-1}) > r \end{cases}.  \]
For aggregating Z-scores ($F_0 = \Phi$), we have
\[R^{-1}_{\text{avg}}(c_{-1}; r) = \Phi\left(Lr - \sum_{i>1} \Phi^{-1}(c_i) \right), \quad R^{-1}_{\text{min}}(c_{-1}; r) = \begin{cases} 1, &\quad \min(c_{-1}) \leq \Phi(r) \\ \Phi(r), & \quad \min(c_{-1}) > \Phi(r) \end{cases}.  \]

\begin{figure}[!htb]
\centering
\includegraphics[width=0.45\textwidth]{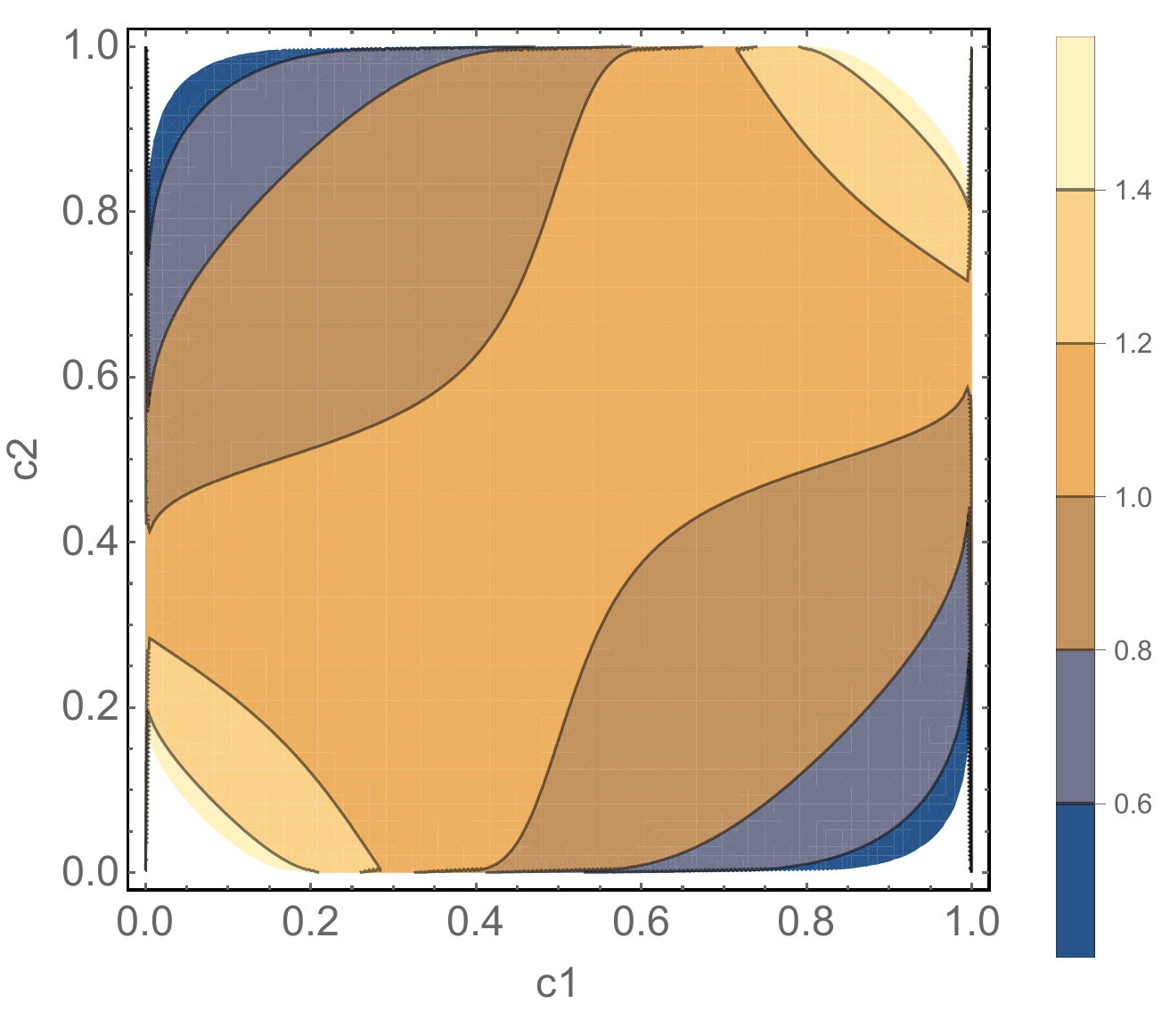}
\includegraphics[width=0.39\textwidth]{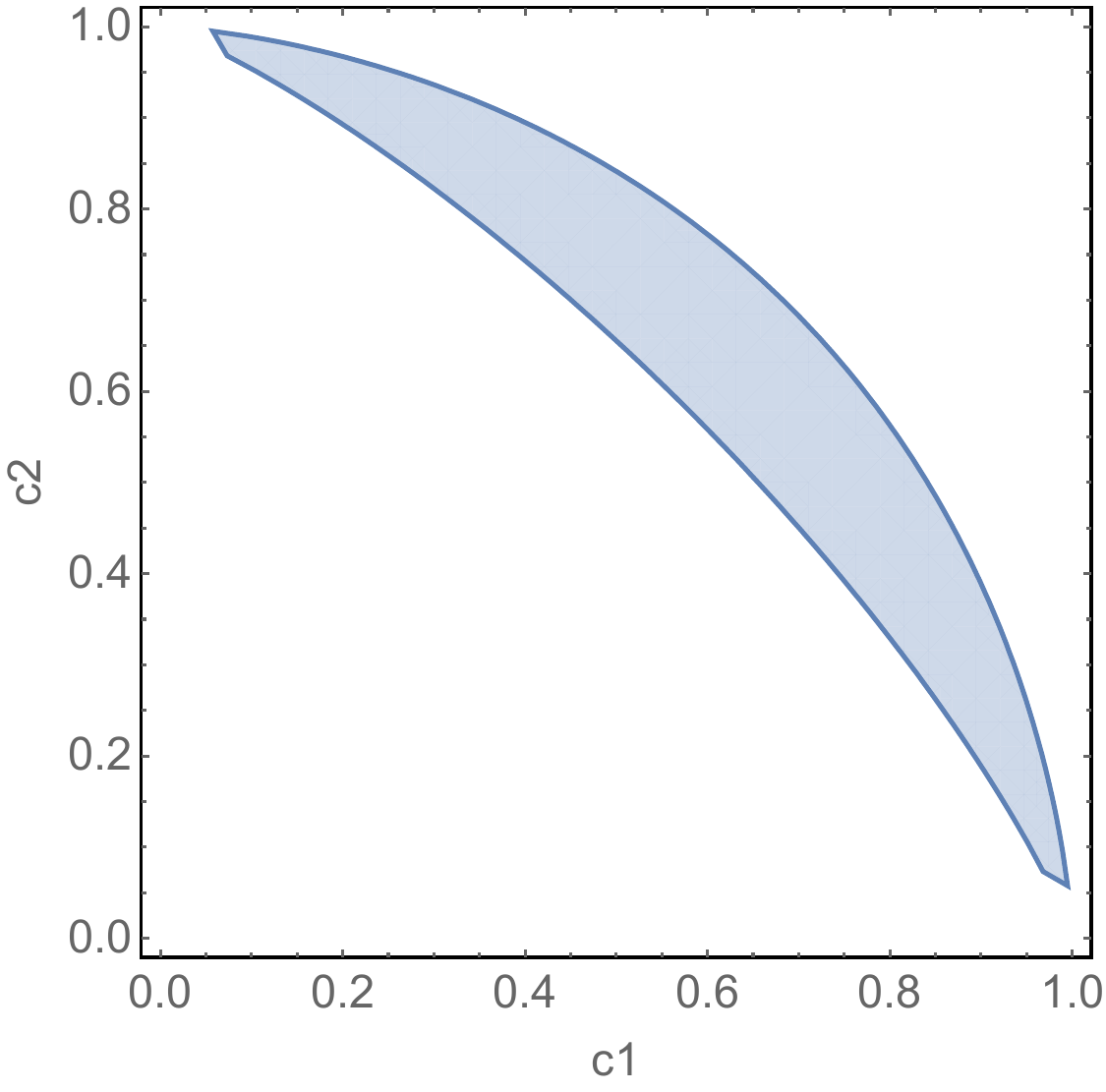}
\caption{Left: $f(c_1 \mid c_2)$ under Gaussian copula with $L=2$ and $\rho = 0.2$, where the function tends to $\infty$ at $(0,0)$ and $(1,1)$. Right: region defined in \cref{assump:f-equi-symm}, on which $f(c_1 \mid c_2)$ is continuous. }
\label{fig:copula-normal}
\end{figure}

\subsection{An illustrative example} \label{app:power-numerical}
In this section, we consider a simple toy example proposed by an anonymous referee to illustrate an application of our theory and the degree to which it matches numerical results. The numerical results also indicate that the conclusion of \cref{thm:rank-asymp} should reasonably hold for the original \cref{alg:agg-test}, although it is proved for a variant of the algorithm that uses two independent copies of the data.

\begin{example} \label{ex:app-normal}
Let $X_1, \dots, X_n$ be drawn iid from $\N(\mu, 1)$. Suppose that we want to test $H_0: \mu = 0$ against $H_1: \mu > 0$. We consider the single-split test statistic 
\[ T_n^{(1)}(X) := (np)^{-1/2} \sum_{i \in I^{(1)}} X_i,\]
where $I$ is a subset of size $\floor{np}$ drawn uniformly at random from $\{1,\dots,n\}$.  Clearly, $T_n^{(1)} \rightarrow_d \N(0,1)$ under $H_0$.

Let us consider the aggregated statistic $S_n := L^{-1} \sum_l T_n^{(l)}$. Under local alternatives $P_n = P_{\mu_n}$ with $\mu_n = c n^{-1/2}$ and $c>0$, the aggregated statistic satisfies
\begin{equation} \label{eqs:app-ex-S}
S_n \rightarrow_d \N(c \sqrt{p}, p + (1-p) / L).
\end{equation}
It follows that the oracle critical value is 
\[ C_{\alpha} := G_{P_0}^{-1}(1-\alpha) = z_{\alpha} \sqrt{p + (1-p) / L},  \]
where $z_{\alpha} = \Phi^{-1}(1-\alpha)$. Meanwhile, we have $G_{m,P_n} \approx \N(c \sqrt{p} \sqrt{m/n}, p + (1-p)/L)$ and hence 
\[ G_{m,P_n}^{-1}(1-\alpha) \approx c \sqrt{p} \sqrt{m/n} + G_{P_0}^{-1}(1-\alpha). \]
Roughly speaking, by taking $M$ to be a large constant, \cref{thm:rank-asymp,thm:crit-ordinary} imply
\begin{align*}
\E \sqrt{n/m}\, (\hat{C}_{\alpha}(\text{rank}) - C_{\alpha}) &= 0, \\
\E \sqrt{n/m}\, (\hat{C}_{\alpha}(\text{ord. sub.}) - C_{\alpha}) &= c \sqrt{p}.
\end{align*}
These results agree with our simulations plotted in the top and middle panels of \cref{fig:simu-first-order}: the top panel is from the two-sample variant of \cref{alg:agg-test} and the middle panel is from the original \cref{alg:agg-test}. A similar result also holds for other choices of $S$; see the bottom panel of the same figure. 

Moreover, we can compute the power functions as follows. Let $G_{n,P_n}$ be the distribution function of $S_n$. Given that the convergence rate in \cref{eqs:app-ex-S} should be faster than $\sqrt{m/n}$ (Berry--Esseen), the power of the oracle test is 
\begin{equation*}
\begin{split}
\text{pow}(\text{oracle}) = 1 - G_{n,P_n}(G_{P}^{-1}(1-\alpha)) &=  1 - \Phi\left( z_{\alpha} - c \sqrt{\frac{p}{p+(1-p)/L}}\right) + o(\sqrt{m/n})\\
\text{(when $L$ is large)} \quad & \approx \boxed{1 - \Phi(z_{\alpha} - c) + o(\sqrt{m/n}).}
\end{split}
\end{equation*}
Using \cref{thm:rank-asymp} (and taking $M$ to be a large constant), the power of rank-transformed subsampling is 
\begin{equation*} 
\begin{split}
\text{pow}(\text{rank}) &= 1 - \E G_{n,P_n}(\tilde{G}_n^{-1}(1-\alpha)) \\
&= \text{pow}(\text{oracle}) + o(\sqrt{m/n}) \\
&= 1 - \Phi\left( z_{\alpha} - c \sqrt{\frac{p}{p+(1-p)/L}}\right) + o(\sqrt{m/n}) \\
\text{(when $L$ is large)} \quad &\approx \boxed{1 - \Phi(z_{\alpha} - c) + o(\sqrt{m/n}).}
\end{split}
\end{equation*}
By \cref{thm:crit-ordinary}, where we use $\tau = \lim \sqrt{n/m}(G_{m,P_n}^{-1}(1-\alpha) - G_{P_0}^{-1}(1-\alpha)) = c \sqrt{p}$ and $G_{\text{alt}} = \N(c \sqrt{p}, p + (1-p)/L)$, the power of ordinary subsampling is 
\begin{equation*} 
\begin{split}
\text{pow}(\text{ord. sub.}) &= 1 - \E G_{n,P_n}(\hat{G}_n^{-1}(1-\alpha)) \\
&= 1 - G_{n,P_n}(G_{P}^{-1}(1-\alpha)) - \sqrt{m/n} \, \tau \, G'_{\text{alt}}(G_{P}^{-1}(1-\alpha)) + o(\sqrt{m/n}) \\
&= 1 - \Phi\left( z_{\alpha} - c \sqrt{\frac{p}{p+(1-p)/L}}\right) \\
& \quad - \sqrt{m/n} \, \phi \left( z_{\alpha} - c \sqrt{\frac{p}{p+(1-p)/L}}\right) c  \sqrt{\frac{p}{p+(1-p)/L}} + o(\sqrt{m/n})\\
\text{(when $L$ is large)} \quad &\approx \boxed{1 - \Phi(z_{\alpha} - c) - \sqrt{m/n} \, \phi(z_{\alpha} - c)\, c  +  o(\sqrt{m/n}).}
\end{split}
\end{equation*}
Note that $\phi(z_{\alpha} - c)\, c$ is increasing in $c$ for $c \in (0, z_{\alpha})$. 

\begin{figure}[t]
\centering
\includegraphics[width=0.85\textwidth]{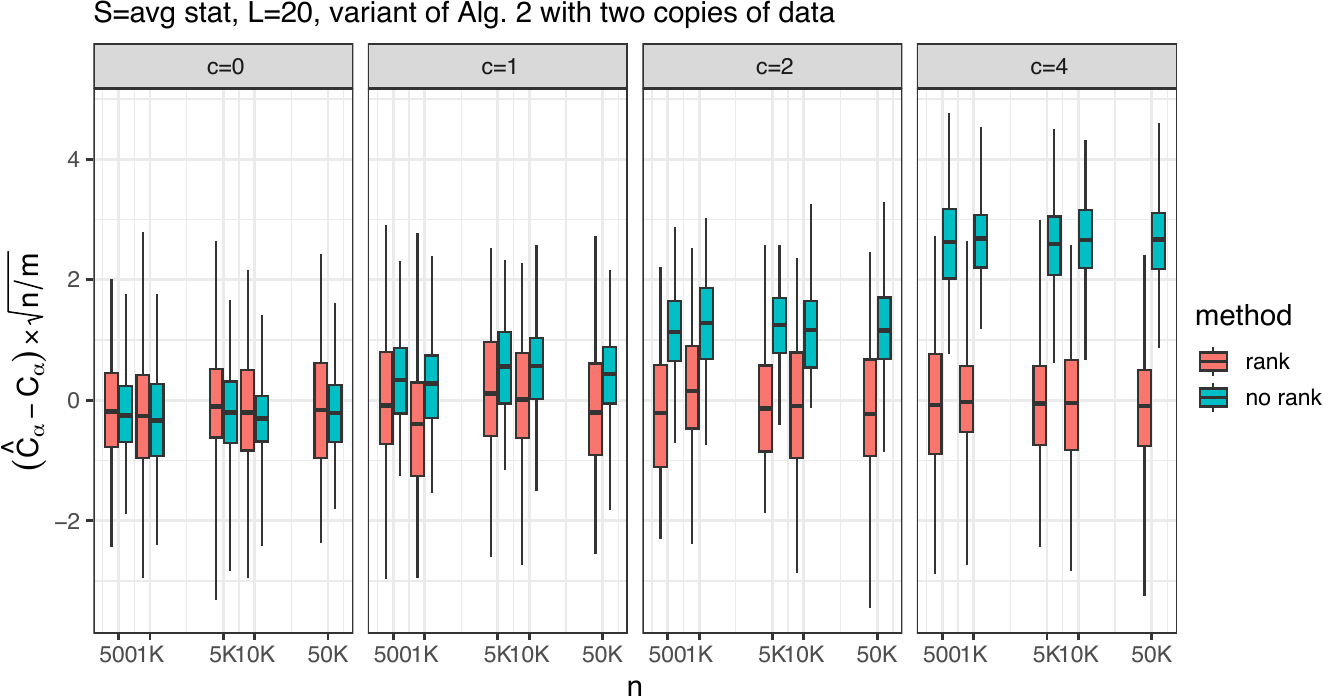}
\includegraphics[width=0.85\textwidth]{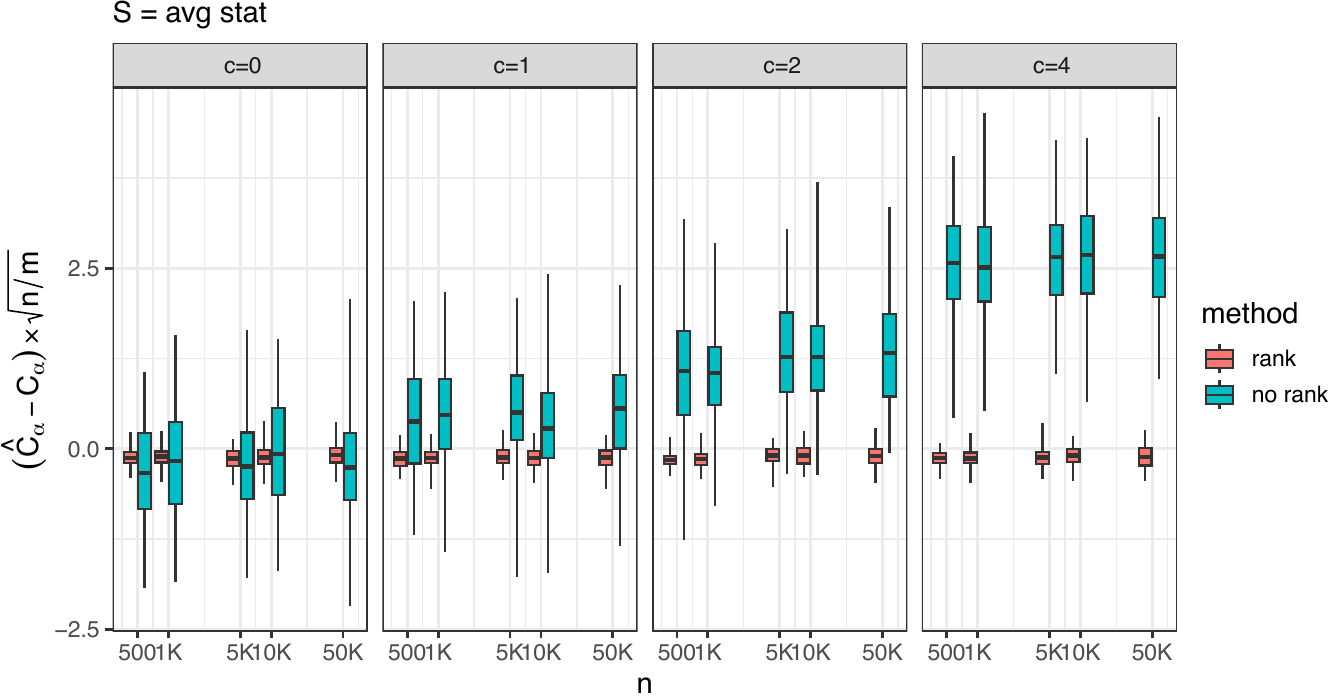}
\includegraphics[width=0.85\textwidth]{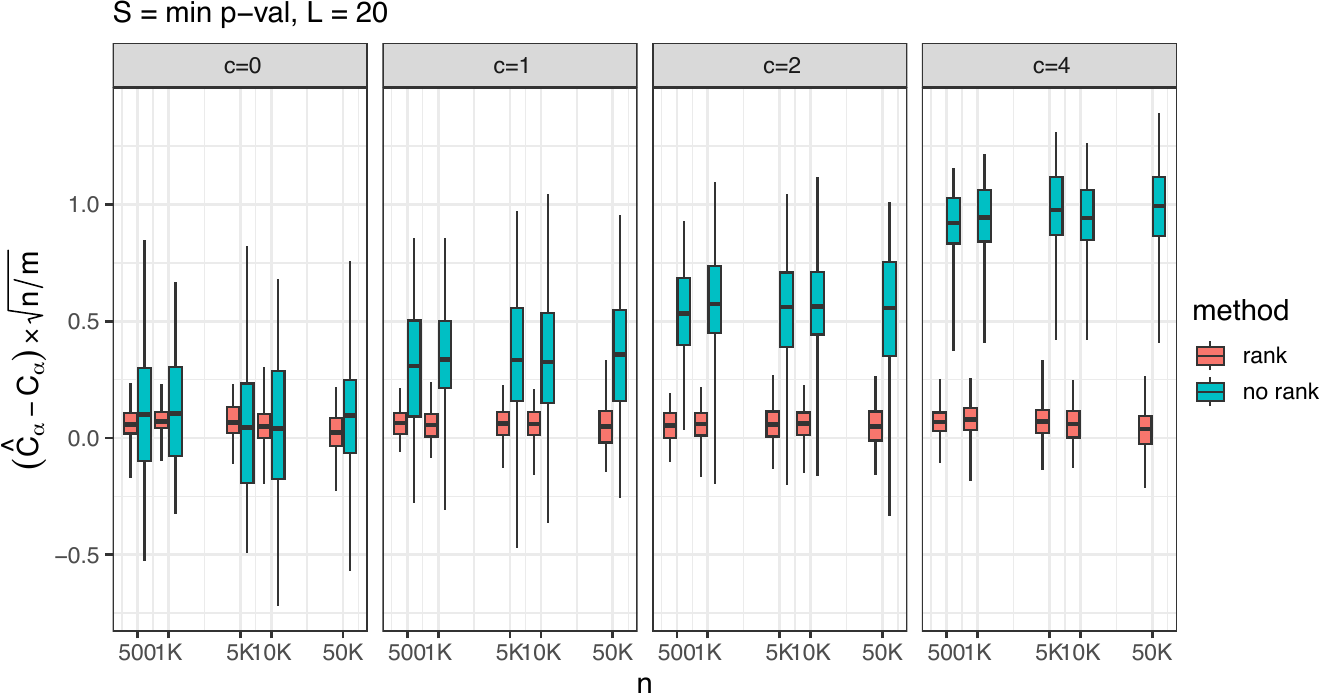}
\caption{Numerical results for \cref{ex:app-normal}: first-order bias in approximating the oracle critical value under $\mu = c n^{-1/2}$ (`\texttt{rank}': rank-transformed subsampling, `\texttt{no rank}': ordinary subsampling). The thick middle line in the box plots represents the mean value of the first-order bias $\sqrt{n/m} (\hat{C}_{\alpha} - C_{\alpha})$. Top: the variant of \cref{alg:agg-test} that uses two independent copies of $X$ with $S = \text{avg}$; Middle: \cref{alg:agg-test} with $S = \text{avg}$; Bottom:  \cref{alg:agg-test} with $S$ being (one minus) the minimum p-value.}
\label{fig:simu-first-order}
\end{figure}
\end{example}

\subsection{Auxiliary results} \label{app:aux-hadamard}
\begin{lemma} \label{lem:equicont-compact}
Let $f_n(x)$ be a sequence of real-valued functions on a compact, Euclidean set $K$. Suppose $f_n \rightarrow f$ pointwise on $K$. Suppose that for a dense subset $K' \subseteq K$,
\begin{equation} \label{eqs:local-equi}
|f_n(x_n) - f_n(x)| \rightarrow 0, \quad \text{for every $x \in K'$ and every sequence $(x_n)_{n \in \mathbb{N}} \subseteq K$ with $x_n \rightarrow x$}.
\end{equation}
Then, $f_n \rightarrow f$ uniformly on $K$. 
\end{lemma}
\begin{proof}
Without loss of generality, by subtracting $f$ from $f_n$ we can assume $f = 0$. 
We prove uniform convergence by showing that for any $\epsilon > 0$, there exists $N'(\epsilon) > 0$ such that $|f_n(x)| < \epsilon$ for every $x \in K$ and every $n > N'(\epsilon)$. 
First, as $f_n \rightarrow 0$ pointwise on $K$, for every $x \in K$, there exists $M(x, \epsilon) > 0$ such that $|f_n(x)| < \epsilon$ for every $n > M(x, \epsilon)$. Second, observe that condition \cref{eqs:local-equi} is equivalent to
\begin{multline*}
\forall x \in K', \epsilon > 0, \quad \exists\, \delta(x,\epsilon) > 0, N(x, \epsilon) > 0 \;\text{ s.t.,}\\
|f_n(x') - f_n(x)| < \epsilon, \quad \forall\, x' \in  K,\, n: \|x'-x\| < \delta(x, \epsilon),\; n > N(x, \epsilon).
\end{multline*}
Consider $\cup_{x \in K'} B_{\delta(x, \epsilon/2)}(x)$, which is an open cover of $K$ as $K'$ is dense in $K$. Because $K$ is compact, it must admit a finite subcover of $K$, which we denote as $\cup_{i \in I} B_{\delta(x_i, \epsilon/2)}(x_i)$ with every $x_i \in K$ and $|I| < \infty$. Hence, for every $x \in K$, there exists $\hat{x} \in \{x_i: i \in I\}$ such that $x \in B_{\delta(\hat{x},\epsilon/2)}(\hat{x})$. Then, for every $n > \max_{i \in I} N(x_i, \epsilon / 2) \vee \max_{i \in I} M(x_i, \epsilon / 2)$, we have
\[ |f_n(x)| \leq |f_n(x) - f_n(\hat{x})| + |f_n(\hat{x})| \leq \epsilon / 2 + \epsilon / 2 = \epsilon, \quad \forall x \in K,\]
which completes the proof.
\end{proof}

\begin{proposition} \label{prop:monotone-cont-ae}
The set of discontinuity points of a coordinatewise non-decreasing function $g:\prod_{k=1}^d [a_k, b_k] \to \R$ has Lebesgue measure zero, where $a_k < b_k$ for $k=1,\dots,d$.
\end{proposition}
\begin{proof}
	It suffices to show that for all $\epsilon > 0$,
	\[
	A_{\epsilon} := \bigcap_{\delta >0} \left\{x : \sup_{y \in B_{\delta}(x)} |g(y)-g(x)| \geq \epsilon \right\}
	\]
	is a Lebesgue  null set, where $B_\delta(x)$ denotes the $\ell_2$-norm closed ball of radius $\delta$ centred at $x$. Indeed, if $x$ is a discontinuity point, then there exists some $n \in \mathbb{N}$ such that  $x \in A_{1/k}$ for all integer $k \geq n$.

	Let $\epsilon > 0$ be given. Now for $j \in \mathbb{Z}$, let $D_j := g^{-1}((-\infty,\epsilon j])$ and denote by $\partial D_j$ its boundary. We claim that $\partial D_j$ is null. To see this, take $x \in \partial D_j$ and note that as $f$ is coordinatewise non-decreasing, $C:=\prod_{k=1}^d [a_k, x_k) \cup \prod_{k=1}^d (x_k, b_k]$ is disjoint from $\partial D_j$. But then, for all $r>0$, there exists $y$ such that $B_{cr}(y) \subset B_r(x) \cap (\partial D_j)^c$, so for all $r > 0$,
	\[
	\frac{\mu(\partial D_j \cap B_r(x))}{\mu(B_r(x))} = 1 - \frac{\mu((\partial D_j)^c \cap B_r(x))}{\mu(B_r(x))} \leq 1 - \frac{\mu(B_{cr}(y))}{\mu(B_r(x))} = 1 - c^d.
	\]
	Thus each $x \in \partial D_j$ has Lebesgue density (if it exists) strictly less than $1$. However, by the Lebesgue density theorem \citep[Cor.~2.14]{mattila1995geometry}, almost all $x \in \partial D_j$ have Lebesgue density
	\[
	\lim_{r \searrow 0} \frac{\mu(\partial D_j \cap B_r(x))}{\mu(B_r(x))} =1,
	\]
	so $\partial D_j$ must be a null set. Thus $D:=\cup_j \partial D_j$ is also a null set.
	
	We now claim that $D \supseteq A_{\epsilon}$, which will give the result. Take $x \in D^c$, so $x \in \text{int}(D_j)$, for some $j$; we take the minimal such $j$. Then for some $\delta > 0$, $B_{\delta}(x) \subset D_j \setminus \overline{D}_{j-1}$, where the overline denotes the closure. Thus for any $y \in B_{\delta}(x)$, we have
	\[
	\epsilon(j-1) < g(x), \, g(y) \leq \epsilon j,
	\]
	so $|g(y) - g(x)| < \epsilon$ and hence $x \notin A_\epsilon$, which completes the proof.
\end{proof}

\begin{lemma} \label{lem:pr-equi}
Under the assumptions of \cref{prop:hada-aggregate}, there exists a dense subset $K'$ of $[a,b]$, such that for any $\epsilon > 0$, $r \in K'$ and $[a,b] \ni r_t \rightarrow r$ as $t \rightarrow 0$, it holds that 
\[ \pr\left(|R_l^{-1}(X_{-l}^t; r_t) - R_l^{-1}(X_{-l}^t; r)| > \epsilon \right) \rightarrow 0. \]
\end{lemma}
\begin{proof}
By monotonicity of $R$, $R_l^{-1}(x_{-l}; r)$ is coordinatewise non-decreasing in $(-x_{-l}, r)$. From \cref{prop:monotone-cont-ae}, it follows that $(x_{-l}, r) \mapsto R_l^{-1}(x_{-l}; r)$ is continuous at almost every $(x_{-l}, r) \in [0,1]^{L-1} \times [a,b]$.
Stated equivalently, for every $r \in K'$, where $K'$ equals $[a,b]$ excluding a null set, and almost every $x_{-l} \in [0,1]^{L-1}$ (where the null set excluded can depend on $r$), $R_l^{-1}$ is continuous at $(x_{-l}, r)$. Clearly, $K'$ is dense in $[a,b]$.

Fix $\epsilon > 0$, $r \in K'$ and $[a,b] \ni r_t \rightarrow r$. By the reverse Fatou's lemma, we have 
\begin{equation*} 
\begin{split}
\limsup \pr\left(|R_l^{-1}(X_{-l}^t; r_t) - R_l^{-1}(X_{-l}^t; r)| > \epsilon \right) &\leq \pr\left(\{|R_l^{-1}(X_{-l}^t; r_t) - R_l^{-1}(X_{-l}^t; r)| > \epsilon\} \;\text{i.o.} \right) \\
&= \pr\left( \limsup |R_l^{-1}(X_{-l}^t; r_t) - R_l^{-1}(X_{-l}^t; r)| > \epsilon \right).
\end{split}
\end{equation*}
By the fact that $\|x_{-l}^t - x_{-l}\|_{\infty} \leq t \|h_t\|_{\infty} \rightarrow 0$ and the continuity of $R_l^{-1}$ argued above, for Lebesgue a.e. $x_{-l}$, we have $R_l^{-1}(x_{-l}^t; r_t) \rightarrow R_l^{-1}(x_{-l}; r)$, $R_l^{-1}(x_{-l}^t; r) \rightarrow R_l^{-1}(x_{-l}; r)$ and hence
\[\limsup |R_l^{-1}(x_{-l}^t; r_t) - R_l^{-1}(x_{-l}^t; r)| = \lim |R_l^{-1}(x_{-l}^t; r_t) - R_l^{-1}(x_{-l}^t; r)| = 0. \]
Finally, because $X_{-l}$ has a distribution absolutely continuous w.r.t. the Lebesgue measure, the result follows from $\pr\left( \limsup |R_l^{-1}(X_{-l}^t; r_t) - R_l^{-1}(X_{-l}^t; r)| > \epsilon \right) = 0$. 
\end{proof}

\begin{lemma} \label{lem:lim-symm}
Let $X_n$ be a sequence of real-valued random variables. If for every subsequence $n_k$, there is a further subsequence $n_{k(j)}$ such that $X_{n_{k(j)}}$ converges to a symmetric law (the law may depend on the subsequence), then we have
\[ \E[-M \vee X_n \wedge M] \rightarrow 0, \quad \text{for every $M>0$}. \]
\end{lemma}
\begin{proof}
Fix any $M>0$. By our assumption, for every subsequence $n_k$, there exists a further subsequence $n_{k(j)}$ such that $-M \vee X_{n_{k(j)}} \wedge M$ converges to a symmetric law. Because it is bounded, we have $\E[-M \vee X_{n_{k(j)}} \wedge M] \rightarrow 0$. The result follows because the limit does not depend on the subsequence. 
\end{proof}

\begin{lemma} \label{lem:bb}
Let $\{V_{i,k}: i=1,\dots,n,\; k=1,\dots, K,\}$ be random variables such that 
\[V_{1,k}, \dots, V_{n,k} \iid F_n, \quad k=1,\dots,K,  \]
where $F_n$ is a distribution function satisfying
\[ \sqrt{n} \|F_n - \Id\|_{\infty} \rightarrow 0. \]

Let $\mathbb{V}_n$ be the empirical distribution function of $\{V_{i,k}: i=1,\dots, n, \; k=1,\dots,K\}$. Then, the empirical process $\sqrt{n} (\mathbb{V}_n - \Id)$, as a random bounded function on $[0,1]$, is asymptotically tight and asymptotically measurable. Further, there exists a subsequence $n_j$ along which 
\[ \sqrt{n_j} (\mathbb{V}_{n_j} - \Id) \rightsquigarrow \xi \quad \text{in $D[0,1]$}, \]
where $\xi =_{d} -\xi$ and $\xi$ is almost surely uniformly continuous (the distribution of $\xi$ can depend on the subsequence) with $\xi(0) = \xi(1) = 0$. 
\end{lemma}
\begin{proof}
Let $\mathbb{V}_{n,k}$ be the empirical distribution function of $\{V_{i,k}: i=1,\dots,n\}$. For each $k$, it holds that $(V_{1,k}, \dots, V_{n,k}) =_d (F_n^{-1}(U_1), \dots, F_n^{-1}(U_n))$ for iid $U_1,\dots,U_n \sim \unif(0,1)$. Therefore, we have
\[ \mathbb{V}_{n,k}(x) =_d n^{-1} \sum_{i} \I\{F_n^{-1}(U_i) \leq x \} = n^{-1} \sum_{i} \I\{U_i \leq F_n(x) \} = \mathbb{U}_n \circ F_n(x),\]
where $\mathbb{U}_n$ is the empirical distribution function of $\{U_i\}$. It follows that 
\[ \sqrt{n}(\mathbb{V}_{n,k} - \Id) =_d \sqrt{n}(\mathbb{U}_{n} - \Id) + \sqrt{n}(F_n - \Id) + \sqrt{n}(\mathbb{U}_n - \Id)(F_n - \Id).\]
The empirical process $\sqrt{n}(\mathbb{U}_{n} - \Id)$ converges weakly to a standard Brownian bridge. By our assumption, $\sqrt{n} \|F_n - \Id\|_{\infty} \rightarrow 0$. By continuous mapping theorem, $\|\sqrt{n}(\mathbb{U}_n - \Id)(F_n - \Id)\|_{\infty} \rightarrow_p 0$. Therefore, we conclude that $\sqrt{n}(\mathbb{V}_{n,k} - \Id)$ converges weakly to a standard Brownian bridge. 
It then follows that the sequence $\sqrt{n} (\mathbb{V}_{n,k} - \Id)$ is asymptotically tight and asymptotically measurable. Therefore, 
\[ \sqrt{n}(\mathbb{V}_{n,1} - \Id, \dots, \mathbb{V}_{n,K} - \Id), \]
as a random bounded $[0,1] \rightarrow \mathbb{R}^K$ function, is also asymptotically tight and asymptotically measurable. By Prohorov's theorem \citep[18.12]{van2000asymptotic}, there exists a subsequence $n_j$ such that 
\[\sqrt{n_j}(\mathbb{V}_{n_j,1} - \Id, \dots, \mathbb{V}_{n_j,K} - \Id) \rightsquigarrow (\xi_1, \dots, \xi_K), \]
where every $\xi_k$ is a standard Brownian bridge. By the continuous mapping theorem, their average 
\[ \sqrt{n_j}(\mathbb{V}_{n_j} - \Id) \rightsquigarrow (\xi_1 + \dots + \xi_K) / K =: \xi.  \]
Because $\xi_k =_{d} - \xi_k $ for each $k$, we know $\xi =_{d} -\xi$. Because each $\xi_k$ is uniformly continuous and satisfies $\xi_k(0) = \xi_k(1) = 0$ almost surely \citep[\S19.1]{van2000asymptotic}, $\xi$ also almost surely satisfies these properties. 
\end{proof}
 
\clearpage
\section{Calibration of cross-fitted DML} \label{app:dml}
We first provide details for the numerical example in \cref{sec:dml}. In the partially linear model, we let $X \sim \N(0,1)$ and $s(X):=4 \sqrt{|X|}$. The heteroscedastic errors are drawn as 
\[ V = s(X)\, \epsilon_V, \quad \xi = s(X)\, \epsilon_{\xi} \]
with $\epsilon_V \sim \gam(0.5, 1) - 0.5$ and $\epsilon_V \sim \gam(0.3, 1) - 0.3$. The regression functions are chosen as
\begin{align*}
m_0(x) &= x + \cos(x) + \exp(x) / (1 + \exp(x)), \\
g_0(x) &= \left[ -10 x + 3 \cos(4x) x^2 / (1 + \exp(x / 6)) \right] / 10.
\end{align*}

\medskip Below is an example of performing DML when one of the two (parametric) nuisance models is misspecified. Because of the double robustness of the debiased score, the resulting cross-fitted DML estimator is still consistent and in fact asymptotically normal. However, unlike the scenario discussed in \cref{sec:dml} where the per-fold estimators $\hat{\theta}^{(1)}, \dots, \hat{\theta}^{(L)}$ are asymptotically uncorrelated (so correlation is only a finite-sample phenomenon), because of the misspecification, $\hat{\theta}^{(1)}$ and $\hat{\theta}^{(2)}$ in the example below are negatively correlated even in large sample. Hence, the standard plugin estimator for the asymptotic variance (see \cref{eqs:sigma-dml}) is inconsistent. Nevertheless, the estimator based on rank-transformed subsampling is still consistent. 

\begin{example}[DML under misspecification] \label{ex:dml-misspec}
Consider the following specification of the partially linear model considered in \cref{sec:dml}:
\begin{equation*}
\begin{split}
D = \beta_1 X + \beta_2 X^2 + V, \quad &\E[V \mid X] = 0, \\
Y = \theta_0 D + \gamma_1 X + \gamma_2 X^2 + \xi, \quad &\E[\xi \mid D, Y] = 0,
\end{split}
\end{equation*}
where $\theta_0$ is the parameter of interest. For simplicity, we assume $X, \epsilon_X, \epsilon_Y$ are drawn independently from $\N(0,1)$. The DML estimator of $\theta_0$, as the solution to the empirical Robinson's score function, is the least squares coefficient of regressing the residual $r_Y := Y - \hat{E}[Y \mid X]$ on the residual $r_D := D - \hat{E}[D \mid X]$, where the estimation of regression functions and the final least squares are performed on two separate parts of the sample. 

Suppose that $\hat{E}[D \mid X]$ is well-specified but $\hat{E}[Y \mid X]$ is misspecified as $\eta_1 X$ (missing $X^2$), whereas the true model is $\eta_1 X + \eta_2 X^2$ with $\eta_1  = \theta_0 \beta_1 + \gamma_1$ and $\eta_2  = \theta_0 \beta_2 + \gamma_2$. Consider the DML estimator $\hat{\theta}_\text{DML} = (\hat{\theta}^{(1)} + \hat{\theta}^{(2)}) / 2$ for the case of $L=2$ folds with 
\[ \hat{\theta}^{(l)} = \frac{\sum_{i \in I^{(l)}} r_{Y_i} r_{D_i}}{\sum_{i \in I^{(-l)}} r_{D_i}^2 }, \quad l=1,2, \]
where $(I^{(1)},  I^{(2)})$ is a random split of the sample into two parts of equal size. By an asymptotic linear expansion of both numerator and denominator and applying the delta method, it can be shown that $\hat{\theta}^{(1)}, \hat{\theta}^{(2)}$ admit the following asymptotic expansion:
\begin{equation*}
\begin{split}
\sqrt{n/2}\, (\hat{\theta}^{(1)} - \theta_0) &= \frac{1}{\sqrt{n/2}} \sum_{i \in I^{(1)}} \left(\eta_2 X_i^2 V_i + V_i \xi_i \right) - \frac{1}{\sqrt{n/2}} \sum_{i \in I^{(2)}} \eta_2 X_i^2 V_i + o_p(1), \\
\sqrt{n/2}\, (\hat{\theta}^{(2)} - \theta_0) &= \frac{1}{\sqrt{n/2}} \sum_{i \in I^{(2)}} \left(\eta_2 X_i^2 V_i + V_i \xi_i \right) - \frac{1}{\sqrt{n/2}} \sum_{i \in I^{(1)}} \eta_2 X_i^2 V_i + o_p(1).
\end{split}
\end{equation*}
Applying CLT to the RHS, we derive 
\[ \sqrt{n/2} \begin{pmatrix} \hat{\theta}^{(1)} - \theta_0 \\ \hat{\theta}^{(2)} - \theta_0 \end{pmatrix} \rightarrow_{d} \N\left(0, \begin{pmatrix} \sigma^2 & \rho \sigma^2 \\ \rho \sigma^2 & \sigma^2 \end{pmatrix} \right), \]	
where 
\begin{equation} \label{eqs:misspec-rho}
\sigma^2 = 6 \eta_2^2 + 1, \quad \rho = - 6 \eta_2^2 / (1 + 6 \eta_2^2). 
\end{equation}
Note that $-1 < \rho <0 $ whenever $\eta_2 \neq 0$, i.e., when $\hat{\E}[Y \mid X]$ is misspecified. It then follows that 
\[ \sqrt{n}\, (\hat{\theta}_\text{DML} - \theta_0) \rightarrow_d \N\left(0, (1+\rho) \sigma^2 \right). \]

In the below, we take $(\beta_1, \beta_2) = (1, 1/2)$, $(\gamma_1, \gamma_2) = (1,1)$, $\theta_0 = 1$ and let $V, \xi \sim \N(0,1)$. Asymptotic values of correlation and variance, along with values under finite sample, are compared to their estimates in \cref{tab:dml-misspec}. As can be seen from the table, the rank-transformed subsampling estimates are consistent while the standard DML plugin estimate is too large. See also \cref{tab:simu-dml-misspec} for the coverage of confidence intervals under $L=2$ and $L=5$. 

\begin{table}[!htb]
\centering
\caption{Correlation, standard deviation, their rank-transformed subsampling estimates, as well as the standard plugin estimate (via \cref{eqs:sigma-dml}) under model misspecification ($L=2$). The $n=\infty$ row refers to asymptotic values computed from \cref{eqs:misspec-rho}.} \label{tab:dml-misspec}
\begin{tabular}{cccccc}
\toprule
$n$ & $\rho$ & $\tilde{\rho}$ & $\sd \sqrt{n}(\hat{\theta}_{\text{DML}} - \theta_0)$ & $\hat{\sigma}_{\text{ls}} \sqrt{1 + \tilde{\rho}}$ & $\hat{\sd}_{\text{plugin}} \sqrt{n}(\hat{\theta}_{\text{DML}} - \theta_0)$ \\
\midrule
500 & -0.82 & -0.53 & 1.63 & 2.85 &  3.14 \\
1000 & -0.86 & -0.66 & 1.39 & 2.34 & 3.14 \\
5000 & -0.92 & -0.85 & 1.08 & 1.51 & 3.12 \\
10000 & -0.92 & -0.89 & 1.06 & 1.30 & 3.12 \\
50000 & -0.93 & -0.92 & 1.05 & 1.05 & 3.12 \\
\midrule
$\infty$ & -0.93 & & 1 & \\
\bottomrule
\end{tabular}
\end{table}

\begin{table}[!htb]
\begin{center}
\caption{Coverage of nominal $95\%$ confidence intervals under model misspecification (brackets: median width of intervals multiplied by $\sqrt{n}$)} \label{tab:simu-dml-misspec}
\begin{tabular}{cccccccccc}
\toprule
\multicolumn{1}{c}{}&&\multicolumn{2}{c}{$L=2$}&&\multicolumn{2}{c}{$L=5$}\\ 
\cline{3-4}\cline{6-7}
$n$ && \cref{eqs:rank-dml-CI} & DML && \cref{eqs:rank-dml-CI} & DML\\ 
\midrule
500 && ~0.994 [11.6] & 0.999 [12.3] && 0.988 [6.1]& 1.000 [12.2] \\
1000 && 0.999 [9.5] & 1.000 [12.3] && 0.975 [5.3]& 1.000 [12.2] \\
5000 && 0.996 [6.0] & 1.000 [12.2] && 0.968 [4.2]& 1.000 [12.2] \\
10000 && 0.990 [5.1] & 1.000 [12.2] && 0.956 [3.9]& 1.000 [12.2] \\
50000 && 0.961 [4.1] & 1.000 [12.2] && 0.941 [3.8]& 1.000 [12.2] \\
\bottomrule
\end{tabular}
\end{center}
\end{table}

\end{example}

\section{On the bootstrap} \label{app:bootstrap}
A curious reader may wonder if the rank transform can be applied to the bootstrap instead of subsampling to achieve the same goal. The short answer is no, if we want to maintain type-I error control under minimal assumptions. In contrast to subsampling, which only requires minimal assumptions (the existence of a non-degenerate limit law) to be consistent, the standard $n$-out-of-$n$ bootstrap also requires the regularity of a ``root'', namely a function of both the test statistic and parameters of the underlying distribution such that its distribution is locally pivotal; see \citet{beran1997diagnosing} and \citet[\S1.6]{politis1999subsampling}. This is considerably stronger than our \cref{cond:pivotal,assump:stable-G} required for pointwise level control, which only concern the behaviour of the test statistic (instead of a root) under each fixed null (instead of neighbourhoods of each null). In fact, for multiple-split, hunt-and-test procedures considered in this paper, it is often unclear whether such a regular root (as an $L$-dimensional vector) exists, and if so, how to construct one.
Consider the test for the null of multivariate
unimodality in \cref{sec:unimodal}, for example. The existence of such a regular root would mean that were data to come from a (slightly) non-unimodal distribution, a certain transformation could be applied to data to restore its unimodality --- such a transformation seems rather hard to construct.
Thus, we use subsampling to avoid these limitations; replacing it with the bootstrap can fail to control the type-I error unless $(T_n^{(1)}, \dots, T_n^{(L)})$ itself is a regular root.

\section{Relation to prepivoting} \label{app:prepivot}
In resampling inference, ranks also come up when describing prepivoting \citep{beran1987prepivoting,beran1988prepivoting}, which refers to the technique of applying the probability integral transform to a statistic using its bootstrap null distribution, i.e., turning a statistic into its bootstrap p-value. Prepivoting is most useful when the asymptotic null distribution of the statistic contains unknown parameters. To achieve improvement for such cases, prepivoting must be iterated twice or more times with the nested bootstrap, before comparing the final transformed pivot to $\unif(0,1)$. Prepivoting reduces the dependency of the sampling distribution of the statistic on the underlying data distribution and can offer higher-order refinements to the bootstrap. We argue that our use of ranks is rather different: (1) while prepivoting improves level control, our rank transform is designed to improve power; (2) the rank transform is applied to the subsample statistics (instead of the test statistic) and applied only once; (3) we use ranks to enforce the marginals of a multivariate statistic but prepivoting is only applicable to a univariate statistic.

\section{Additional numerical results} \label{app:num-extra}
\subsection{Derandomization of \cref{ex:kim-ramdas}} \label{app:non-rep}
\cref{fig:non-rep-KR} shows probability of non-replication (i.e. two applications of the same test on the same data leads to one acceptance and one rejection) for the numerical experiment in \cref{sec:revisit}.

\begin{figure}[!htb]
\centering
\includegraphics[width=0.95\textwidth]{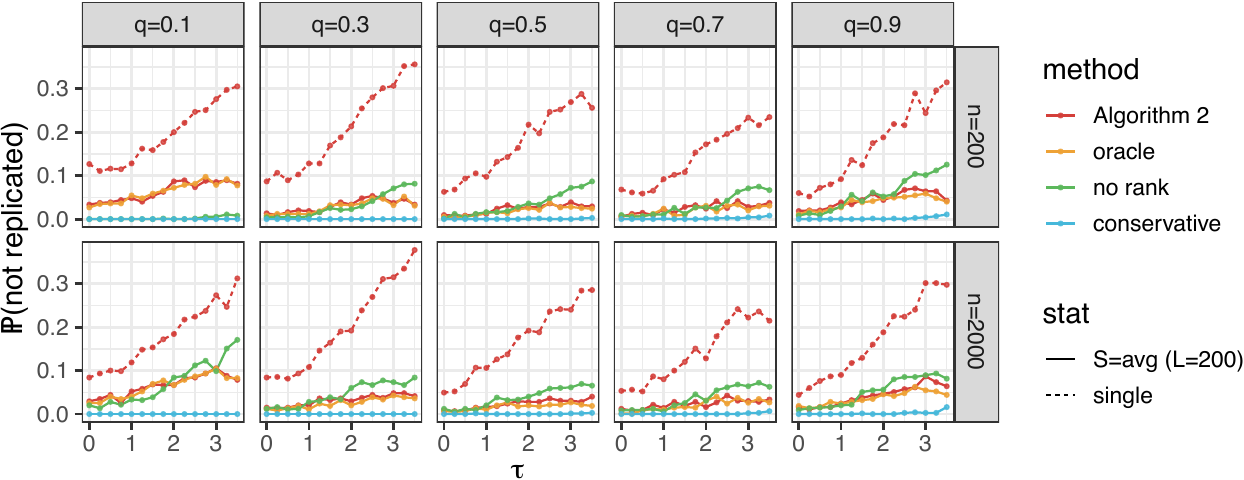}
\caption{Probability of non-replication for testing $\mu = \mathbf{0}$ in \cref{ex:kim-ramdas} at level $\alpha=0.05$, where $\mu = \tau n^{-1/2} v_1$. Tests based on $S_n = (T_n^{(1)} + \dots + T_n^{(L)}) / L$ (solid line) significantly reduces the chance of non-replication from using the single-split test.}
\label{fig:non-rep-KR}
\end{figure}
\subsection{Goodness-of-fit testing for quantile regression} \label{app:quant}
In this section, we include additional numerical results for goodness-of-fit testing of quantile regression considered in \cref{sec:gof}. \cref{fig:quant-nl-2} reports the results when the non-linear function in \cref{eqs:quant-spec} is chosen as $\eta(X) = 2 \exp(-1-X_2-X_3)$.

\begin{figure}[!htb]
\centering
\includegraphics[width=.9\textwidth]{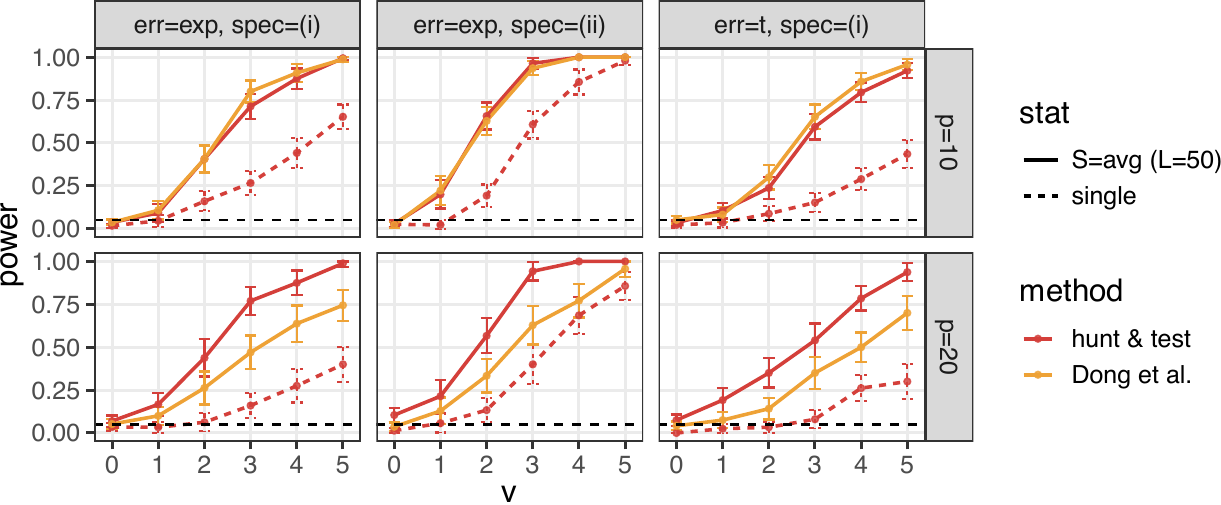}
\caption{Testing goodness-of-fit of a quantile regression model $q_{0.5}(X) = \beta_0 + \beta^{\T} X$: power (95\% CI) at level $\alpha=0.05$ (dashed horizontal) under $n=1000$. The model is well-specified if and only if $v=0$. The non-linear function in \cref{eqs:quant-spec} is $\eta(X) = 2 \exp(-1-X_2-X_3)$.}
\label{fig:quant-nl-2}
\end{figure}
\subsection{Testing generalised conditional independence in a more challenging setting} \label{app:verma}
Consider changing the specification of $Y$ in \cref{sec:verma} to 
\begin{equation} \label{eqs:verma-difficult}
Y = - A_2 + \beta_{H,Y}^{\T} H + (1 - A_1) \epsilon_Y + A_1 \xi_{\nu}, \quad \nu \in (0, \infty],
\end{equation}
where $\xi_{\nu}$ is an independent $t_{\nu}$-distributed random variable. The null hypothesis holds when $\nu = \infty$ and hence $\xi_{\nu} =_d \epsilon_Y$. Because $\cov_Q(A_1,Y) = 0$ under the alternative, using it as the statistic leads to trivial power. Instead, we employ the maximum mean discrepancy (MMD) statistic \citep{gretton2012kernel} between $Y \mid A_1 = 0$ and $Y \mid A_1 = 1$ under $Q$,
with the Gaussian kernel and bandwidth chosen by the median heuristic \citep[\S 8]{gretton2012kernel},
which is able to detect distributional discrepancy beyond the mean.
Under the null hypothesis, the asymptotic distribution of the MMD statistic 
depends on unknown parameters \citep[Theorem 12]{gretton2012kernel} that are very difficult to handle under IPW. Nevertheless, the aggregated post-resampling permutation tests just works out of the box; see \cref{fig:verma-difficult}.

\begin{figure}[!htb]
\centering
\includegraphics[width=.9\textwidth]{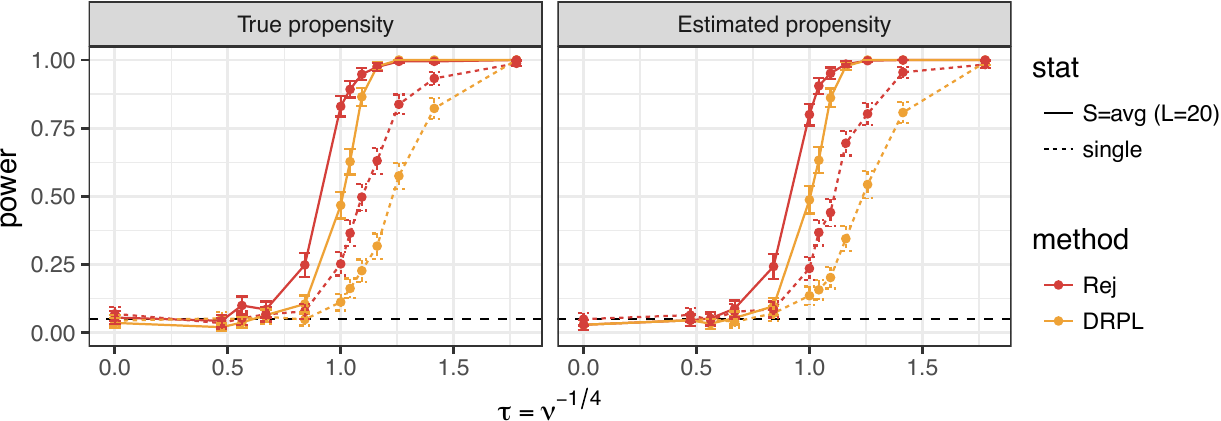}
\caption{Power (95\% CI) for testing no individual direct effect of $A_1$ on $Y$ ($\g$ of \cref{fig:verma}) under the more difficult setting in \cref{app:verma}, for which tests based on $|\cov_Q(A_1,Y)|$ cannot detect the effect. The null hypothesis corresponds to $\tau = 0$. Here `single' refers to post-rejection-sampling (\texttt{Rej}) and post-DRPL (\texttt{DRPL}) permutation test of the MMD statistic; `avg' refers to the corresponding aggregated test based on the average of 20 p-values, calibrated by \cref{alg:agg-test}. Sample size is $n=1000$ and level is 0.05 (dashed horizontal). See also \cref{fig:verma-difficult-vs-conservative} in the appendix.}
\label{fig:verma-difficult}
\end{figure}

\subsection{Comparison with conservative $p$-value merging} \label{app:comp-conservative}
In this section, we include numerical results that compare our method to various conservative, deterministic merging methods for non-independent p-values introduced by \citet{vovk2020combining}. In what follows, \texttt{single} refers to the p-value resulting from a single application of a randomized test (e.g., one single data split). Labels \texttt{avg} and \texttt{min} respectively refer to subsampling based inference for the average p-value and the minimum p-value. Conservatively merged p-values are labelled as \texttt{M.xx} for the following merging functions $p = f(p_1, \dots, p_L)$:
\begin{description}
\item[Arithmetic (\texttt{M.arith})] $p = 2\, (p_1 + \dots + p_L) / L$,
\item[Geometric (\texttt{M.geom})] $p = e\, (p_1 \dots p_L)^{1/L}$,
\item[Bonferroni (\texttt{M.Bonf})] $p = L\, \min(p_1, \dots, p_L)$,
\item[Compound Bonferroni-geometric (\texttt{M.Bonf-geom})] $p = 2 \min \{L\min(p_1, \dots, p_L),   e\, (p_1 \dots p_L)^{1/L} \}$.
\end{description}
These merging functions are precise or asymptotically precise, which roughly mean that the multiplicative constants in their definitions cannot be improved in general; see \citet{vovk2020combining}.

In the below, \cref{fig:unimodal-vs-conservative} shows results for detecting multivariate unimodality considered in \cref{sec:unimodal}; 
\cref{fig:quant-vs-conservative} shows results for goodness-of-fit testing for quantile regression considered in \cref{sec:gof};
\cref{fig:verma-vs-conservative,fig:verma-difficult-vs-conservative} show results for testing generalised conditional independence in the two settings considered in \cref{sec:verma}.

\begin{figure}[!htb]
\centering
\includegraphics[width=.95\textwidth]{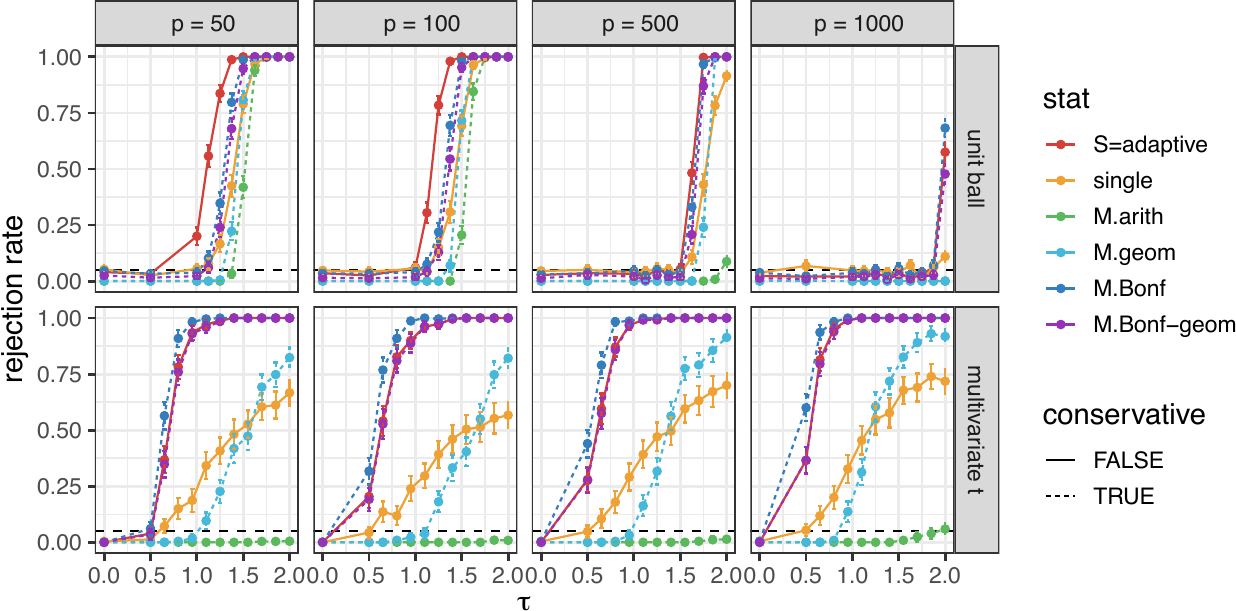}
\caption{Supplement to \cref{fig:unimodal}: comparison of single-split p-value (\texttt{single}), subsampling p-value (\texttt{S=adaptive}) and conservative $p$-value merging functions (\texttt{M.xx}, dashed curves) for testing multivariate unimodality in $p$ dimensions considered in  \cref{sec:unimodal}. The null hypothesis holds when $\tau = 0$ and the level is 0.05. Subsampling and conservative merging are based on $L=50$ exchangeable p-values.}
\label{fig:unimodal-vs-conservative}
\end{figure}

\begin{figure}[!htb]
\centering
\includegraphics[width=1.\textwidth]{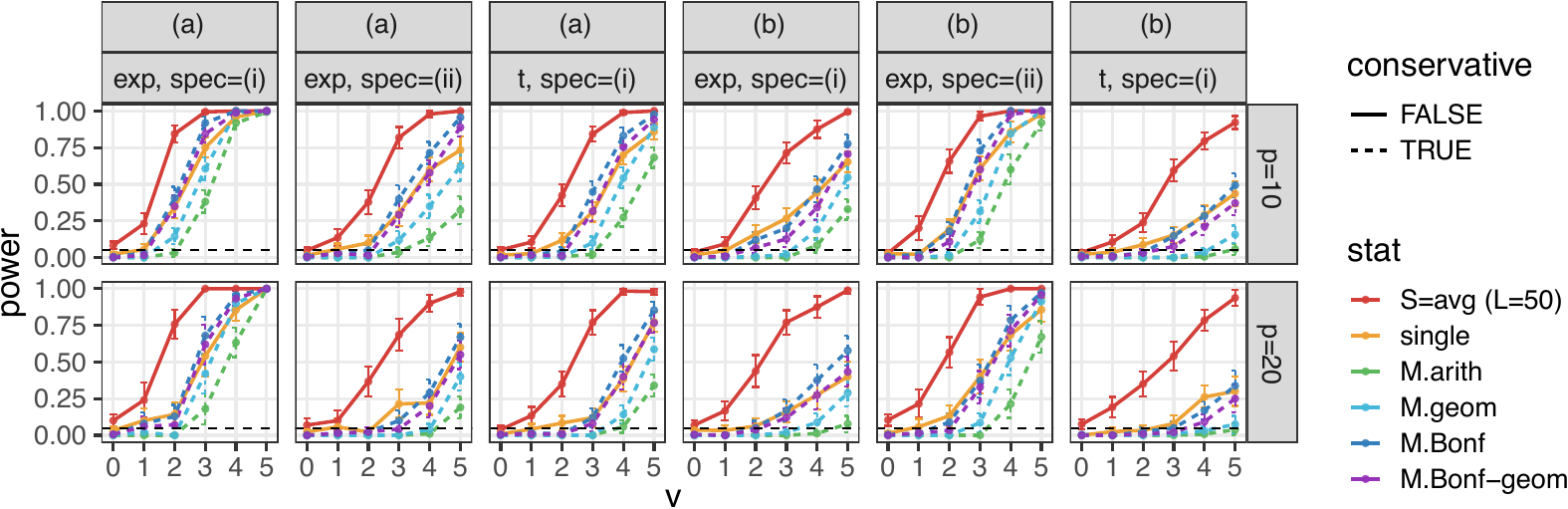}
\caption{Supplement to \cref{fig:quant,fig:quant-nl-2}: comparison of single-split p-value (\texttt{single}), subsampling p-value (\texttt{S=avg}) and conservative $p$-value merging functions (\texttt{M.xx}, dashed curves) for goodness-of-fit testing of quantile regression considered in \cref{sec:gof}. The null hypothesis holds when $v = 0$ and the level is 0.05. Subsampling and conservative merging are based on $L=50$ exchangeable p-values. The non-linear function in \cref{eqs:quant-spec} is: (a) $\eta(X) = 4\sqrt{X_1^2 + X_2^2}$, (b) $\eta(X) = 2 \exp(-1-X_2-X_3)$.}
\label{fig:quant-vs-conservative}
\end{figure}

\begin{figure}[!htb]
\centering
\includegraphics[width=.8\textwidth]{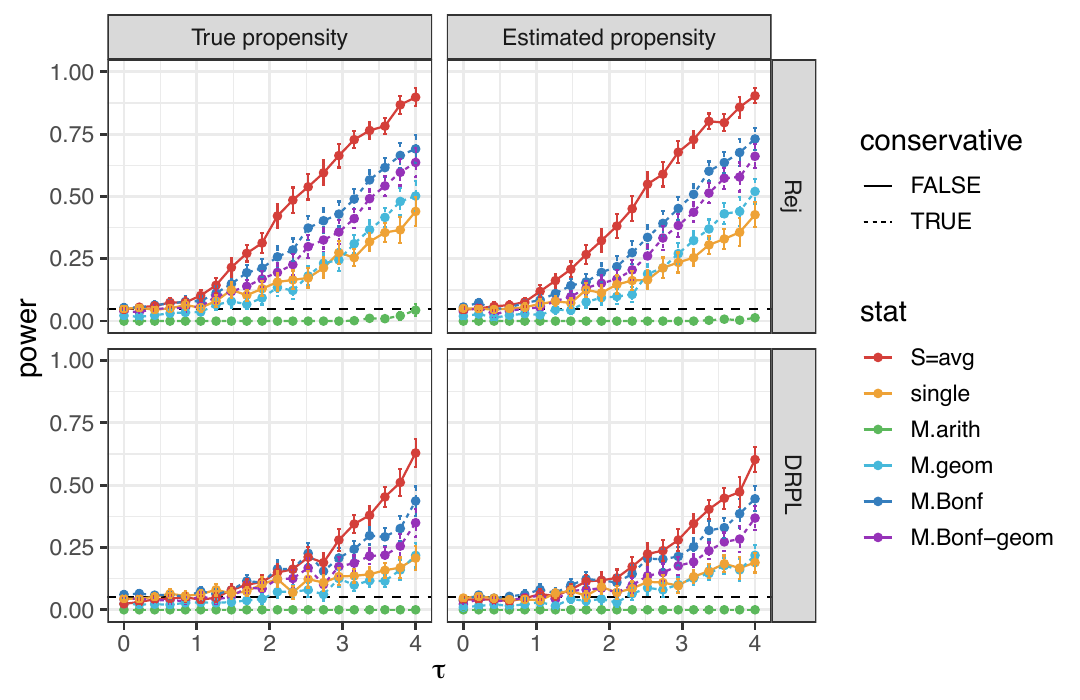}
\caption{Supplement to \cref{fig:simu-verma}: comparison of single-split p-value (\texttt{single}), subsampling averaged p-value (\texttt{avg}) and conservative $p$-value merging functions (\texttt{M.xx}, dashed curves) for testing a generalised conditional independence; see \cref{sec:verma} for the setting. The null hypothesis holds when $\tau = 0$ and the level is 0.05. Subsampling and conservative merging are based on $L=20$ exchangeable p-values.}
\label{fig:verma-vs-conservative}
\end{figure}

\begin{figure}[!htb]
\centering
\includegraphics[width=.8\textwidth]{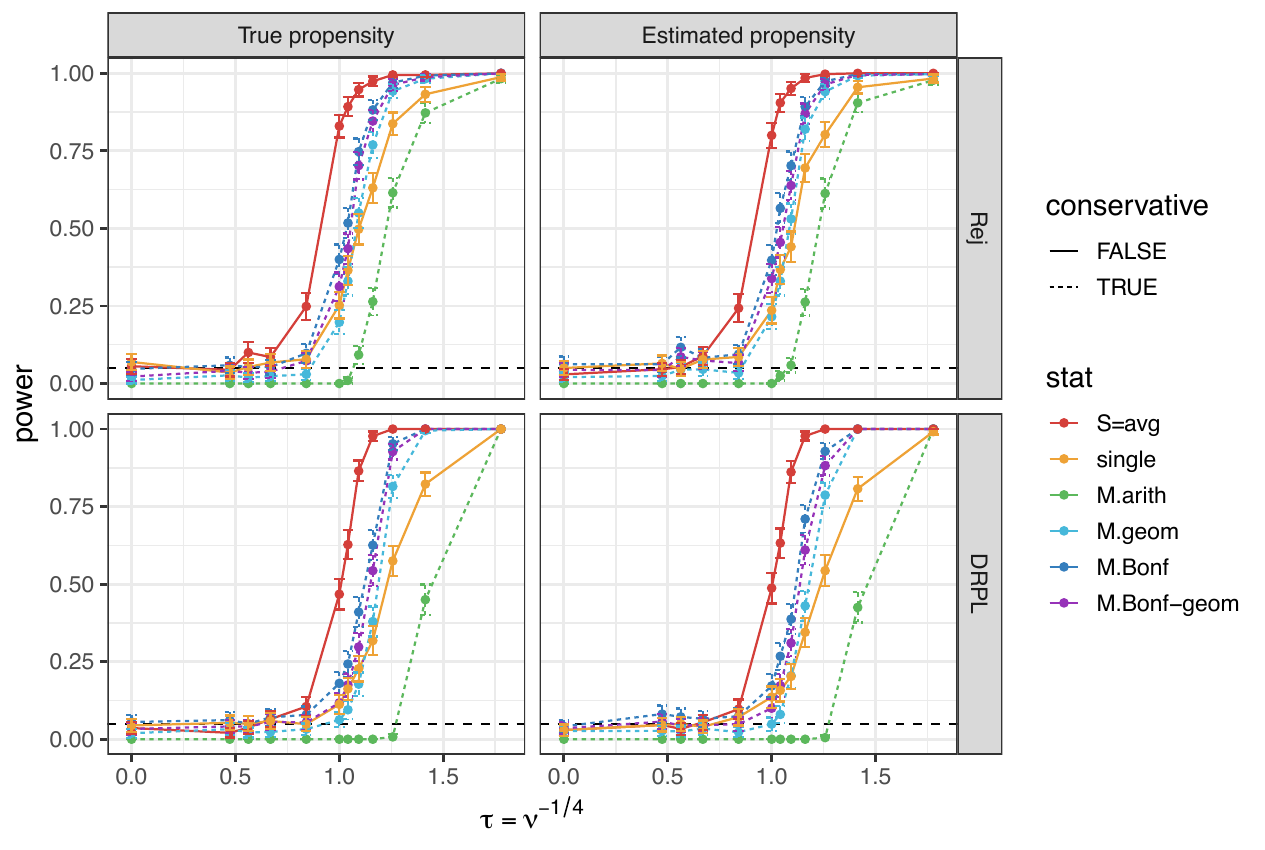}
\caption{Supplement to \cref{fig:verma-difficult}: comparison of single-split p-value (\texttt{single}), subsampling averaged p-value (\texttt{avg}) and conservative $p$-value merging functions (\texttt{M.xx}, dashed curves) for testing a generalised conditional independence under the more difficult setting \cref{eqs:verma-difficult}; see \cref{sec:verma}. The null hypothesis holds when $\tau = 0$ and the level is 0.05. Subsampling and conservative merging are based on $L=20$ exchangeable p-values.}
\label{fig:verma-difficult-vs-conservative}
\end{figure}

\end{document}